\documentclass[a4paper,10pt]{article}
\usepackage{booktabs} 
\usepackage{natbib} 
\usepackage[margin=1in]{geometry}

\setcitestyle{authoryear}

\title{An Algorithm for the Assignment Game Beyond Additive Valuations}

\usepackage{authblk}

\author{Eric Balkanski\thanks{eb3224@columbia.edu}}
\author{Christopher En\thanks{Corresponding Author: ce2456@columbia.edu}}
\author{Yuri Faenza\thanks{yf2414@columbia.edu}}
\affil{Columbia University, IEOR Department}

\usepackage[utf8]{inputenc}
\usepackage{amsmath,amsthm,amssymb}
\usepackage{thmtools, thm-restate}

\usepackage{graphicx}
\usepackage{subcaption}

\usepackage{xcolor}

\usepackage{hyperref}
\hypersetup{colorlinks=true,linkcolor=blue,citecolor=blue,urlcolor=blue}

\usepackage{soul}

\usepackage{color-edits}
\addauthor{eb}{red}
\addauthor{ce}{blue}
\addauthor{yf}{purple}

\usepackage[normalem]{ulem}

\usepackage{algorithm}
\usepackage{algpseudocode}
\algdef{SE}[DOWHILE]{Do}{doWhile}{\algorithmicdo}[1]{\algorithmicwhile\ #1}%

\usepackage{enumerate}

\theoremstyle{plain}
\newtheorem{theorem}{Theorem}
\newtheorem{corollary}[theorem]{Corollary}
\newtheorem{lemma}[theorem]{Lemma}
\newtheorem{proposition}[theorem]{Proposition}

\theoremstyle{definition}
\newtheorem{definition}[theorem]{Definition}
\newtheorem{example}[theorem]{Example}

\theoremstyle{remark}
\newtheorem{remark}[theorem]{Remark}

\renewcommand{\[}{\begin{equation}}
\renewcommand{\]}{\end{equation}}

\newcommand{\bea}{\begin{eqnarray}}
\newcommand{\eea}{\end{eqnarray}}

\newcommand{\argmax}{\mathop{\mathrm{argmax}}}

\begin{document}

\maketitle

\begin{abstract}

The assignment game, introduced by Shapley and Shubik (1971), is a classic model for two-sided matching markets between buyers and sellers. In the original assignment game, it is assumed that payments lead to transferable utility and that buyers have unit-demand valuations for items. Two important and mostly independent lines of work have studied more general settings with imperfectly transferable utility and gross substitutes valuations. Multiple efficient algorithms have been proposed for computing a competitive equilibrium, the standard solution concept in assignment games, in these two settings. Our main result is an efficient algorithm for computing competitive equilibria in a setting with both imperfectly transferable utility and gross substitutes valuations. Our algorithm combines augmenting tree techniques from maximum matching and algorithms for matroid intersection. We also show that, in two mild generalizations of our model, computing a competitive equilibrium is NP-hard.

\end{abstract}

\section{Introduction}

The assignment game, introduced over half a century ago by \citet{shapley1971assignment}, is a standard model to study two-sided markets between buyers and sellers. One of the main solution concepts in this game, the competitive equilibrium, is an efficient outcome composed of allocations and prices where the buyers are allocated a set of items that maximizes their utility at the prices specified by the outcome. Analyzing and computing competitive equilibria in the assignment game has applications in housing markets between buyers and sellers \citep{shapley1971assignment}, hiring markets between employers and candidates \citep{gul2000english, dupuy2020taxation}, and online advertising markets between advertisers and publishers \citep{dutting2011expressive}.

In the original model of the assignment game, each seller sells a single copy of an item (unit supply), and each buyer is interested in acquiring a single item (unit demand). Buyers and sellers are assumed to have \emph{transferable utility} (TU), meaning that payments are used to share surplus between agents without losing utility to frictions. \citet{shapley1971assignment} showed that competitive equilibria always exist by casting the problem as a linear program. This existential result thus also serves as an efficient algorithm for computing the entire set of competitive equilibria. There has since been substantial work studying various extensions of the assignment game.

A first body of work has relaxed the transferable utility assumption and instead considers settings with \emph{imperfectly transferable utility} (ITU). This relaxation is motivated by the various impediments to the transfer of utility between agents, such as nonlinear taxes or fees a buyer may face when making a trade. \citet{kaneko1982central} and \citet{quinzii1984core} gave nonconstructive proofs that competitive equilibria exist when buyers have unit demands, even under ITU, and \citet{demange1985strategy} described various structural properties of the set of competitive equilibria under unit demands and ITU. Unlike in the TU setting, the initial existential and structural results under ITU do not imply efficient algorithms for computing competitive equilibria. A subsequent line of work has thus focused on the computational side of competitive equilibria under ITU. In particular, \citet{alkan1990core} provided an algorithm for computing a competitive equilibrium in this setting, though they did not provide a polynomial bound on the running time. \citet{alkan1992equilibrium} later improved the algorithm and showed that the new version runs in polynomial time.

Another line of work that generalizes the initial model of \citet{shapley1971assignment} has relaxed the unit demand assumption for buyers in the market by considering richer classes of valuations. \citet{crawford1981job} extend the assignment game to a many-to-one model where each buyer may match with multiple sellers, and prove the existence of competitive equilibria. In their model, buyers are assumed to have \emph{additive} valuations, and to demand a fixed number of matches. \citet{kelso1982job} consider an even broader class of valuations by introducing \emph{gross substitutes} valuations as a sufficient condition for the existence of competitive equilibria when buyers have valuations over substitutable items. Since this introduction, gross substitutability has proven to be a very general yet powerful condition on valuations that captures many desirable properties. On the computational side, \citet{kelso1982job} show that under gross substitutes valuations, a natural ascending auction procedure called \emph{Walrasian t\^atonnement} is guaranteed to find a competitive equilibrium. While this process is primarily used as a proof device in \citet{kelso1982job}, it was later adapted to run in polynomial time (see, e.g., \citet{leme2017gross}). \citet{murota1996valuatedi} and \citet{murota1996valuatedii} present the first purely combinatorial cycle-cancelling algorithm for the same problem that runs in strongly polynomial time.
\smallskip

\noindent {\bf Applications.} Even though gross substitutes and imperfectly transferable utilities have been studied mostly independently, there are many application areas where markets exhibit both item substitutabilities and nonlinear frictions, taxes, or fees. Here we discuss one common setting; we present additional applications in Appendix~\ref{sec:proofs:applications}.

\paragraph{Hiring markets.} In hiring markets, workers often exhibit substitutable skills. As a result, the valuation functions of  firms for workers are typically not additive. For example, a technology firm would often value two software engineers less than the sum of their value for each individual software engineer because of a significant overlap in the tasks that the two engineers are qualified to complete. As described in the classic hiring market model of \citet{kelso1982job}, this phenomenon of ``diminishing marginal returns'' can be represented by gross substitutes valuation functions. In addition, each firm must pay its workers a salary, and each worker's salary is subject to progressive  income taxes that assign a tax rate to different tax brackets. Furthermore, if a worker is a remote worker located in another state, their salary is subject to increased tax rates. These  tax rates that depend on tax brackets imply that the post-tax salary received by a worker is a piecewise linear function of the pre-tax salary paid by the firm. This gap between the pre-tax and post-tax salaries causes an imperfect transfer of utility that has been modeled using ITU (see, e.g., \citet{dupuy2020taxation, galichon2019costly}). 

The scenarios discussed here and in Appendix~\ref{sec:proofs:applications} have motivated the study of generalized assignment games that feature both ITU and gross substitutes valuations. As we discuss in more detail in Section \ref{sec:intro:related-works}, even though some existential and structural results have been shown in similar or more general settings, there is, to the best of our knowledge, no existing computational result in such a setting. The combination of nonlinearity in buyers' valuations (gross substitutes) and nonlinearity of their loss in utility due to frictions (ITU) introduces novel algorithmic challenges that are further discussed in Section~\ref{sec:intro:technical-overview}.

\smallskip

\noindent {\bf Our Contributions.} Our main contribution is a polynomial-time algorithm for computing a competitive equilibrium under imperfectly transferable utilities and gross substitutes valuations. Our model for ITU, which we call the \textsc{QITU} model (short for \emph{Quasilinear and Imperfectly Transferable Utility} model, and described formally in Section~\ref{sec:prelim:problem-statement}) assumes gross substitutes valuations and generalizes the gross substitutes TU models of \citet{kelso1982job} and the unit-demand ITU models of \citet{demange1985strategy}, \citet{alkan1990core}, and \citet{alkan1992equilibrium}, under the additional assumption that transfers are piecewise linear. Our model also generalizes the model of \cite{dutting2011expressive} with the additional assumptions that transfers are continuous and surjective.\footnote{\cite{dutting2011expressive} allow transfers to be noncontinuous, and simply assume that the transfer functions are sufficiently positive instead of surjective, as they do not account for negative prices.}
\begin{theorem}\label{thm:main}
    In the \textsc{QITU} model, there is a polynomial time algorithm that computes a competitive equilibrium.
\end{theorem}
As in previous models, we take as input oracles for the valuation functions of each buyer, along with the full description of functions defining the rate at which utility can be transferred between agents. In addition to achieving the first efficient algorithm under ITU and gross substitutes valuations, our result also unifies many of the existing algorithmic results in the previously listed models that are special cases of the \textsc{QITU} model. In particular, we show the the bidder-optimality of the output of our algorithm, see Theorem~\ref{thm:min_ce}. As such, our algorithm generalizes the ascending auction style mechanisms in various related settings.

We next show that the assumption that buyers have quasilinear utilities is required for our result by considering two relaxations of our setting, called the \textsc{NQBu} model (short for \emph{Non-Quasilinear with Budget}) and \textsc{NQBr} model (short for \emph{Non-Quasilinear with Breakpoint}). In the \textsc{NQBu} model, each buyer has gross substitutes valuations and perfectly transferable utility, but a single buyer also has an additional budget constraint. In the \textsc{NQBr} model, each buyer has gross substitutes valuations and perfectly transferable utility, except for a single buyer whose utility function is piecewise linear in money, with a single breakpoint. These models are of interest as a way to model budgets or risk aversion in buyers, which are common properties in many real-world markets. However, we prove negative results showing that computing a competitive equilibrium in these settings is NP-hard.

\begin{theorem}\label{thm:hardness2}
    Computing a competitive equilibrium in the \textsc{NQBu} model is NP-hard, even when an equilibrium is known to exist.
\end{theorem}

\begin{theorem}\label{thm:hardness}
    Computing a competitive equilibrium in the \textsc{NQBr} model is NP-hard, even when an equilibrium is known to exist.
\end{theorem}

The paper is organized as follows. We conclude the introductory section with an overview of the techniques used to prove Theorem~\ref{thm:main} (Section~\ref{sec:intro:technical-overview}) and a discussion of additional related work (Section~\ref{sec:intro:related-works}). In Section \ref{sec:preliminaries}, we formally introduce the \textsc{QITU} model and state the problem of interest. In Section \ref{sec:framework}, we introduce the basic structures on which the algorithm from Theorem~\ref{thm:main} operates and formally describe the algorithm. In Section \ref{sec:FPI}, we introduce and analyze the main subroutine of the algorithm, used to increase prices in the spirit of Walrasian t\^atonnement. This section is the most technical section and the main technical contribution of the paper. In Section \ref{sec:main_result}, we analyze the remaining subroutines and combine our results into the proof for Theorem \ref{thm:main}. In Section~\ref{sec:lattice_IC}, we explore the following structural and algorithmic features:   lattice structures, bidder-optimality, and incentive-compatibility. In Section \ref{sec:negative}, we present Theorems \ref{thm:hardness2} and \ref{thm:hardness} Appendix~\ref{sec:proofs:applications} and Appendix~\ref{sec:proofs:prelims} present extended discussion on applications and the model, respectively. Appendix \ref{sec:proofs:gs} contains helper results on gross substitutes valuations used in the other sections. Missing proofs from Sections \ref{sec:framework}, \ref{sec:FPI}, \ref{sec:main_result}, \ref{sec:lattice_IC},
 and \ref{sec:negative} appear in Appendix \ref{sec:proofs:ce_cert}, \ref{sec:proofs:FPI}, \ref{sec:proofs:main_result}, \ref{sec:proofs:lattice_IC},
 and \ref{sec:proofs:negative}, respectively. Section~\ref{sec:r-gross} contains an analysis of the limits of extending the concept of gross substitutes.

\subsection{Technical overview}\label{sec:intro:technical-overview}

Similar to auction mechanisms for the maximum weight bipartite matching problem (e.g.,~\citep{demange1986multi}) and for generalizations of the assignment game including \citet{alkan1992equilibrium} and \citet{dutting2011expressive}, our algorithm maintains a matching $\mu$ and a set of item prices $p$ at all stages, and alternates between two subroutines: a search for a $\mu$-augmenting path and a price increase.

The invariant maintained throughout the algorithm by the pair $(\mu,p)$ is \emph{partial stability}, a new concept that can be seen as a relaxed version of the equilibrium property. The algorithm operates throughout on a structure we call the \emph{marginal demand graph}, which indicates which buyers demand which items at the current outcome $(\mu, p)$. The augmenting path subroutine is employed to increase the size of $\mu$, while preserving partial stability. When no partially stable matching larger than $\mu$ can be found at the current prices $p$, the price increase kicks in. We employ a modified version of Walrasian t\^atonnement on the marginal demand graph, which finds a price increase that preserves partial stability. We further show that the alternating of these two procedures converges to an equilibrium. 

The key difficulty, and the main difference from prior work by \citet{alkan1992equilibrium} and \citet{dutting2011expressive}, comes from handling gross substitutes valuations. Unlike in their simpler unit-demand models, these valuations mean that a buyer's marginal value for an item changes based on the other items they hold. This dynamic creates a major challenge: every time the allocation is adjusted, buyers' marginal values shift, which in turn alters the sets of items they would be willing to trade (their indifference sets). As a result, the marginal demand graph we use to find better allocations is constantly changing. A standard ``augmenting path'' in such a graph could actually harm a buyer. Our solution is to leverage the special matroidal properties of gross substitutes demand sets, allowing us to find augmenting paths that improve the overall matching without reducing any buyer's utility.

Furthermore, when calculating price increases, \citet{dutting2011expressive} use a linear programming duality trick seen in \citet{alkan1989existence} to compute the price increase via a maximum weight matching problem. The nonadditive nature of our setting means this trick no longer applies directly, as the entire family of demanded sets for any given buyer is not easily expressible within a static graph or a single linear program. Instead, we can leverage properties of gross substitutes valuations to solve a related weighted common matroid basis problem and then use the solution to define an appropriate linear program. From there, we can reapply the duality trick to compute the desired price increase.

\subsection{Further related work}\label{sec:intro:related-works}

A large body of work extends the models of \citet{shapley1971assignment} and \citet{demange1985strategy}. \citet{rochford1984symmetrically} presents an alternative solution concept for the unit demand TU model based on ``threat utility'' and shows that these solutions also always exist within the set of competitive equilibria. \citet{roth1988interior} show that the same results hold in the setting with unit demand and ITU. \citet{hatfield2005matching} present a very general setting encompassing gross substitutes valuations. They show that the set of competitive equilibria is nonempty and forms a finite lattice, and present an algorithm for computing competitive equilibria. However, because of the discretization of prices, their model is incomparable to ours. \citet{hatfield2013stability} present a more complex model that removes the two-sided assumption. They identify \emph{full substitutability}, a generalization of gross substitutability, as a sufficient condition for the existence of competitive equilibria. \citet{fleiner2019trading} then develop a similar one-sided matching model that incorporates ITU, encompassing our model. They show that equilibria are guaranteed to exist under full substitutability; however, their results rely on fixed-point theorems and are thus intrinsically inefficient.

There has also been substantial work studying gross substitutes, as this class of functions naturally appears in many other contexts. \citet{nisan2006communication} provide another competitive equilibrium algorithm for the gross substitutes TU model of \citet{kelso1982job} based on linear programming. \citet{dress1990valuated}, \citet{dress1995rewarding}, \citet{murota1996convexity} study various functions related to matroids and discrete convexity, and \citet{fujishige2003note} show that their concepts are all equivalent to gross substitutes.

Algorithmically, there has been less work beyond the basic ITU model of \citet{demange1985strategy} and the gross substitutes model of \citet{kelso1982job}. \citet{colini2015multiple} develop a polynomial algorithm for a setting inspired by online ad auctions where buyers demand multiple items, but valuations for items are correlated across buyers, and utility is perfectly transferable except for a single discontinuity used to model a budget constraint. \citet{goel2015polyhedral} consider a related setting where items have public valuations and allocations must satisfy certain polyhedral or polymatroidal constraints. \citet{dutting2015auctions} provide a polynomial algorithm for a setting with private budgets, where values are additive and correlated across buyers.

\section{Preliminaries}\label{sec:preliminaries}

\subsection{Problem statement}\label{sec:prelim:problem-statement}

In the assignment game, there are finite sets of buyers $B$ and items $S$. Each buyer $i\in B$ has a \emph{valuation function} $v_i: 2^S\to \mathbb R$, where $v_i(T)$ is their value for receiving items $T \subseteq S$, and a \emph{utility function} $u_i: S\times \mathbb R^S\to\mathbb R$, where $u_i(T, p)$ is their utility for receiving items $T$ when the  prices are $p \in \mathbb R^S$. In the original model of \citet{shapley1971assignment}, the valuation functions are assumed to satisfy unit-demand, i.e., $v_i(T) = \max_{j\in T}v_i(\{j\})$, and the utility functions, defined as $\mathring u_i(T, p) := v_i(T) - \sum_{j\in T}p_j$, exhibit transferable utility.

\citet{kelso1982job} generalized the assignment game by allowing more complex valuations than unit-demand. In their model, the \emph{demand correspondence} $\mathring F_p(i) := \{T \subseteq S\mid \mathring u_i(T, p)\ge \mathring u_i(T', p) \text{ for all } T'\subseteq S\}$ of buyer $i$ at prices $p$ is the family of sets  of items $T$ that maximize  buyer $i$'s utility at prices $p$. Then, valuations are assumed to satisfy \emph{gross substitutes}:

\begin{definition}[Gross Substitutes \citep{kelso1982job}]\label{def:gs_kc} A valuation function
    $v_i$ satisfies gross substitutes if for any price vectors $p, p'\in \mathbb R^S$ such that $p\le p'$, and $T\in \mathring F_p(i)$, there is a set $T'\in \mathring F_{p'}(i)$ such that $T\cap\{j\mid p_j=p_j'\}\subseteq T'$.
\end{definition}

The second generalization is the assignment game with \emph{imperfectly transferable utility} (ITU),  studied by \citet{alkan1989existence} and \citet{dutting2011expressive}. For each buyer-item pair $(i, j)$, there is an \emph{effective price} function $q_{ij}:\mathbb R\to\mathbb R$ where $q_{ij}(p_j)$ represents the ``true'' price paid by agent $i$ for an item $j$ with ``listed'' price $p_j$. ITU is the case where the effective price functions need not be the identity. In particular, $q_{ij}$ is assumed to be piecewise linear, strictly increasing in $p_j$, and $q_{ij}(0) = 0$. The right derivative of the effective price function is denoted $q'_{ij}$ and it is assumed to be defined at every point of the domain of $q_{ij}$. The utility functions are then defined as \begin{equation}\label{eq:utility}u_i(T, p, q) := v_i(T) - \sum_{j\in T}q_{ij}(p_j),\end{equation} where as in the assignment game, valuations $v_i:2^S \rightarrow \mathbb{R}$ are unit-demand.

Now consider the following model incorporating both gross substitute valuations and imperfectly transferable utility. Each buyer $i$ has a valuation function $v_i:2^S\to \mathbb R$ and an effective price function $q_{ij}:\mathbb R\to\mathbb R$ for each $j\in S$, so that the utility function $u_i(T,p,q)$ is defined as in~\eqref{eq:utility}, and the demand correspondence is given by $F_{p,q}(i):= \{T\in S\mid u_i(T, p, q)\ge u_i(T', p, q),\ \forall T'\subseteq S\}$. The main model we study is the QITU model, where the following restrictions are imposed. Effective price functions $q_{ij}$ are assumed to be piecewise linear, continuous, strictly increasing, surjective onto $\mathbb R$, and such that $q_{ij}(0) = 0$. The valuation functions satisfy \emph{ITU-gross substitutes}:

\begin{definition}[ITU-Gross Substitutes]\label{def:defn_gs}
    A valuation function $v_i$ satisfies ITU-gross substitutes if for any effective price functions $q = \{q_{i,j}\}_{i\in B, j \in S}$ that are  continuous, strictly increasing, and surjective onto $\mathbb{R}$, price vectors $p, p'\in \mathbb R^S$ such that $p\le p'$, and $T\in F_{p,q}(i)$, there is a set $T'\in F_{p', q}(i)$ such that $T\cap\{j\mid p_j=p_j'\}\subseteq T'$. \end{definition}

We can show that ITU-gross substitutes is equivalent to gross substitutes. Thus, going forward we refer to all the $v_i$ simply as gross substitutes.

\begin{lemma}\label{lem:gs_equals_gskc}
    A valuation function $v_i$ satisfies gross substitutes if and only if it satisfies ITU-gross substitutes. 
\end{lemma}

The proof is deferred to Appendix~\ref{sec:proofs:prelims}. Note that the utility functions in the QITU model are additively separable over the effective prices $q_{ij}(p_j)$. This assumption is in line with the motivating examples described in the introduction and Appendix~\ref{sec:proofs:applications}, such as hiring markets and advertising markets. For instance, in the former case, employee income taxes are independent of other employee salaries, or the total amount of salary paid by an employer. We also show in Theorem~\ref{thm:hardness} that nonseparable utilities lead to NP-hardness.

We note that the structural equivalence of ITU-gross substitutes and standard gross substitutes relies only on the effective price functions being continuous, strictly increasing, and surjective. However, the additional assumption of piecewise linearity within the QITU model is strictly necessary for our computational results. Specifically, the finite convergence of our price-adjustment algorithm relies on the slopes of the effective price functions remaining constant within discrete intervals, allowing us to compute valid step sizes for price increases.

In the QITU model, a \emph{one-to-many matching} $\mu\subseteq B\times S$ is a subset of buyer-item pairs such that each item $j$ is matched to at most one buyer. We also write $\mu(i)$ to denote the set of items $i$ is matched to, and $\mu(j)$ as the buyer to which $j$ is matched. An \emph{outcome} is a pair $(\mu, p)$ with a matching $\mu$ and a vector of prices $p = (p_1,\ldots, p_{|S|})$ for the items. An outcome is \emph{feasible} if all unmatched items have price $0$. The primary solution concept is the \emph{competitive equilibrium}, which is a feasible outcome such that each buyer is matched to a set of items that maximizes their utility subject to the prices specified by the outcome.

\begin{definition}[Competitive equilibrium]\label{def:ce}
    An outcome $(\mu, p)$ with $p\ge 0$ is a \emph{competitive equilibrium} if the outcome is feasible, and for every buyer $i$, we have $\mu(i)\in F_{p,q}(i)$.
\end{definition}

The competitive equilibrium solution concept was also the focus of previous models by \cite{shapley1971assignment}, \cite{kelso1982job}, \cite{alkan1992equilibrium}, and \cite{dutting2011expressive}. Our goal is to compute a competitive equilibrium in polynomial time. Unless otherwise stated, we assume from now on that valuations, effective prices, and utility functions satisfy the assumptions of the QITU model. Moreover, when the effective price functions are clear from context, we abuse notation and write $u_i(T, p)$ and $F_p(i)$ instead of $u_i(T, p, q)$ and $F_{p, q}(i)$.

\subsection{Gross substitutes and stability}\label{subsec:gross_subs_stability}

The \emph{marginal value} of a set of items $T'$ given $T$ is defined as $v_i(T'\mid T) := v_i(T\cup T') - v_i(T)$. Similarly, the \emph{marginal utility} of $T'$ given $T$ at prices $p$ is defined as $u_i(T'\mid T, p) := u_i(T\cup T', p) - u_i(T, p)$.

To simplify definitions and proofs, we assume the existence of $|B||S|$ dummy items, with the property that for any buyer $i$, any dummy item $j_0$, and any bundle of items $T$, the marginal value of the dummy item is $v_i(j_0\mid T) = 0$. Notice that this assumption is without loss of generality since, at any competitive equilibrium $(\mu, p)$, if dummy item $j_0$ is matched to buyer $i$, then $u_i(\mu(i), p) \ge u_i(\mu(i)\setminus\{j_0\}, p)$, which means $p_{j_0} = 0$. Thus, $u_i(j_0\mid T, p) = 0$ for all buyers $i$ and all bundles $T$. Essentially, dummy items are always worth nothing and cost nothing: from a competitive equilibrium in the instance with dummy items, we can obtain a competitive equilibrium for the original instance by removing dummy items from all buyers' bundles.

We can show that, under gross substitutes valuations, competitive equilibria coincide with the notion of (pairwise) stability from matching theory, because a buyer has their most preferred bundle if and only if there are no single improvements. This result is well-known (see, e.g., \citet{fleiner2019trading}), but we include a simple proof in Appendix~\ref{sec:proofs:gs} for completeness.

\begin{definition}[Stability]
     An outcome $(\mu, p)$ is \emph{stable} if every buyer $i$ has no profitable \emph{additions}, \emph{swaps}, or \emph{drops}: 
    \begin{itemize}
        \item \textbf{Additions}: For every item $j\not\in \mu(i)$, we have $u_i(j\mid \mu(i),p)\le 0$.
        \item \textbf{Swaps}: For every $j_1\in\mu(i)$ and $j_2\not\in\mu(i)$, we have $u_i(\mu(i), p)\ge u_i(\mu(i)\cup\{j_2\}\setminus\{j_1\}, p)$.
        \item \textbf{Drops}: For every $j\in \mu(i)$, we have $u_i(\mu(i)\setminus\{j\}, p)\le u_i(\mu(i), p)$.
    \end{itemize}
\end{definition}

\begin{lemma}[Competitive equilibrium coincides with feasible and stable]\label{lem:ce_fsbl_stable}
    An outcome $(\mu, p)$ with $p\ge 0$ is a competitive equilibrium if and only if it is feasible and stable.
\end{lemma}

\subsection{Matroids}

A key component of our algorithm and its analysis is the concept of matroid.
\begin{definition}[Matroid]
    A matroid $M$ is a pair $(E,\mathcal I)$, where $E$ is a finite \emph{ground set} and $\mathcal I$ is a family of subsets of $E$ known as \emph{independent sets}, that satisfies:
    \begin{enumerate}
        \item $\emptyset\in \mathcal I$.
        \item If $A_1\in \mathcal I$ and $A_2\subseteq A_1$, then $A_2\in \mathcal I$.
        \item If $A_1,A_2\in\mathcal I$ and $|A_1| > |A_2|$, then there exists $x\in A_1\setminus A_2$ such that $A_2\cup \{x\}\in\mathcal I$.
    \end{enumerate}
    A \emph{basis} $T\in\mathcal I$ of a matroid is an inclusionwise maximal independent set.
\end{definition}

We will primarily make use of the well known basis exchange property (see, e.g., \cite{oxley2006matroid}).

\begin{lemma}[Strong basis exchange property]\label{lem:basis_exchange}
    Let $E$ be a finite ground set, and $\mathcal F$ be a nonempty family of subsets of $E$ of uniform cardinality. Then, $\mathcal F$ is the set of bases of a matroid if and only if for any $T_1,T_2\in \mathcal F$ and $x\in T_1\setminus T_2$, there exists $y\in T_2\setminus T_1$ such that $T_1\cup\{y\}\setminus \{x\}, T_2\cup\{x\}\setminus \{y\}\in \mathcal F$.
\end{lemma}

We will take advantage of existing algorithms for solving minimum weight matroid intersection problems. These algorithms require access to independence oracles, defined below.

\begin{definition}[Independence oracle]
    Given a matroid $M = (E,\mathcal I)$, an independence oracle for $M$ is a subroutine which indicates whether a set $F\subseteq E$ is a member of $\mathcal I$.
\end{definition}

\section{The Algorithmic Framework}\label{sec:framework}

In this section, we first introduce the main structures necessary to present the algorithm. In Section~\ref{sec:framework:mdg}, we discuss the \emph{marginal demand graph}, which is the basic graph on which the algorithm operates. In Section~\ref{sec:framework:MAT}, we introduce \emph{maximal alternating trees}, which are the fundamental structures used in the algorithm. In Section~\ref{sec:framework:alg}, we describe the main algorithm in detail.

To describe and analyze our algorithm, it is helpful to think of a buyer $i$ as a collection of $|S|$ identical unit-demand buyers, or \emph{unit-buyers}, who can be matched to at most one item. We define $\text{ubuy}(i)$ to be the set of unit-buyers corresponding to $i$, and we define the set $U$ of all the unit-buyers as $U = \bigcup_{i\in B} \text{ubuy}(i)$. Given a unit-buyer $k$, we let $\text{buy}(k)$ be the buyer such that $k \in \text{ubuy}(\text{buy}(k))$. Let $\text{copy}(k) = \text{ubuy}(\text{buy}(k))$ be the set of all unit-buyers who share the same original buyer as $k$. When it is clear from context, given a unit-buyer $k$, we write $u_k := u_{\text{buy}(k)}$ to refer to the utility function of the buyer associated with $k$, and $q_{kj} := q_{\text{buy}(k)j}$ to refer to their effective price function. We refer to this setting as the \emph{one-to-one market} and to the original as the \emph{one-to-many market}.

\subsection{Marginal demand graphs}\label{sec:framework:mdg}

A \emph{one-to-one outcome} $(\tau, p)$ is a matching $\tau\subseteq U\times S$ and a set of prices $p$ such that each unit-buyer and item is matched at most once. An outcome $(\mu, p)$ of the one-to-many market is said to be the \emph{one-to-many projection of the one-to-one outcome $(\tau, p)$} if $(i,j)\in\mu$ if and only if $(k,j)\in \tau$ for some $k\in \text{ubuy}(i)$. We define a one-to-one outcome $(\tau, p)$ to be stable and feasible if its one-to-many projection is stable and feasible. Similarly, we define a one-to-one outcome to be a competitive equilibrium if its one-to-many projection is a competitive equilibrium.

\begin{definition}[Partial stability]\label{def:partial_stability}
    An outcome $(\mu, p)$ in the one-to-many market is \emph{partially stable} if, for all buyers $i$ and all sets of items $T$ with $|T|\le |\mu(i)|$, we have $u_i(T, p)\le u_i(\mu(i), p)$. A one-to-one outcome $(\tau, p)$ is partially stable if its one-to-many projection is partially stable. 
\end{definition}

When the price vector $p$ is clear from context, we may also refer to the matching $\mu$ as partially stable. Partial stability means that a buyer's current matched bundle is their most preferred bundle of the current size (or smaller). Partial stability can also be interpreted as ``no profitable swaps or drops.'' The proof is deferred to Appendix \ref{sec:proofs:ce_cert}.

\begin{lemma}[Partial stability relaxes stability]\label{lem:partial_stability_alt_defn}
    An outcome $(\mu, p)$ is partially stable if and only if every buyer $i$ has no profitable drops or swaps:
    \begin{itemize}
        \item \textbf{Swaps}: For every $j_1\in\mu(i)$ and $j_2\not\in\mu(i)$, we have $u_i(\mu(i), p)\ge u_i(\mu(i)\cup\{j_2\}\setminus\{j_1\}, p)$.
        \item \textbf{Drops}: For every $j\in \mu(i)$, we have $u_i(\mu(i)\setminus\{j\}, p)\le u_i(\mu(i), p)$.
    \end{itemize}
\end{lemma}

Given a partially stable one-to-one outcome $(\tau, p)$, define the \emph{marginal demand correspondence} of a unit-demand buyer $k$ as the set of items $F_{\tau, p}(k)\subseteq S$ where $j_1\in F_{\tau, p}(k)$ if
\begin{itemize}
    \item \textbf{Matched}: $j_1 = \tau(k)$, or
    \item \textbf{Potential alternative}: $j_1\not\in \tau(\text{copy}(k))$, $k$ is matched to an item other than $j_1$, and $u_{k}(j_1\mid \tau(\text{copy}(k)) \setminus \tau(k), p) = 0$, or
    
    \item \textbf{Favorite item}: $j_1\not\in \tau(\text{copy}(k))$, $k$ is unmatched in $\tau$, and $u_{k}(j_1\mid \tau(\text{copy}(k)), p)\ge \max_{j_2\not\in \tau(\text{copy}(k))} u_{k}(j_2\mid \tau(\text{copy}(k)), p)$. Note that $u_{k}(j_1\mid \tau(\text{copy}(k)), p)\ge 0$ due to dummy items.

\end{itemize}
Essentially, if $k$ is matched in $\tau$, $F_{\tau, p}(k)$ contains the item $\tau(k)$, as well as any other item $j_1$ not matched to a copy of $k$ such that $j_1$ can be swapped with an item currently matched to a copy of $\text{buy}(k)$ without changing the utility of $\text{buy}(k)$. If $k$ is unmatched in $\tau$, then $F_{\tau, p}(k)$ contains each unmatched item whose addition to $\tau(\text{copy}(k))$ would lead to the largest utility increase, which is always nonnegative due to dummy items. The marginal demand correspondence of an item $j$ is given by $F^{-1}_{\tau, p}(j):=\{k\mid j\in  F_{\tau, p}(k)\}$. The marginal demand correspondence for a set of unit-buyers $U'\subseteq U$ is given by $F_{\tau, p}(U') := \bigcup_{k\in U'}F_{\tau, p}(k)$. Similarly, for a set of items $S'\subseteq S$, we define $F_{\tau, p}^{-1}(S') := \bigcup_{j\in S'}F_{\tau, p}^{-1}(j)$.
Next, we define the main graph upon which the algorithm will operate.

\begin{definition}[Marginal demand graph]
    Given a partially stable one-to-one outcome $(\tau, p)$, the \emph{marginal demand graph} $D(\tau, p)$ is the graph $(U\cup S, E)$ where $(k,j)\in E$ if and only if $j\in F_{\tau, p}(k)$. 
\end{definition}

Given a subgraph $H$ of the marginal demand graph $D(\tau, p)$, we define $U(H) = U\cap V(H)$ and $S(H) = S\cap V(H)$, where $V(G)$ are the vertices of a graph $G$. Also, define $B(H) = \{\text{buy}(k)\mid k\in U(H)\}$. Similarly, if $E'$ is a set of edges in $D(\tau, p)$, then we define $V(E') = \{i\mid(i,j)\in E'\}\cup\{j\mid (i,j)\in E'\}$, as well as $U(E') = U\cap V(E')$, and $S(E') = S\cap V(E')$.

Our first important result describes the conditions necessary to verify an outcome is a competitive equilibrium. This result is based on the concept of maximal alternating trees, described next.

\subsection{Maximal alternating trees}\label{sec:framework:MAT}

\begin{definition}[Maximal alternating tree (MAT)]
    Let $H = (U\cup S, E)$ be a graph, $\tau\subseteq E\subseteq U\times S$ be a one-to-one matching, and $k_0\in U$ be an unmatched unit-buyer. An \emph{alternating tree} $\mathcal T$ in $H$ rooted at $k_0$ is a tree subgraph of $H$ such
 that all paths from the root $k_0$ alternate between edges in $\tau$ and edges not in $\tau$. An alternating tree $\mathcal T$ in a marginal demand graph $D(\tau, p)$ is \emph{maximal} if
    \begin{itemize}
        \item for every unit-buyer $k\in \mathcal T$, we have $F_{\tau, p}(k)\subseteq \mathcal T$, and
        \item for every item $j\in\mathcal T$, we have $\tau(j)\subseteq \mathcal T$.
    \end{itemize}
    See Figure~\ref{fig:MDG_MAT} for an example. If $j\in \mathcal T$ and $(k,j)\in\tau$, then $(k,j)$ is an edge in $\mathcal T$, as the (unique) path from the root to $k$ is alternating, has an even number of edges, the first of which is unmatched. 
\end{definition}

\begin{remark}\label{rem:mat_exists}
    MATs are a common structure seen in bipartite matching problems. It is a standard result that whenever a vertex (i.e. a unit-buyer) is not matched, a MAT exists and can be found using a process similar to breadth-first search (see, e.g., \cite{schrijver2003combinatorial}), which we prove in Lemma~\ref{lem:existmat}. Conversely, the existence of a MAT also implies the existence of an unmatched unit-buyer by definition. We leverage these facts in Algorithm \ref{alg:mat}.
\end{remark}

Note that given the existence of dummy items, every competitive equilibrium can be expanded into a competitive equilibrium such that all unit-buyers are matched, by matching any unmatched unit-buyers to dummy items. We can then show that the absence of a MAT serves as a certificate of a competitive equilibrium. The proof is deferred to Appendix \ref{sec:proofs:ce_cert}. 

\begin{lemma}[MATs and competitive equilibrium]\label{lem:smo_core_cm}
    Let $(\tau, p)$ be a feasible and partially stable outcome. Then, $(\tau, p)$ is a competitive equilibrium where all unit-buyers are matched if and only if there does not exist a MAT in $D(\tau, p)$.
\end{lemma}

Since a MAT exists if and only if an unmatched unit-buyer exists, Lemma~\ref{lem:smo_core_cm} shows that if we can match every unit-buyer while maintaining partial stability, we have a competitive equilibrium.

\begin{figure}[]
    \centering
    \begin{subfigure}[t]{.25\linewidth}
    \centering\includegraphics[scale=1]{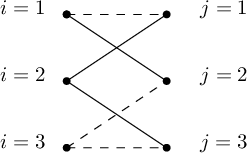}
    \caption{\footnotesize The original graph and partially stable matching}
    \end{subfigure}
    \begin{subfigure}[t]{.25\linewidth}
    \centering\includegraphics[scale=1]{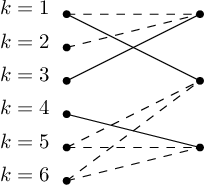}
    \caption{\footnotesize Marginal demand graph}
    \end{subfigure}
    \begin{subfigure}[t]{.25\linewidth}
    \centering\includegraphics[scale=1]{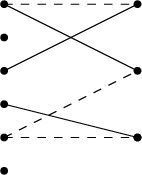}
    \caption{\footnotesize The MAT rooted at $k = 5$}
    \end{subfigure}
    \caption{A marginal demand graph and MAT, with $\text{ubuy}(1) = \{1,2\}$, $\text{ubuy}(2) = \{3,4\}$, and $\text{ubuy}(3) = \{5,6\}$. Solid (resp., dashed) lines represent the current matching (resp., unmatched edges in each (unit-)buyer's demanded set). We omit some unit-buyer copies for clarity.}
    \label{fig:MDG_MAT}
\end{figure}

\subsection{Description of the algorithm}\label{sec:framework:alg}

A key component of the algorithm is the use of augmenting paths, a classical tool from graph theory. Given a matching $\tau$ in a graph $G$, an augmenting path is a path that alternates between edges in $\tau$ and $G\setminus \tau$ and begins and ends at unmatched nodes. Flipping every edge in an augmenting path allows us to increase the size of the matching by 1.

MATs are the main structure on which the algorithm operates. At each iteration, the algorithm identifies a MAT that we call the \emph{current MAT}. If an item in the current MAT is unmatched, the algorithm finds an augmenting path from the unmatched root buyer $k_0$ to the unmatched item and increases the size of the matching $\tau$ by flipping all edges in the path from unmatched to matched or vice versa. Otherwise, it carefully increases the prices of items in the MAT until buyers in the MAT demand other items, thus increasing the size of the current MAT.

The algorithm is formally defined in Algorithm~\ref{alg:gs_core_mech}. Starting with prices $p=0$ and an empty matching $\tau$, it iterates as long as there exists an unmatched unit-buyer $k_0$. At each outer-iteration (lines \ref{alglin:gs_core_mech:outer_while}-\ref{alglin:gs_core_mech:outer_endwhile}), the algorithm constructs the marginal demand graph $D(\tau, p)$ and identifies a MAT $\mathcal T$ rooted at $k_0$ using the subroutine \textsc{FindMAT} (line \ref{alglin:gs_core_mech:findmat1}), described in Section~\ref{sec:main_result:findmat}. If the MAT $\mathcal T$ has an unmatched item, the algorithm uses the subroutine \textsc{AugmentingPath}, described in Section~\ref{sec:main_result:augment}, to augment $\tau$ along a path from $k_0$ to that item (lines \ref{alglin:gs_core_mech:augment1}-\ref{alglin:gs_core_mech:augment2}). Otherwise, the algorithm uses the subroutine \textsc{FindPriceIncrease}, described in Section~\ref{sec:FPI}, to conduct a version of Walrasian t\^atonnement by finding a price increase that expands the MAT $\mathcal T$ (lines \ref{alglin:gs_core_mech:inner_while}-\ref{alglin:gs_core_mech:inner_endwhile}).

\begin{algorithm}[H]
    \caption{Competitive equilibrium algorithm for the one-to-many assignment game}
    \label{alg:gs_core_mech}
    \begin{algorithmic}[1]
        \Require Buyers $B$, items $S$, valuation functions $v_i$, effective price functions $q_{ij}$
        \Ensure Competitive equilibrium $(\mu, p)$
        \State Initialize prices $p_j \gets 0$ for all  $j\in S$ and matching $\tau \gets \emptyset$
        \While{there is a unit-buyer  $k_0 \in [|B||S|]$ s.t. $\tau(k_0)=\emptyset$}\label{alglin:gs_core_mech:outer_while}
            \State $\mathcal T \gets \textsc{FindMAT}(k_0, \tau, D(\tau, p))$ \label{alglin:gs_core_mech:findmat1}\Comment{Compute current MAT rooted at $k_0$}
            \While{$\not\exists j\in S(\mathcal T)$ s.t. $\tau(j) = \emptyset$}\Comment{loop: Walrasian t\^atonnement}\label{alglin:gs_core_mech:inner_while}
                \State $\tau, p' \leftarrow \textsc{FindPriceIncrease}(\mathcal T, p,\tau)$ \label{alglin:gs_core_mech:findpriceinc1}
                \State $p\gets p + p'$  \label{alglin:gs_core_mech:findpriceinc2}
                \State $\mathcal T\gets \textsc{FindMAT}(k_0, \tau, D(\tau, p))$ \label{alglin:gs_core_mech:findmat2} \Comment{Update current MAT rooted at $k_0$}
            \EndWhile \label{alglin:gs_core_mech:inner_endwhile}
            \State $P\gets \textsc{AugmentingPath}(\mathcal T, \tau)$ \label{alglin:gs_core_mech:augment1}\Comment{Once MAT has unmatched item, augment}
            \State $\tau\gets \tau\triangle P$ \label{alglin:gs_core_mech:augment2}
        \EndWhile \label{alglin:gs_core_mech:outer_endwhile}
        \State $\mu\gets \{(\text{buy}(k),j)\mid (k,j)\in\tau\}$
        \State \Return $(\mu, p)$ 
    \end{algorithmic}
\end{algorithm}

\section{The \textsc{FindPriceIncrease} Subroutine}\label{sec:FPI}

The primary subroutine of Algorithm~\ref{alg:gs_core_mech} and the main technical contribution of this paper is \textsc{FindPriceIncrease}, described in Algorithm \ref{alg:findpriceincrease}. Given a MAT $\mathcal T$ where all items are matched, the subroutine finds a direction vector $d$ and step size $\lambda^*>0$ such that after increasing the current set of prices by $p \gets p+ \lambda^* d$, the size of the MAT $\mathcal T$ increases. This process is conceptually similar to the classical Walrasian t\^atonnement process, although price increases are not uniform and thus much more difficult to compute. Proofs omitted from Section \ref{sec:FPI} are presented in Appendix \ref{sec:proofs:FPI}.

\subsection{MAT-preserving price increases}\label{sec:FPI:MATppi}

In this section, we introduce MAT-preserving price increases, which are similar to structures from \citet{alkan1992equilibrium} and \citet{dutting2011expressive}. We start with some additional notation.

\begin{definition}[MAT-preserving price increase]\label{def:mat_preserving_price_increase}
    Given a partially stable and feasible one-to-one outcome $(\tau, p)$, a \emph{price increase} at $(\tau, p)$ is a vector $d\ge 0$ over the items $S$. A MAT-preserving price increase $d$ with respect to the marginal demand graph $D(\tau, p)$ and a MAT $\mathcal T$ rooted at $k_0$ where all items are matched, satisfies:
    \begin{enumerate}[(a)]
    \item \label{def:mat_preserving_price_increase:positive}
            $d_j>0$ if and only if $j\in S(\mathcal T)$.
    \item \label{def:mat_preserving_price_increase:matching_tree}
        There is some partially stable matching $\tau'$  at $p$ and MAT $\mathcal T'$ rooted at $k_0$ in $D(\tau', p)$ such that
        \begin{enumerate}[i.]
            \item  \label{def:mat_preserving_price_increase:matching} 
                $\tau'$  coincides with $\tau$ on $U\setminus U(\mathcal T)\times S\setminus S(\mathcal T)$.
            \item \label{def:mat_preserving_price_increase:tree}
                $V(\mathcal T) = V(\mathcal T')$.
            \item \label{def:mat_preserving_price_increase:utility}
                For sufficiently small $\lambda > 0$, we have $(\tau', p+\lambda d)$ is partially stable and the edges of $\mathcal T'$ are contained in the marginal demand graph. We write this compactly as $\mathcal T'\subseteq D(\tau', p+\lambda d)$.
        \end{enumerate}
    \end{enumerate}
\end{definition}

Note that the current MAT rooted at $k_0$ changes from $\mathcal T$ to $\mathcal T'$ after the MAT-preserving price increase. The reason is that given a specific MAT $\mathcal T$, there may not exist a price increase that preserves exactly the edges of $\mathcal T$ (as shown in Example~\ref{example} below). However, as we will later see in Algorithm \ref{alg:findpriceincrease}, the matching $\tau$ can be adjusted to find a price increase that preserves the edges of another MAT $\mathcal T'$ that contains the same unit-buyers and items. In this way, the vertices of the MAT rooted at $k_0$ remain the same after a sufficiently small price increase along $d$.

\begin{example}
\label{example}
    Consider the MAT rooted at $k_0=3$ in Figure \ref{fig:MAT_price_inc}(a). The three unit-buyers correspond to distinct buyers. For each unit buyer $k$, we have $v_{k}(T) = |T|$. We have $p_1 = p_2 = 0$. Then, define the effective price functions as follows: $q_{11}(p_1) = 2p_1$, $q_{12}(p_2) = p_1$, $q_{21}(p_1) = p_1$, $q_{22}(p_2) = 2p_2$, $q_{31}(p_1) = p_1$, and $q_{32}(p_2) = p_2$. We can verify that both $\tau = \{(1,1), (2,2)\}$ and $\tau' = \{(1,2), (2,1)\}$ are partially stable, and that the marginal demand graph $D(\tau, p)$ is a complete bipartite graph. However, there is no price increase $d$ such that for small $\lambda >0$, we have $u_1(2\mid 1, p+\lambda d)\le 0$ and $u_2(1\mid 2, p+\lambda d)\le 0$. Instead, we can change the matching to $\tau'$ and the MAT to $\mathcal T'$ as given in Figure \ref{fig:MAT_price_inc}(b). Setting $d_1 = d_2 = 1$, we can verify that for small $\lambda > 0$, $\tau'$ is partially stable at $p+\lambda d$, and that $\mathcal T'\subseteq D(\tau', p+\lambda d)$.
\end{example}

\begin{figure}[H]
    \centering
    \begin{subfigure}[t]{.3\linewidth}
    \centering\includegraphics[scale=1]{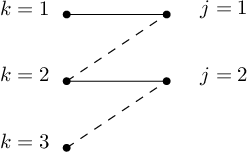}
    \caption{\footnotesize Original matching $\tau$, MAT $\mathcal T$}
    \end{subfigure}
    \begin{subfigure}[t]{.3\linewidth}
    \centering\includegraphics[scale=1]{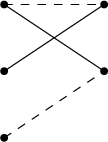}
    \caption{\footnotesize Adjusted matching $\tau'$, MAT $\mathcal T'$}
    \end{subfigure}
    \caption{An example of a MAT to be preserved in a MAT-preserving price increase. Solid lines represent the current matching, and dashed lines represent other edges of the marginal demand graph in the MAT given the current matching.}
    \label{fig:MAT_price_inc}
\end{figure}
 
The following lemma gives a sufficient condition for a price increase $d$ to be a MAT-preserving price increase. This condition depends on the slopes $q'_{kj}(p)$ of the effective price functions $q_{kj}(p)$. The proof is deferred to Appendix \ref{sec:proofs:FPI:MATppi}.

\begin{lemma}[MAT-preserving price increase and slopes of effective price functions]\label{lem:envyfree_slopes}
    Given a partially stable outcome $(\tau, p)$ and MAT $\mathcal T$, let $d$ be a price increase that satisfies condition (\ref{def:mat_preserving_price_increase:positive}) from Definition \ref{def:mat_preserving_price_increase}. Let $\tau'$ and $\mathcal T'$ satisfy conditions (\ref{def:mat_preserving_price_increase:matching}.) and (\ref{def:mat_preserving_price_increase:tree}.). If $q'_{kj_1}(p_{j_1})d_{j_1} \le q'_{kj_2}(p_{j_2})d_{j_2}$ for all $(k, j_1)\in \mathcal T'$ and $j_2\in F_{\tau', p}(k)$, then condition (\ref{def:mat_preserving_price_increase:utility}) holds. Thus, $d$ is a MAT-preserving price increase with respect to the marginal demand graph $D(\tau, p)$ and MAT $\mathcal T$.
\end{lemma}

\paragraph{Overview of the remainder of Section~\ref{sec:FPI}.} To find a MAT-preserving price increase with respect to a marginal demand graph $D(\tau, p)$ and a MAT $\mathcal T$, it remains to find a price increase $d$, matching $\tau'$, and a MAT $\mathcal T'$ that satisfy the conditions of Lemma~\ref{lem:envyfree_slopes}. The first step is to find a minimum weight perfect matching $\tau^*$ over $U(\mathcal T) \setminus \{k_0\}\times S(\mathcal T)$ such that $\tau \setminus E(\mathcal T) \cup \tau^*$ is partially stable, where the weights depend on the slopes of the effective price functions. Section~\ref{sec:FPI:matroid} shows such a matching can be found by solving a minimum weight common basis problem over two matroids. Section~\ref{sec:FPI:duality} then shows that this minimum weight perfect matching $\tau^*$ can be used to construct a linear program whose optimal solution is then transformed into a price increase that satisfies the conditions of Lemma~\ref{lem:envyfree_slopes}. Then, in Section \ref{sec:FPI:alg} we give the global description of the subroutine.

\subsection{Finding a minimum weight partially stable perfect matching via the weighted common basis problem}\label{sec:FPI:matroid}

For the remainder of Section \ref{sec:FPI}, fix an outcome $(\tau, p)$, with one-to-many projection $\mu$. Also, fix a MAT $\mathcal T$ with root $k_0$, where all items in $\mathcal T$ are matched (recall that subroutine \textsc{FindPriceIncrease} is only called over MATs $\mathcal T$ whose items are all matched). We define $B'(\mathcal T) := \{\text{buy}(k)\mid k\in U(\mathcal T)\setminus \{k_0\}\}$ to be the buyers corresponding to unit-buyers in $\mathcal T$, other than the root $k_0$, and $W(\mathcal T) := B'(\mathcal T)\times S(\mathcal T)$ to be the complete set of edges between these buyers and the items in $\mathcal T$. The first step of the \textsc{FindPriceIncrease} subroutine is to assign each edge $(i,j)\in W(\mathcal T)$ a weight $\log q'_{ij}(p_j)$, which is the log of the slope of the effective price function of buyer $i$ for item $j$ at price $p_j$. We wish to compute a minimum weight set of edges $\mu^*\subseteq W(\mathcal T)$ such that each agent has the same number of matches under $\mu$ and $\mu^*$, and $\mu\setminus E(\mathcal T)\cup \mu^*$ is partially stable. We call such a matching $\mu^*$ a \emph{minimum weight partially stable perfect matching}.

The partial stability constraint can be formulated as a matroid constraint over the edges. This constraint also ensures that the buyers have the same number of matches as before. A second matroid constraint is used to ensure that the items are perfectly matched as well. Restricting to bases of the matroids imposes the maximum cardinality constraint. Putting these results together, we can cast the minimum weight partially stable perfect matching problem as a minimum weight common basis problem. The proofs in this section are deferred to Appendix \ref{sec:proofs:FPI:matroid}. 

We first define the matroid $M_B$. Essentially, a basis in $M_B$ is a set of edges  that maximizes the sum of the utilities to the buyers, subject to the bundle sizes originally allocated to each buyer. Thus, each buyer is assigned a bundle that it demands at the current prices under the size constraint.

\begin{lemma}[Buyer demand matroid]\label{lem:MB_matroid}
Let $(\tau,p)$ be a partially stable outcome with one-to-many projection $\mu$ and ${\mathcal T}$ a MAT in $D(\tau,p)$. For $i \in B'({\mathcal T})$, let ${\mathcal E}_i= \{E \subseteq \{i\} \times S({\mathcal T}) : |E| = |U'(\mathcal T) \cap \text{ubuy}(i)|\}$.
Define the family of sets
$$\mathcal J_B = \left\{\bigcup_{i\in B'(\mathcal T)}F_i \ \bigg |\ F_i\in \arg\max_{F \in {\mathcal E}_i} u_i(S(F) \mid \mu(i)\setminus S({\mathcal T}),p)\right\}.$$ 
That is, $\mathcal J_B$ is the family of sets of edges such that each buyer is connected to a set of items which maximize their individual utility, subject to the capacity constraint. Let ${\mathcal I}_B$ be the family of all subsets of ${\mathcal J}_B$. Then, $M_B=(W({\mathcal T}),{\mathcal I}_B)$ is a matroid with bases ${\mathcal J}_B$.
\end{lemma}

We can efficiently implement an independence oracle for $M_B$.

\begin{lemma}[Independence oracle for $M_B$]\label{lem:M1_independence_oracle}
    There exists an independence oracle for $M_B$ that runs in $O(|B|^2|S|^2)$ time.
\end{lemma}

Next, we can show that a matching is $\mathcal T$-perfect (i.e., matches each item) if and only if it is a basis in the \emph{item capacity matroid}.

\begin{definition}[Item capacity matroid $M_S$]\label{def:item_capacity_matroid} 
    Let $(\tau, p)$ be a partially stable outcome and  $\mathcal T$ a MAT at $(\tau, p)$. Let $\mathcal I_S=\{E\subseteq W(\mathcal T): \, |E(j)|\le 1 \, \hbox{for all } j\in S(\mathcal T)\}$ and $M_S = (W(\mathcal T), \mathcal I_S)$, where $E(j) = \{(i,j)\in E\}$.
\end{definition}

$M_S$ is a partition matroid (see, e.g.,~\cite{oxley2006matroid}). The next two lemmas characterize the bases of $M_S$ and give an independence oracle for $M_S$. They are well known facts from bipartite matching (see, e.g., \cite{schrijver2003combinatorial}), so the proofs are omitted.

\begin{lemma}[Basis of $M_S$ if and only if saturates item capacities]\label{lem:M2_basis}
    The family $\mathcal J_S$ of bases of $M_S$ is given by all and only the set of edges $E\subseteq W(\mathcal T)$ such that $|E(j)| = 1$ for all $j\in S(\mathcal T)$.
\end{lemma}

We can also efficiently compute an independence oracle for $M_S$.

\begin{lemma}[Independence oracle for $M_S$]\label{lem:M2_independence_oracle}
    There exists an algorithm that identifies whether a set $E$ is independent in $M_S$ in $O(|B|^2|S|)$ time.
\end{lemma}

\begin{definition}[Lift]\label{defn:lift}
    Given a one-to-many matching $\mu$, a \emph{lift} of $\mu$ is a one-to-one matching $\tau$ that satisfies $\{(\text{buy}(k), j)\mid (k,j)\in \tau\} = \mu$. We define the function $\textsc{Lift}(\mu)$ to output a lift of $\mu$, selecting one arbitrarily if more than one exist.
\end{definition}

We can then show that a set of edges in $W(\mathcal T)$ that is a minimum weight common basis of $M_B$ and $M_S$ has a lift that is a partially stable perfect matching on $V(\mathcal T)\setminus \{k_0\}$.

\begin{lemma}[Min weight common basis finds min weight partially stable perfect matching]\label{lem:min_weight_common_basis}
    A lift $\tau^* = \textsc{Lift}(\mu^*)$ of the optimal solution $\mu^*$ to the weighted minimum common basis problem
    \[\min_{E \in \mathcal J_B \cap \mathcal J_S}\sum_{(i,j)\in W(\mathcal T)}\log q'_{ij}(p_j)\]
    is a minimum weight partially stable perfect matching on $V(\mathcal T)\setminus \{k_0\}$, where the weight of an edge $(k,j)$ is also given by $\log q'_{kj}(p_j)$.
\end{lemma}

We also show in the proof of Lemma \ref{lem:min_weight_common_basis} that the original partially stable matching $\mu$ is a basis of both $M_B$ and $M_S$, though it may not have minimum weight. We also know that an efficient algorithm exists for finding a common basis of minimum weight.

\begin{lemma}[Common matroid basis algorithm]\label{lem:min_weight_common_basis_alg}
    There exists an algorithm \textsc{CommonBasis} that finds the minimum weight common matroid basis in $O(|B|^3|S|^3)$ oracle calls.
\end{lemma}

The algorithm repurposes an existing weighted matroid intersection algorithm, such as those by \citet{lawler1975matroid} and \citet{frank1981weighted}, and applies it to a transformed set of weights. 

\subsection{The duality trick}\label{sec:FPI:duality}

In this section, we describe the linear programming duality trick that is used to find a MAT-preserving price increase. In particular, we show that an optimal solution to a linear program defined as a function of $\tau^* = \textsc{Lift}(\mu^*)$, the lift of the optimal solution to the minimum weight common basis problem defined in Section~\ref{sec:FPI:matroid}, satisfies complementary slackness constraints with $\tau^*$ itself. We begin with the definition of this linear program. 
 
\begin{definition}[Linear program for MAT-preserving price increase]
    Given a partially stable outcome $(\tau, p)$, with unmatched unit-buyer $k_0$ and corresponding MAT $\mathcal T$ where all items in $S(\mathcal T)$ are matched, define the MAT-preserving price increase linear program $LP(\tau, p, \mathcal T)$ over variables $\omega,\rho$ as:
    \begin{align*}
        \max & \sum_{k\in U(\mathcal T)}\omega_i + \sum_{j\in S(\mathcal T)}\rho_j&\\
         \text{s.t. } & \omega_k + \rho_j \le \log q'_{kj}(p_j) & \forall (k, j)\in H
     \end{align*}     
     where $H = D(\tau, p)\cap V(\mathcal T)\setminus\{k_0\}$.
\end{definition}

We now state the main technical contribution of this section, Lemma \ref{lem:MWlegalmatching_correct}: the constraints in $LP(\tau^*, p,\mathcal T)$ corresponding to edges $(k,j) \in \tau^*$ are tight. These constraints being tight for $(\omega^*, \rho^*)$ are crucial to use $\rho^*$ to compute the desired price increase via complementary slackness.

\begin{lemma}[Min-weight partially stable perfect matching satisfies equality constraints]\label{lem:MWlegalmatching_correct}
    Let $\mu^*$ be an optimal solution to the weighted minimum common basis problem given by
    \[\min_{E\in\mathcal J_B\cap\mathcal J_S} \sum_{(i,j)\in W(\mathcal T)}\log q'_{ij}(p_j)\]
    and let $\tau^* = \textsc{Lift}(\mu^*)$. Let $(\omega^*, \rho^*)$ be an optimal solution to $LP(\tau^*, p,\mathcal T)$. Then, for all $(k,j)\in \tau^*$, we have $\omega^*_k + \rho^*_j = \log q'_{kj}(p_j)$.
\end{lemma}

As a first step towards the proof of Lemma~\ref{lem:MWlegalmatching_correct}, we need a helper result on the structure of the marginal demand graph. Essentially, if a buyer $i$ has many unit-buyer copies with both matched and unmatched edges in the marginal demand graph, we can either swap all of the edges simultaneously while preserving partial stability, or we can find additional edges in the marginal demand graph.

\begin{lemma}[Structure of marginal demand graph]\label{lem:alt_cycles}
    Let $(\tau, p)$ be a feasible and partially stable outcome. Assume there exist, for some $z \geq 2$, $(k_1,j_1),\ldots(k_z, j_z)\in \tau$, and $(k_1,j_1'),\ldots(k_z, j_z')\in D(\tau, p)\setminus\tau$, where $k_1,\ldots, k_z\in \text{ubuy}(i)$ for some $i$ and are distinct and all of the $j_\ell'$ are distinct. Then, at least one of the following is true:
    \begin{enumerate}[(a)]
        \item For all subsets $X\subseteq [z]$, we have 
        $\tau(\text{ubuy}(i))\cup \{j_x'\mid x\in X\}\setminus \{j_x\mid x\in X\} \in \argmax_{T\subseteq S:|T|=|\tau(\text{ubuy}(i))|}u_i(T, p)$, i.e., the bundle $\tau(\text{ubuy}(i))\cup \{j_x'\mid x\in X\}\setminus \{j_x\mid x\in X\}$ maximizes $u_i(\cdot, p)$, among bundles of size $|\tau(\text{ubuy}(i))|$.
        \item There exists a sequence of distinct indices $x_1,\ldots, x_s \in [z]$ s.t. $(k_{x_1}, j'_{x_2}, k_{x_2}, j'_{x_3}, k_{x_3}, \dots, k_{x_s}, j'_{x_1}, k_{x_1})$ is a cycle in $D(\tau,p)\setminus \tau$.
    \end{enumerate}
\end{lemma}

The proof of this result is deferred to Appendix \ref{sec:proofs:FPI:duality}. See Figure \ref{fig:MDG_structure} for an example of the lemma. This intermediate result will help us show that if a negative weight alternating cycle exists in the marginal demand graph, then a negative weight alternating cycle exists that also preserves partial stability when augmented. We can now prove Lemma~\ref{lem:MWlegalmatching_correct}.

\begin{figure}[H]
    \centering
    \begin{subfigure}[t]{.3\linewidth}
    \centering\includegraphics[scale=1]{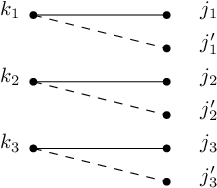}
    \caption{\footnotesize Original demanded set}
    \end{subfigure}
    \begin{subfigure}[t]{.3\linewidth}
    \centering\includegraphics[scale=1]{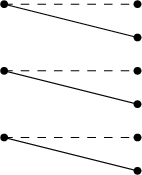}
    \caption{\footnotesize Possible alternative demanded set}
    \end{subfigure}
    \begin{subfigure}[t]{.3\linewidth}
    \centering\includegraphics[scale=1]{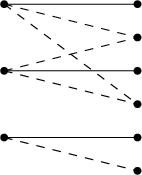}
    \caption{\footnotesize One possible version of the full marginal demand graph}
    \end{subfigure}
    \caption{A marginal demand graph as shown in Lemma \ref{lem:alt_cycles}, where $k_1,k_2,k_3\in \text{ubuy}(i)$. Consider the portion shown in Figure \ref{fig:MDG_structure}(a). The matching $\tau$ shown by the solid lines is partially stable. Then, by Lemma \ref{lem:alt_cycles}, either the matching in Figure \ref{fig:MDG_structure}(b) also provides $i$ with optimal utility $u_i(\cdot, p)$ among bundles of size $|\tau(\text{ubuy}(i))|$, or we can find additional unmatched edges in the marginal demand graph, such as in Figure \ref{fig:MDG_structure}(c).}
    \label{fig:MDG_structure}
\end{figure}

\begin{proof}{Proof of Lemma~\ref{lem:MWlegalmatching_correct}}
    Consider the linear program $ DP(\tau, p, \mathcal T):$
    \begin{align*}
        \min & \sum_{k\in U(\mathcal T),j\in S(\mathcal T)}\lambda_{kj}\log q'_{kj}(p_j) &\\
         \text{s.t. } & \sum_{j\in F_{\tau, p}(k)}\lambda_{kj} = 1 & \forall k\in U'(\mathcal T) \\
          & \sum_{k\in F^{-1}_{\tau, p}(j)}\lambda_{kj} = 1 & \forall j\in S(\mathcal T) \\
          & \lambda_{kj}\ge 0 & \forall (k, j)\in H
    \end{align*}
    
    where as before, $H = D(\tau, p)\cap V(\mathcal T)\setminus\{k_0\}$. Also, $U'(\mathcal T)$ is defined as $U(\mathcal T)\setminus \{k_0\}$. We see that $DP(\tau^*, p, \mathcal T)$ is a minimum weight perfect matching linear program, and is the dual of $LP(\tau^*, p, \mathcal T)$. We wish to show that $\tau^*$ is an optimal solution for $DP(\tau^*, p, \mathcal T)$. Let $y\in \{0,1\}^{U(\mathcal T)\setminus\{k_0\}\times S(\mathcal T)}$ be indicators for the edges in the  residual network at $\tau^*$, i.e,  $y_{kj} = 1$ if and only if $\tau^*_{kj} = 0$ and $y_{jk} = 1$ if and only if $\tau^*_{kj} = 1$, with $y_{kj}$ and $y_{jk}$ equal to 0 otherwise, for each $(k,j)\in D(\tau, p)\cap V(\mathcal T)\setminus\{k_0\}$. It is well known that $\tau^*$ is a minimum weight perfect matching if and only if there does not exist a negative-weight cycle in the residual network at $\tau^*$ (see, e.g., \cite{schrijver2003combinatorial}). Suppose for contradiction that there exists a negative weight cycle $C = (j_1, k_1, j_2, k_2,\ldots, j_z, k_z)$ with fewest edges, where $\tau^*_{k_\ell j_\ell} = 1$ and $\tau^*_{k_\ell j_{\ell+1}} = 0$ for all $\ell$ (and $j_{z+1}:= j_1$). Suppose first $\tau^*\triangle C$ is a partially stable perfect matching. Then by construction its weight is strictly smaller than the weight of $\tau^*$, a contradiction to Lemma~\ref{lem:min_weight_common_basis}. Thus, $\tau^* \triangle C$ is not partially stable. In particular, some buyer must appear in multiple unit-buyer copies in $C$. If not, then by definition of the marginal demand graph, we have $u_{k_\ell}(\tau^*(\text{copy}(k_\ell)), p) = u_{k_\ell}(\tau^*(\text{copy}(k_\ell))\cup\{j_{\ell+1}\}\setminus \{j_\ell\}, p)$. That is, in the augmented matching $\tau^*\triangle C$, every buyer has the same utility as in $\tau^*$, which means that $\tau^*\triangle C$ is again partially stable, a contradiction.
  
    Let therefore $i$ be a buyer with multiple unit-buyer copies in $V(C)$. Then, each such unit-buyer $k_{x}\in \text{ubuy}(i)\cap V(C)$ has an edge $(k_{x},j_{x}')\in D(\tau, p)\setminus \tau^*$, given by one of its edges in $C$. We can now apply Lemma \ref{lem:alt_cycles}. Suppose that Lemma \ref{lem:alt_cycles}(a) is true for every buyer $i$ with multiple unit-buyer copies in $V(C)$, where $k_1,\dots, k_z$ are the copies of the unit buyer $i$ in $C$, and $j_1,\dots, j_z$ (resp, $j_1',\dots, j_z'$) are the items adjacent to them in $\tau^*$ (resp., $C \setminus \tau^*$). Fixing any such $i$, we know by Lemma \ref{lem:alt_cycles}(a) that by setting $X = \{x\mid k_{x}\in\text{ubuy}(i)\cap V(C)\}$, the bundle \begin{equation}\label{eq:augmented}\tau^*(\text{ubuy}(i))\cup\{j_{x}'\mid x\in X\}\setminus \{j_{x}\mid x\in X\}\end{equation} maximizes $u_i(\cdot, p)$ among bundles of size $|\tau^*(\text{ubuy}(i))|$. However,~\eqref{eq:augmented} is exactly equal to $(\tau^*\triangle C)(\text{ubuy}(i))$. Thus, $\tau^*\triangle C$ matches this buyer $i$ with a bundle maximizing $u_i(\cdot, p)$ among sets of size $|\tau^*(\text{ubuy}(i))|$. Since $i$ was an arbitrary buyer with multiple copies in $C$ and clearly $\tau^*\triangle C$ matches every buyer $i'$ with exactly one unit-buyer copy in $C$ to a bundle of size $|\tau^*(\text{ubuy}(i'))|$ maximizing $u_{i'}(\cdot, p)$, $\tau^*\triangle C$ is partially stable, a contradiction. 

    Thus, there must exist a buyer $i$ with multiple unit-buyer copies in $V(C)$ for which Lemma \ref{lem:alt_cycles}(b) applies with $j_x$ (resp.~$j'_x$) being the items on edges in $C\cap \tau^*$ (resp.~$C\setminus \tau^*)$ incident to its unit-buyer copies. Fix one such buyer $i$, with unit-buyers $k_{1},\cdots, k_z\in \text{ubuy}(i)\cap V(C)$, matched edges $(k_x, j_x)\in C$, and unmatched edges $(k_x, j_x') = (k_x, j_{x+1})\in C$ for each $1\le x\le z$. Let $x_1,\cdots, x_s$ be the sequence of indices described in Lemma \ref{lem:alt_cycles}(b). Then, $(k_{x_1}, j'_{x_2}, k_{x_2}, j'_{x_3}, k_{x_3}, \dots, k_{x_s}, j'_{x_1})$ is a cycle in $D(\tau, p)\setminus\tau$. Call each of the edges of the form $(k_{x_\ell}, j'_{x_{\ell+1}})$ in this cycle a \emph{chord}. Now, consider the set of cycles $C_1,\ldots, C_s$, where $C_\ell$ is given by $C_\ell = (k_{x_\ell}, j'_{x_{\ell+1}} = j_{x_{\ell+1}+1}, k_{x_{\ell+1}+1}, j_{x_{\ell+1}+2},k_{x_{\ell+1}+2},\ldots, k_{x_{\ell-1}}, j_{x_\ell}, k_{x_\ell})$. That is, $C_\ell$ is the ``shortcutted'' cycle obtained by first traversing the chord $(k_{x_\ell}, j'_{x_{\ell+1}})$, then following $C$ until we return back to $k_{x_\ell}$. Note that $w(k_{x_\ell}, j'_{x_{\ell+1}}) = w(k_{x_{\ell+1}}, j'_{x_{\ell+1}}) = w(k_{x_{\ell+1}}, j_{x_{\ell+1}+1})$, where $w(k,j) = \log q'_{kj}(p_j)$ is the objective weight function from $DP$, as the unit-buyers $k_{x_\ell}$ and $k_{x_{\ell+1}}$ are copies of the same buyer. It follows that the weight of each $C_\ell$ is given by $w(C_\ell) = w(k_{x_\ell}, j'_{x_{\ell+1}}) + w(j_{x_{\ell+1}+1}, k_{x_{\ell+1}+1})+ w(k_{x_{\ell+1}+1}, j_{x_{\ell+1}+2}) + \dots+ w(k_{x_\ell-1}, j_{x_\ell-1}) + w(j_{x_\ell-1}, k_{x_\ell})$, which can be rewritten as $w(k_{x_\ell+1}, j_{x_{\ell+1}+1}) + w(j_{x_{\ell+1}+1}, k_{x_{\ell+1}+1})+ w(k_{x_{\ell+1}+1}, j_{x_{\ell+1}+2}) + \dots+ w(k_{x_\ell-1}, j_{x_\ell-1}) + w(j_{x_\ell-1}, k_{x_\ell})$. That is, the weight of $C_\ell$ is equal to the sum of the weights of the arcs of $C$ from $k_{x_{\ell+1}}$ to $k_{x_\ell}$. Adding the weights of $C_1,\ldots, C_s$ gives us the weight of the arcs of $C$ from $k_{x_2}$ to $k_{x_1}$, then $k_{x_3}$ to $k_{x_2}$, and so on, completing with the arc from $k_{x_{s+1}} = k_{x_1}$ to $k_{x_s}$. In other words, the sum of the weights of $C_1,\dots, C_s$ is equivalent in weight to a grand cycle starting from $k_{x_1}$, following $C$ an integer $m$ times around, and finishing back at $k_{x_1}$. It follows that $\sum_{\ell = 1}^s w(C_\ell) = m\cdot w(C) < 0$. There must then be at least one cycle $C_\ell$ with $w(C_\ell) < 0$. Furthermore, this cycle has strictly fewer edges than $C$ since two copies of the same buyer cannot be at distance $2$ in $C$. This contradicts the choice of $C$. 
 \end{proof}

Finally, once we have the optimal solution $(\omega^*, \rho^*)$ to $LP(\tau^*, p,\mathcal T)$, we apply a transformation to it, described in Algorithm \ref{alg:connectMAT}, that guarantees there exists a MAT $\mathcal T'$ containing $\tau^*$ such that, for the transformed solution, the constraints of $LP(\tau^*, p,\mathcal T)$ corresponding to edges in $\mathcal T'$ are all tight. We note that, before the transformation, complementary slackness guarantees that, for the optimal solution $(\omega^*, \rho^*)$, the constraints $\omega^*_k+\rho^*_j\le \log q'_{kj}(p_j)$ corresponding to edges $(k,j)\in \tau^*$ are tight. Also note that given two subgraphs $H, H'$ of the marginal demand graph, we may abuse notation and write $H\cap V(H')$ to denote the subgraph of $H$ contained in $U(H')\times S(H')$. 

In Algorithm \ref{alg:connectMAT}, $H_\rho$ is the graph of edges $(k,j)\in D(\tau^*, p)\cap V(\mathcal T)$ such that the constraint $\omega_k+\rho_j \le \log q'_{kj}(p_j)$ is tight. The goal is for $H_\rho$ to contain a MAT $\mathcal T'$ containing $\tau^*$. In line \ref{alglin:connectmat:T_prime}, Algorithm \ref{alg:connectMAT} computes the initial MAT $\mathcal T'$ rooted at $k_0$. The goal is to add edges to $H_\rho$ such that the MAT rooted at $k_0$ contains the same unit-buyers and items as $\mathcal T$. Then, in lines \ref{alglin:connectmat:epsilon} to \ref{alglin:connectmat:decr_omega}, the algorithm uniformly increases the entries of $\rho$ in $S(\mathcal T)\setminus S(\mathcal T')$ and decreases the entries of $\omega$ in $U(\mathcal T)\setminus U(\mathcal T')$. Translating $\omega$ and $\rho$ preserves tightness for all edges for which the dual constraint is tight, and tightens at least one more constraint, and an edge between $U(\mathcal T')$ and $S(\mathcal T)\setminus S(\mathcal T')$ enters $H_\rho$. The new edge is then added to $\mathcal T'$, as every edge from a unit-buyer $k\in U(\mathcal T')$ must be contained in any MAT containing $k$ by definition. This process repeats on the transformed solution $(\omega, \rho)$. The loop terminates once $\mathcal T'$ covers the same unit-buyers and items as $\mathcal T$, and Algorithm \ref{alg:connectMAT} returns the final solution $(\omega, \rho)$. The proof of next result is deferred to Appendix \ref{sec:proofs:FPI:duality}.

\begin{algorithm}[]
    \caption{\textsc{ConnectMAT}}
    \label{alg:connectMAT}
    \begin{algorithmic}[1]
        \Require Max weight partially stable matching $\tau^*$, prices $p$, original MAT $\mathcal T$ rooted at $k_0$, optimal solution $(\omega^*,\rho^*)$ to $LP(\tau^*, p,\mathcal T)$
        \Ensure Solution $(\omega, \rho)$ to  $LP(\tau^*, p,\mathcal T)$ such that there exists a MAT $\mathcal T'$ such that primal constraints corresponding to edges in $\mathcal T'$ are tight over $(\omega, \rho)$
        \State $\omega, \rho\gets \omega^*, \rho^*$
        \State $\omega_{k_0}\gets \min_{j\in S(\mathcal T)}\log q'_{k_0j}(p_j) - \rho_j$ \Comment{extend $\omega$ to $k_0$\label{alglin:connectmat:omega_k0}}
        \Do \label{alglin:connectmat:do}
            \State $H_\rho \gets \{(k, j)\in D(\tau^*, p)\cap V(\mathcal T)\mid \omega_k+\rho_j = \log q'_{kj}(p_j)\}$  \label{alglin:connectmat:H_rho}
            \State $\mathcal T'\gets \textsc{FindMAT}(H_\rho, \tau^*, k_0)$ \label{alglin:connectmat:T_prime}
            \State $\epsilon\gets \min_{(k, j)\in D(\tau^*, p):\ k\in U(\mathcal T'),\ j\in S(\mathcal T)\setminus S(\mathcal T')}\log q'_{kj}(p_j)-\omega_k-\rho_j$\label{alglin:connectmat:epsilon}\Comment{$\epsilon$ tightens a constr.}
            \State Set $\rho_j\gets \rho_j + \epsilon$ for all $j\in S(\mathcal T)\setminus S(\mathcal T')$ \label{alglin:connectmat:incr_rho}
            \State Set $\omega_k \gets \omega_k - \epsilon$ for all $k\in U(\mathcal T)\setminus U(\mathcal T')$ \label{alglin:connectmat:decr_omega}\Comment{by choice of $\epsilon$, one constraint tightens}
        \doWhile{$|U(\mathcal T')|+|S(\mathcal T')| < |U(\mathcal T)|+|S(\mathcal T)|$\Comment{so, $H_\rho$ increases by one edge each iteration}\label{alglin:connectmat:while}}
        \State \Return $\omega, \rho$
    \end{algorithmic}
\end{algorithm}

\begin{lemma}[Correctness and runtime of Algorithm \ref{alg:connectMAT}]\label{lem:connectmat}
    Algorithm \ref{alg:connectMAT} outputs $\omega, \rho$ such that $H_\rho$ contains a MAT $\mathcal T'$ containing $\tau^*$. Furthermore, it terminates in $O(|B||S|^3)$ time.
\end{lemma}

\subsection{The MAT-preserving price increase algorithm}\label{sec:FPI:alg}

We can now describe Algorithm \ref{alg:findpriceincrease}, \textsc{FindPriceIncrease},  for computing the MAT-preserving price increase $d$ and step size $\lambda^*$. Let \textsc{CommonBasis}($\mathcal I_B$, $\mathcal I_S$, $w$) be an algorithm that takes in two matroid independence oracles and a weight function and outputs the minimum weight common basis. Finally, let $\textsc{LinSeg}_{i, j}(p_j)$ be the highest index $\ell$ such that $p_j$ belongs to the $\ell^{th}$ linear segment for the effective price function $q_{ij}$. Algorithm \ref{alg:findpriceincrease} first solves the minimum weight common basis problem on $M_B$ and $M_S$. Then, it uses the solution $\tau^*$ to define an appropriate version of $LP$, where by Lemma \ref{lem:MWlegalmatching_correct}, the optimal solution satisfies complementary slackness with $\tau^*$. The algorithm calls \textsc{ConnectMAT} in line \ref{alglin:findpriceincrease:connectmat} to modify the dual variables $\omega^*, \rho^*$, and also extends $\omega$ to include $\omega_{k_0}$. Next, it uses the duality trick to compute the price increase $d$ (line \ref{alglin:findpriceincrease:duality_trick}). Finally, it calculates the maximum step size $\lambda^*$ the price can be increased along $d$ (line \ref{alglin:findpriceincrease:lambda}), until a new edge enters the marginal demand graph $D(p+\lambda^* d)$ (line \ref{alglin:findpriceincrease:lambda1}) or  $p+\lambda^* d$ leaves the current linear domain (line \ref{alglin:findpriceincrease:lambda2}). We prove in Appendix \ref{sec:proofs:FPI:alg} the correctness and runtime of Algorithm \ref{alg:findpriceincrease}.

\begin{algorithm}[]
    \caption{\textsc{FindPriceIncrease}}
    \label{alg:findpriceincrease}
    \begin{algorithmic}[1]
        \Require MAT $\mathcal T$ with respect to $\tau$ with root $k_0$ in $D(\tau, p)$, where all items in $S(\mathcal T)$ are matched
        \Ensure MAT-preserving price increase and step size $\lambda^* d$ for $\mathcal T$ and corresponding matching $\tau'$
        \State $\mu^*\gets$\textsc{CommonBasis}($\mathcal I_B$, $\mathcal I_S$, $w$) \label{alglin:findpriceincrease:maxweight}
        \State $\tau^*\gets \textsc{Lift}(\mu^*) \cup (\tau\setminus E(\mathcal T))$ \Comment{$\tau^*$ is an optimal solution to $LP$}
        \State $\omega^*, \rho^*\gets LP(\tau^*, p, \mathcal T)$\label{alglin:findpriceincrease:omegarho}\Comment{define $LP$ using max weight common basis}
        \State $\omega^*, \rho^*\gets \textsc{ConnectMAT}(\tau^*, p, \mathcal T, \omega, \rho)$ \label{alglin:findpriceincrease:connectmat} \Comment{modifies dual variables to preserve full MAT}
        \State Set $d_j \gets e^{-\rho^*_j}$ for all $j\in S(\mathcal T)$ and $d_j \gets 0$ otherwise \label{alglin:findpriceincrease:duality_trick}\Comment{duality trick}
        \State $\lambda_1^* \leftarrow \min\{\lambda \geq 0 :  (k, j) \in D(\tau^*, p + \lambda d) \text{ for some } k \in U(\mathcal T) \text{ and } j \not\in F_{\tau^*, p}(k)\}$ \label{alglin:findpriceincrease:lambda1}
        \State $\lambda_2^* \leftarrow \min\{\lambda \geq 0 :  \textsc{LinSeg}_{\text{buy}(k), j}(p_j) \neq \textsc{LinSeg}_{\text{buy}(k), j}(p_j + \lambda d_j) \text{ for some } k \in U \text{ and } j \in S \}$ \label{alglin:findpriceincrease:lambda2}
        \State $\lambda^* \leftarrow \min\{\lambda_1^*, \lambda_2^*\}$ \label{alglin:findpriceincrease:lambda}\Comment{max step size to preserve MAT and stay within linear domain}
        \State \Return $\tau^*, \lambda^* d$
    \end{algorithmic}
\end{algorithm}

\begin{lemma}[Correctness and runtime of Algorithm \ref{alg:findpriceincrease}]\label{lem:gs_price_inc}
    Algorithm \ref{alg:findpriceincrease} finds a MAT-preserving price increase $d$ and maximal step size $\lambda^*$ such that for all $\lambda\in [0, \lambda^*]$, $(\tau^*, p+\lambda d)$ is partially stable, $p+\lambda d$ remains within the same linear domain, and an alternating tree $\mathcal T'$ with the same unit-buyers, items, and root $k_0$ as $\mathcal T$ is preserved in $D(\tau^*, p+\lambda d)$. Furthermore, if $p+\lambda^* d$ does not reach the border of the current linear domain, the number of items in the MAT rooted at $k_0$ increases by 1. The algorithm runs in time $O(|B|^5|S|^5)$.
\end{lemma}

\section{The Main Result}\label{sec:main_result}

In this section, we complete the proof of the main result. In Section \ref{sec:main_result:findmat}, we analyze \textsc{FindMAT}, the subroutine for computing a MAT. In Section \ref{sec:main_result:augment}, we analyze \textsc{AugmentingPath}, the subroutine for increasing the size of the matching. Finally, in Section \ref{sec:main_result:main_result}, we prove Theorem \ref{thm:main}.

\subsection{The FindMAT subroutine}\label{sec:main_result:findmat}

Given an unmatched unit-buyer $k_0$ that demands at least one item, \textsc{FindMAT}, formally described as Algorithm~\ref{alg:mat} below, finds a MAT $\mathcal T$ rooted at $k_0$. The process is essentially a breadth-first search on the edges of the marginal demand graph. Initially, $\mathcal T$ is a tree with a single vertex, $k_0$. For each item $j \in F_{\tau, p}(k_0)$ demanded by $k_0$, the edge $(k_0, j)$ is added to $\mathcal T$ and, for all buyers  $k \in \tau(j)$ matched to $j$,  the edges $(k, j)$ are also added to $\mathcal T$ (line \ref{alglin:mat:addedges}). The algorithm then repeats this process with the newly added buyers $k$ to $\mathcal T$, but without connecting to a vertex that is already in $\mathcal T$.

The proof of next result is deferred to Appendix~\ref{sec:proofs:main_result}.

\begin{algorithm}[]
    \caption{\textsc{FindMAT}}
    \label{alg:mat}
    \begin{algorithmic}[1]
        \Require Graph $H$, matching $\tau\subseteq H$, and unmatched unit-buyer $k_0$ with an edge in $H$
        \Ensure MAT $\mathcal T$ in $H$ rooted at $k_0$
        \State Initialize $\mathcal T\gets\emptyset$,\ $U' \gets \{k_0\}$ \label{alglin:mat:init}
        \While{$\exists k \in U'$\label{alglin:mat:while}}
            \For{$j\not\in V(\mathcal T)$ s.t. $(k,j)\in H$\label{alglin:mat:j}}
                \State $U'' \leftarrow \tau(j)\setminus V(\mathcal T)$ \label{alglin:mat:u''}
                \State $\mathcal T\gets \mathcal T\cup \{(k,j)\} \cup \{(k',j)\mid k'\in U'' \}$ \label{alglin:mat:addedges}
                \State $U' \leftarrow U' \cup U'' \setminus \{k\}$ \label{alglin:mat:updateu}
            \EndFor
        \EndWhile
        \State \Return $\mathcal T$
    \end{algorithmic}
\end{algorithm}

\begin{lemma}
\label{lem:existmat}
For any one-to-one outcome $(\tau, p)$ and unmatched unit-buyer $k_0$ that demands at least one item,  Algorithm \ref{alg:mat} finds a MAT rooted at $k_0$ in $O(|B||S|^2)$ running time.
\end{lemma}

\subsection{The AugmentingPath subroutine}\label{sec:main_result:augment}

The final subroutine is \textsc{AugmentingPath}, called in line \ref{alglin:gs_core_mech:augment1} of Algorithm \ref{alg:gs_core_mech}. The subroutine finds the shortest augmenting path from the root $k_0$ of a MAT $\mathcal T$ to a leaf $j$ that is undermatched. 

Note that if an augmenting path $P$ contains two unit-buyer copies $k_1, k_2$ of the same buyer $i$, then it is possible that after augmenting the current matching $\tau$ along $P$, buyer $i$ is no longer matched to a bundle maximizing $u_i(\cdot, p)$ among bundles of size $|\tau(\text{ubuy}(i))|$. This possibility is due to the non-linear nature of gross substitutes valuations, so that by swapping one item, the marginal value of another item may change in an unpredictable manner. We can show that by selecting the shortest such path to augment, we can preserve partial stability while increasing the size of the matching. The main idea of the proof is to use Lemma \ref{lem:alt_cycles} to find shortcuts in the augmenting path; if no shortcuts exist, then we show the augmenting path must preserve partial stability.

\begin{lemma}[Augmenting preserves partial stability]\label{lem:gs_aug_preserves_stability}
    Let $(\tau, p)$ be a partially stable outcome. Let $\mathcal T$ be a MAT rooted at unmatched unit-buyer $k_0$. Let $P$ be the shortest alternating path from $k_0$ to any item $j_0$ in $\mathcal T$ that is not matched, and let $\tau\triangle P$ be the matching obtained by augmenting $\tau$ along $P$. Then, $\tau\triangle P$ is partially stable at $p$.
\end{lemma}
\begin{proof}{Proof of Lemma \ref{lem:gs_aug_preserves_stability}}
    We wish to show that after an augmentation, each buyer $i$ is matched to a bundle that optimizes $u_i(\cdot, p)$ among bundles of size $|\tau\triangle P(i)|$. Let $(k_0, j_1)$ be the first edge in $P$. Add an auxiliary item $j_a$ with price $p_{j_a} = 0$, and edit the valuation function of $i_0 = \text{buy}(k_0)$ to be $v_{i_0}'$, defined by $v'_{i_0}(T) = v_{i_0}(T)$ if $j_a\not\in T$, and $v'_{i_0}(T) = v_{i_0}(T\setminus \{j_a\}) + u_{i_0}(j_1\mid \tau(\text{ubuy}(i_0)), p)$ if $j_a\in T$. Also define the corresponding utility function $u'_{i_0}(T, p) = v'_{i_0}(T) - \sum_{j\in T}q_{i_0j}(p_j)$. For all $T\subseteq S$ such that $j_a\not\in T$, we have $v'_{i_0}(j_a\mid T) = u'_{i_0}(j_a\mid T, p) = u_{i_0}(j_1\mid \tau(\text{ubuy}({i_0})), p)$. In other words, the marginal value and marginal utility of $j_a$ under $v_{i_0}'$ and $u_{i_0}'$ is always equal to $u_{i_0}(j_1\mid \tau(\text{ubuy}({i_0})), p)$. To see that $v_{i_0}'$ still satisfies gross substitutes, notice that $v'_{i_0}(T)$ is the sum of a gross substitutes and an additive function, which is gross substitutes as well (see, e.g., \cite{de2020geometry}).
    
    Next, edit the current matching to become $\tau' = \tau\cup\{(k_0, j_a)\}$ and the augmenting path $P$ to become the alternating path $P' = P\cup\{(k_0, j_a)\}$. This step is necessary to ensure that every buyer with a unit-buyer in $P$ has the same number of matches before and after augmentation. Before augmenting along $P'$, every buyer $i\ne i_0$ still has the bundle $\tau'(i) = \tau(i)$, which maximizes $u_i(\cdot, p)$ among bundles of size $|\tau'(i)|$. For buyer $i_0$, the bundle $\tau'(i_0) = \tau(i_0)\cup\{j_a\}$ also maximizes $u_{i_0}(\cdot, p)$ among bundles of size $|\tau'(i_0)|$. To see this, we can write
    $$u'_{i_0}(\tau'(i_0), p) = u'_{i_0}(\tau(i_0)\cup\{j_a\}, p) = u'_{i_0}(\tau(i_0), p) + u'_{i_0}(j_a\mid \tau(i_0), p)$$
    and thus
    $$ u'_{i_0}(\tau'(i_0), p)= u'_{i_0}(\tau(i_0), p) + u'_{i_0}(j_1\mid \tau(i_0), p) = u'_{i_0}(\tau(i_0)\cup\{j_1\}, p).$$
    The last term is known to maximize $u_{i_0}(\cdot, p)$ among bundles of size $|\tau'(i_0)|$, by definition of the marginal demand graph. It follows that $\tau'$ is partially stable in the new instance.

    Now we show that the matching remains partially stable in $\tau'\triangle P'$, after augmentation. Consider any buyer $i$ with a single unit-buyer copy $k$ in $P'$. Then, by definition of the marginal demand graph, $u_i(\tau(i), p) = u_i((\tau\triangle P)(i), p)$, so $i$ has optimal utility among bundles of size $|\tau(i)|$.

    Now, consider any buyer $i$ with multiple unit-buyer copies $k_1,\ldots, k_z$ in $P'$, appearing from first to last in that order. We can then apply Lemma \ref{lem:alt_cycles} to this sequence, with matched edges $(k_\ell, j_\ell)\in P'\cap\tau'$ and unmatched edges $(k_{\ell}, j_\ell')\in P'\setminus\tau'$. Note that if $i = i_0$, our use of $j_a, P',$ and $\tau'$ is what allows these edges to exist. Assume for contradiction that $i$ has strictly less utility $u_i(\cdot, p)$ after augmentation along $P'$, so that $\tau'(i)\cup\{j_1',\ldots, j_z'\}\setminus\{j_1,\ldots, j_z\}$ does not maximize $u_i(\cdot, p)$ among bundles of size $|\tau'(i)|$. Then Lemma \ref{lem:alt_cycles} (a) is not true, and so Lemma \ref{lem:alt_cycles} (b) must be true instead. Thus, there is a sequence of distinct indices $x_1,\ldots, x_s$ such that $(k_{x_1}, j'_{x_2}, k_{x_2}, j'_{x_3}, k_{x_3}, \dots, k_{x_s}, j'_{x_1})$ is a cycle in $D(\tau,p)\setminus \tau$. Without loss of generality, let $x_1$ be the smallest index in $x_1,\cdots, x_s$. Then there is an edge $(k_{x_1}, j'_{x_2})$ in $D(\tau, p)$. Thus, there is an alternating path from $k_0$ to $j_0$ that is shorter than $P$, given by following $P$ from $k_0$ to $k_{x_1}$, then going to $j'_{x_2}$, then following $P$ to $j_0$. This path contradicts our selection of $P$ as the shortest such path. We conclude that each buyer $i$ has the same utility after augmentation, which means that $\tau'\triangle P' = \tau\triangle P$ is also partially stable.
 \end{proof}

We can use any existing algorithm for computing a shortest augmenting path, such as that described in \citet{schrijver2003combinatorial}. This result is  well-known , so the proof is omitted.

\begin{lemma}[Shortest augmenting path algorithm]\label{lem:aug_path_alg} \textsc{AugmentingPath} takes in input a root node $k_0$, a matching $\tau$, and a graph $D(\tau, p)$, and in $O(|B||S|^2)$ time outputs a shortest augmenting path.
\end{lemma}

\subsection{Main result}\label{sec:main_result:main_result}

We can now put everything together to establish the correctness and runtime of Algorithm \ref{alg:gs_core_mech}.

\begingroup
\def\thetheorem{\ref{thm:main}}
\begin{theorem}
    In the \textsc{QITU} model, there is a poly-time algorithm that finds a competitive equilibrium.
\end{theorem}
\addtocounter{theorem}{-1}
\endgroup

\begin{proof}{Proof}
    First, we show that the outcome $(\tau, p)$ maintained by Algorithm \ref{alg:gs_core_mech} is partially stable. We begin with $(\tau, p) = (\emptyset, 0)$, which is partially stable. In each iteration of the outer while loop from lines \ref{alglin:gs_core_mech:outer_while}-\ref{alglin:gs_core_mech:outer_endwhile}, we know by Lemma \ref{lem:existmat} that line \ref{alglin:gs_core_mech:findmat1} correctly computes a MAT $\mathcal T$. If all items in the MAT are matched, the algorithm executes the inner while loop from lines \ref{alglin:gs_core_mech:inner_while}-\ref{alglin:gs_core_mech:inner_endwhile}. In each iteration of the inner while loop, we know by Lemma \ref{lem:gs_price_inc} that lines \ref{alglin:gs_core_mech:findpriceinc1}-\ref{alglin:gs_core_mech:findmat2} updates $(\tau, p)$ while preserving partial stability and weakly increasing the number of vertices of $\mathcal T$. By Lemma \ref{lem:gs_aug_preserves_stability}, augmenting $\tau$ in lines \ref{alglin:gs_core_mech:augment1}-\ref{alglin:gs_core_mech:augment2} also preserves partial stability. We see that $(\tau, p)$ remains partially stable throughout.

    Then, the algorithm only terminates if all unit-buyers $k_0$ are matched, and so no MATs exist in $D(\tau, p)$. By Lemma \ref{lem:smo_core_cm}, when combined with partial stability, this implies the output $(\mu, p)$ is a competitive equilibrium. Thus, if the algorithm terminates, it returns a competitive equilibrium.

    Now we  analyze the runtime of Algorithm \ref{alg:gs_core_mech}. Consider each iteration of the inner while loop. By Lemma \ref{lem:gs_price_inc}, either the number of items in $\mathcal T$ increases, or we reach the edge of the current linear domain, and $\textsc{LinSeg}_{i,j}(p_j)$ increases by one for some $(i,j)$. $\mathcal T$ can increase at most $O(|S|)$ times, and we can increment linear domains at most $O(K)$ times, where $K$ is the total number of linear segments in the effective price functions. Each iteration is dominated by the runtime of \textsc{FindPriceIncrease}, which requires $O(|B|^5|S|^5)$ time by Lemma \ref{lem:gs_price_inc}. The total runtime of the inner while loop is thus $O((K+|S|)|B|^5|S|^5)$. In each iteration of the outer while loop, an unmatched unit-buyer $k_0$ is selected in line \ref{alglin:gs_core_mech:outer_while}, then matched along an augmenting path in line \ref{alglin:gs_core_mech:augment2}. Thus, the number of unmatched unit-buyers decreases by one, for at most $O(|B||S|)$ iterations. Each iteration is dominated by the runtime of the inner while loop, for a total runtime of $O((K+|S|)|B|^6|S|^6)$.
 \end{proof}

\section{Structure and Buyer-Optimality}\label{sec:lattice_IC}

We now study the structure of competitive equilibria in the QITU model, as well as properties of Algorithm~\ref{alg:gs_core_mech}. We start by defining a partial order over equilibrium prices and recalling the known fact that under this partial order the set of equilibrium prices forms a lattice.

\begin{definition}[Join and meet]
    Let $p, p'\in \mathbb R^S$ be prices. Their \emph{join} $\overline p = p\vee p'$ is given by $\overline p_j = \max\{p_j, p'_j\}$ for each $j\in S$. Their \emph{meet} $\underline p = p\wedge p'$ is given by $\underline p_j = \min\{p_j, p'_j\}$ for each $j\in S$.
\end{definition}

\begin{definition}[Partial order over equilibria]
    Let $(\mu, p)$ and $(\mu', p')$ be competitive equilibria. The partial order $(\le)$ over competitive equilibria is defined such that $(\mu, p)\le (\mu', p')$ if and only if $p\le p'$ and $u_i(\mu(i), p) \ge u_i(\mu'(i), p')$ for all $i\in B$.
\end{definition}

A result of \citet{schlegel2022structure} implies that the set of competitive equilibria forms a lattice.

\begin{theorem}[\citet{schlegel2022structure}]\label{thm:ce_lattice}
    Let $(\mu, p)$ and $(\mu', p')$ be competitive equilibria. There exist matchings $\overline\mu, \underline\mu$ such that $(\overline\mu, \overline p)$ and $(\underline\mu, \underline p)$ are also competitive equilibria, $(\underline\mu, \underline p)\le (\mu, p)\le (\overline\mu, \overline p)$, and $(\underline\mu, \underline p)\le (\mu', p)\le (\overline\mu, \overline p)$. Thus, the set of competitive equilibrium prices form a complete bounded lattice under the partial order $(\leq)$, and thus has a minimum and maximum element.
\end{theorem}

As we show in  Appendix~\ref{sec:proofs:lattice_IC}, Algorithm~\ref{alg:gs_core_mech} outputs a buyer-optimal equilibrium. 

\begin{theorem}[Algorithm~\ref{alg:gs_core_mech} returns the buyer-optimal competitive equilibrium]\label{thm:min_ce}
    Algorithm~\ref{alg:gs_core_mech} returns the minimum competitive equilibrium.
\end{theorem}
Finally, it has been shown that in many stable matching settings, there exist buyer-optimal mechanisms that are incentive compatible for buyers \citep{dubins1981machiavelli}.
\begin{definition}[Incentive Compatibility]
    A mechanism $M$ for the QITU model is incentive compatible if,
    for any buyer $i \in B$, effective price functions $q_i$, true valuation functions $\{v_{i'}\}_{i' \in B}$,
    and misreported valuation function $\tilde v_i \neq v_i$ and effective price functions $\tilde q_i\ne q_i$,
    $i$ cannot strictly improve their utility by misreporting  $\tilde v_{i}, \tilde q_i$:
    $$v_i(\tilde \mu(i)) - \sum_{j \in \tilde \mu(i)} q_{ij}(\tilde p_j) \leq v_i(\mu(i)) - \sum_{j \in \mu(i)} q_{ij}(p_j),$$
    where $(\tilde \mu, \tilde p)$ and $(\mu, p)$ are the outcomes of $M$ when the inputs are $\{(v_1, q_1), \ldots, (v_{i-1}, q_{i-1}),$ $(\tilde v_i, \tilde q_i),$ $(v_{i+1}, q_{i+1}), \ldots, (v_n, q_n)\}$ and $\{(v_1, q_1), \ldots, (v_{i-1}, q_{i-1}),(v_i, q_i),$ $(v_{i+1}, q_{i+1}), \ldots, (v_n, q_n)\}$.
\end{definition}

\citet{holmstrom1979groves} showed that in a more general setting, any incentive compatible mechanism has payments equivalent to the VCG mechanism, up to a translation. \citet{gul2000english} showed that in a special case of the QITU model, no competitive equilibrium mechanism has payments equivalent to the VCG mechanism. As such, it is known that no competitive equilibrium mechanism is incentive compatible in the QITU model.

\begin{theorem}[\citet{holmstrom1979groves}, \citet{gul2000english}]\label{thm:non_IC}
    No competitive equilibrium mechanism in the QITU model is incentive compatible, even when buyers have perfectly transferable utility, i.e. $q_{ij}(p_j) = p_j$ for all $i\in B, j\in S$, and the effective price functions are public information known to the mechanism designer.
\end{theorem}

\noindent Note that Theorem~\ref{thm:non_IC} does not rule out the existence of incentive compatible mechanisms which do not return a competitive equilibrium. Incentive compatibility is desirable because it simplifies strategic considerations on participants and may lead to more predictable outcomes for designers. In many settings, though,  mechanisms that are not incentive compatible are preferred to existing incentive compatible mechanisms due to efficiency or other properties (see, e.g., \citep{kesten2010school,despotakis2021first}). For example, repeated game settings such as online ad auctions have eschewed second price auctions for non-incentive compatible mechanisms \citep{despotakis2021first}.

\section{Hardness of Non-Quasilinear Utility Settings}\label{sec:negative}

We consider two natural variations of our model, which we call the \textsc{NQBr} model (Non-Quasilinear with Breakpoint) and \textsc{NQBu} (Non-Quasilinear with Budget). Each model allows a single buyer to have non-quasilinear utility. In the \textsc{NQBu} model, a single buyer has a hard budget constraint. 

\begin{definition}\label{def:nqbu}
    In the \textsc{NQBu} model, there is one buyer with utility given by
        $$u_{i_0}(T, p) = v_{i_0}(T) - r_{i_0}\left(\sum_{j\in T}p_j\right), \; \hbox{ where } \; r_{i_0}(p) = \begin{cases}
            p & p\le C \\
            \infty & p > C
        \end{cases}$$
     for some budget $C$ and some gross substitutes valuation function $v_{i_0}$. For $i\ne i_0$, we have $u_i(T) = v_i(T) - \sum_{j \in T} p_j$, where $v_i$ is a unit demand valuation function, a special case of gross substitutes.
\end{definition}

The assumption of unit demand quasilinear utility for the other agents is more restrictive than the one from the \textsc{QITU} model. Yet, even with this assumption, we can show a hardness result. 

\begingroup
\def\thetheorem{\ref{thm:hardness2}}
\begin{theorem}
    Computing a competitive equilibrium in the \textsc{NQBu} model is NP-hard, even when an equilibrium is known to exist.
\end{theorem}
\addtocounter{theorem}{-1}
\endgroup

The \textsc{NQBr} model is similar; except for the unique buyer, the ``value-for-money'' function $r_{i_0}$ is a piecewise linear continuous function with one breakpoint. 

\begin{definition}\label{def:nqbr}
    In the \textsc{NQBr} model, there is one buyer with utility given by
        $$u_{i_0}(T, p) = v_{i_0}(T) - r_{i_0}\left(\sum_{j\in T}p_j\right), \; \hbox{ where } \; r_{i_0}(p) = \begin{cases}
            p & p\le C \\
            C+ \alpha(p-C) & p > C
        \end{cases}$$
     for some scalars $C$ and $\alpha$, and some gross substitutes valuation function $v_{i_0}$. For $i\ne i_0$, we have $u_i(T) = v_i(T) - \sum_{j \in T} p_j$, where $v_i$ is a unit demand valuation function.
\end{definition}

We show that computing a competitive equilibrium is also NP-hard in the \textsc{NQBr} model.

\begingroup
\def\thetheorem{\ref{thm:hardness}}
\begin{theorem}
    Computing a competitive equilibrium in the \textsc{NQBr} model is NP-hard, even when an equilibrium is known to exist.
\end{theorem}
\addtocounter{theorem}{-1}
\endgroup

Proofs from this section are postponed to Appendix \ref{sec:proofs:negative}. In the proofs, we reduce the integer knapsack problem to a specific type of instance within the \textsc{NQBu} or \textsc{NQBr} model. 

While it has been shown that equilibria may not exist in the \textsc{NQBu} setting \citep{mongell1986note, jagadeesan2021matching}, an algorithm that finds an equilibrium or indicates that no equilibrium exists in polynomial time is still of independent interest. Our construction implies that,  even when restricting to domains for which equilibria are known to exist, computing equilibria is NP-hard in both the \textsc{NQBu} and the \textsc{NQBr} model. 

We also note that with non-separable utilities, the equivalence of gross substitutes definitions shown in Lemma~\ref{lem:gs_equals_gskc} no longer holds. See Theorem~\ref{thm:nonsep} in Appendix~\ref{sec:r-gross} for the complete result.

\section{Conclusion}\label{sec:conclusion}

The main result in this paper is a polynomial-time mechanism for computing a bidder-optimal competitive equilibrium in a setting with both gross substitutes and imperfectly transferable utility. This result expands the known range of scenarios in which finding a competitive equilibrium is tractable. In addition, it unifies prior results for the setting with gross substitutes and the setting with imperfectly transferable utility. Our techniques combine ideas from matroid theory with more classical augmenting path and t\^atonnement approaches. To the best of our knowledge our paper is the first that employs algorithms from matroid theory for the computation of equilibria.

A natural question is whether this result can be generalized to an even broader setting. We provide two negative results that highlight extensions where computational hardness arises. These settings capture two relaxations of the separable transfers assumption of our model, where we allow either a hard budget constraint or a single breakpoint in the utility function of a single buyer.

\section*{Acknowledgements}

Christopher En and Yuri Faenza acknowledge the support of the NSF Grant 2046146 \emph{CAREER: An Algorithmic Theory of Matching Markets}. Eric Balkanski acknowledges the support of NSF grants  CCF-2210501 and IIS-2147361.

\newpage
\nocite{*}
\bibliographystyle{apalike}
\bibliography{main}

\newpage
\appendix

\section{Additional Applications}\label{sec:proofs:applications}

\paragraph{Online ad marketplaces.} Consider an online advertising setting where advertisers wish to buy impressions in various online platforms, and different platforms allow for varying levels of user engagement. As described in \citet{dutting2011expressive}, an advertiser may have per-click valuations; that is, they have some value $v_i(c)$ for the number of times $c$ a user clicks on their advertisement on any platform. This value may also have diminishing marginal returns as $c$ increases, motivating the use of gross substitutes valuations. Content platforms such as TikTok or Youtube may charge on a per-impression basis \citep{understand2024}; that is, they charge a price $p_j$ for each time the advertisement is shown to a user, regardless of whether or not a user clicks on the ad. The advertiser's utility then depends on clicks $(c_{ij})$, which are derived from impressions $(d_{ij})$ via a known click-through rate $r_{ij}$ (where $c_{ij} = r_{ij}d_{ij}$), which can vary by platform and advertiser. Thus, the utility of an advertiser $i$ is given by
\[v_i\left(\sum_{j}c_{ij}\right) - \sum_{j}d_{ij}p_j = v_i\left(\sum_j c_{ij}\right) - \sum_j \frac{c_{ij}p_{j}}{r_{ij}}.\]
As discussed in \citet{dutting2009bidder}, the difference between the per-impression price and the per-click utility for the advertiser results in ITU. The reason is that, with a per-impression price $p_j$ on platform $j$, the effective per-click price  for advertiser $i$ is $p_j/r_{ij}$, which can vary by advertiser and platform. In a competitive auction for advertising impressions, some advertisers may thus be able to pay a lower effective per-click price than others.

\paragraph{Departmental budgets.} Suppose a firm has several teams or departments. This setting could be used to describe teams at a lobbying firm, pods at an investment firm, or separate R\&D research groups at a pharmaceutical company. Each team selects a project, such as a political issue to lobby for, an investment category, or a new type of drug to research. These projects can be subject to gross substitutes valuations; for example, two projects involving different experimental cancer treatments may have some substitutabilities. Additionally, each team has a soft budget, as described in \citet{kornai1986soft}; spending beyond the soft cap incurs some additional penalties to the team, such as reducing its budget for the next year. This soft budget can be modeled using ITU. For example, a team $i$ might have a \$1,000,000 soft budget, but each additional dollar spent effectively costs them \$1.30. In this case, the effective cost of expenditures is given by the piecewise linear function
\[q_i(p) = \begin{cases}
  p & p\le 10^6 \\
  10^6 + 1.3(p-10^6) & p>10^6
\end{cases}\]

\section{Remarks on the Model}\label{sec:proofs:prelims}

In this section, we make several notes on the choice of model, compared to previous assignment game models.

\begin{remark}\label{rem:sellers_items}
    Note that previous one-to-one models such as those presented by \citet{shapley1971assignment} and \cite{demange1985strategy} have ``sellers'' instead of ``items,'' where a seller $j$ has some nonnegative reserve value $s_j$ for their endowed item. There is a reduction from the problem of finding a competitive equilibrum in the setting with sellers to the problem of finding a competitive equilibrium in the setting with items. To see this, given an instance of a problem with sellers, consider the following edited instance. Set each seller's reserve price to 0, and let each buyer $i$ with value $v_{ij}$ have a new value $v_{ij} - q_{ij}(s_j)$. Define the new effective price function to be  $q^*_{ij}(p_j) = q_{ij}(p_j + s_j) - q_{ij}(s_j)$. Then, given an equilibrium $(\mu, p)$ in the original instance, we obtain an equilibrium in the edited instance given by $(\mu, p')$ where $p'_j = p_j - s_j$. Similarly, given any equilibrium $(\mu, p')$ in the edited instance, we obtain an equilibrium in the original instance given by $(\mu, p)$, where $p_j = p'_j + s_j$.
\end{remark} 
    
\begin{remark}\label{rem:unit_cap}
    Our model is also equivalent to the setting where an item may have capacity greater than 1, but each buyer can only match to each item once. If an item $j$ has capacity $\ell$, then we can make $\ell$ copies $j_1,\cdots, j_\ell$ of $j$ with unit capacity. Then, each buyer $i$'s valuation function becomes
    \[\tilde v_i(T) = \begin{cases}
        v_i(T) \qquad\qquad\quad j_1,\cdots, j_\ell\not\in T \\
        v_i(T\cup\{j\}\setminus\{j_1,\cdots, j_\ell\}) \quad \text{o.w.}
    \end{cases}\]
    The new valuation function satisfies gross substitutes if the original valuation function satisfies gross substitutes, and equilibria in the two settings coincide.
\end{remark}

We now prove Lemma~\ref{lem:gs_equals_gskc} regarding the equivalence of ITU-gross substitutes and the standard definition of gross substitutes.

\begin{proof}{Proof of Lemma~\ref{lem:gs_equals_gskc}}
    First, let $v_i$ satisfy ITU-gross substitutes. Fix any effective price functions $q = \{q_{ij}\}_{i\in B, j \in S}$ that are  continuous, strictly increasing, and surjective onto $\mathbb{R}$ and price vectors $p\le p'$. Suppose $T\in \mathring F_{p}(i)$. Define $r, r'\in \mathbb R^S$ by $r_j = q_{ij}^{-1}(p_j)$ and $r'_j = q_{ij}^{-1}(p'_j)$ for all $j\in S$. By monotonicity, $r\le r'$ as well. By definition, $F_{r, q}(i) = \mathring F_{p}(i)$ and so $T\in F_{r,q}(i)$. By ITU-gross substitutes, there exists $T'\in F_{r', q}(i)$ such that $T\cap \{j\mid r_j = r'_j\}\subseteq T'$. Since $r_j = r_j'$ if and only if $p_j = p_j'$, and $F_{r', q}(i) = \mathring F_{p'}(i)$, we see that there exists $T'\in  \mathring F_{p'}(i)$ such that $T\cap \{j\mid p_j = p'_j\}\subseteq T'$. This shows that $v_i$ satisfies gross substitutes.
    
    Next, assume $v_i$ satisfies gross substitutes. Fix any price vectors $r, r'$ with $r\le r'$, and suppose $T\in F_{r, q}(i)$. Set $p_j = q_{ij}(r_j)$ and $p'_j = q_{ij}(r'_j)$ for each $j\in S$. By monotonicity, $p\le p'$ as well. By definition, $\mathring F_p(i) = F_{r,q}(i)$, and so $T\in \mathring F_p(i)$. By gross substitutes, there exists $T'\in \mathring F_{p'}(i)$ such that $T\cap \{j\mid p_j = p'_j\}\subseteq T'$. By definition, $F_{r',q}(i) = \mathring F_{p'}(i)$ and $T\cap \{j\mid r_j = r'_j\} = T\cap \{j\mid p_j = p'_j\}$. So, $T'\in F_{r',q}(i)$ and $T\cap \{j\mid r_j = r'_j\}\subseteq T'$, and thus $v_i$ satisfies ITU-gross substitutes as well.
    
\end{proof}

\section{Properties of Gross Substitutes}\label{sec:proofs:gs}

We present some helper results we will need for the next two sections. Throughout the section, we assume that we are given a set of items $S$, a valuation function $v: 2^S\to \mathbb R$, a piecewise linear, continuous, strictly increasing, surjective onto $\mathbb{R}$  effective price function $q_j: \mathbb R\to\mathbb R$ and such that $q(0)=0$, and a utility function $u:2^S\times \mathbb R\to \mathbb R$, given by $u(T, p) = v(T) - \sum_{j\in T}q_j(p_j)$. Note that these are exactly the assumptions of the QITU model.

Next, we define a greedy procedure for selecting a bundle of items. Essentially, the process iteratively selects the item with greatest marginal utility, until no more items can be profitably selected.

\begin{definition}[Greedy procedure]\label{def:greedy_alg}
    Given a price vector $p$ and $\ell\in \mathbb N$, the \emph{greedy procedure} selects a set of items $T$ of size $\ell$ by the following process:
    \begin{enumerate}
        \item Set $T = \emptyset$.
        \item Let $j\in S\setminus T$ be the item not in $T$ that maximizes $u(j\mid T, p)$; break ties arbitrarily.\label{def:greedy:select}
        \item If $|T|< \ell$, set $T\gets T\cup \{j\}$, and return to step \ref{def:greedy:select}. Else, return $T$.
    \end{enumerate}
\end{definition}

Now, several equivalent definitions of gross substitutes will be useful.

\begin{proposition}\label{prop:gs_equivalent}
    The following are equivalent:
    \begin{enumerate}
        \item (GS) $v_i$ satisfies gross substitutes.
        \item (SI) Given any price vector $p$ and any $T\subseteq S$, if $\mathring u_i(T, p)<max_{T'\subseteq S}\mathring u_i(T', p)$, then there exist $X,Y$ with $|X|,|Y|\le 1$ such that $\mathring u_i(T\cup X\setminus Y, p)> \mathring u_i(T, p)$ \citep{gul2000english}.
        \item (GR) Given any price vector $p$ and any $\ell\in \mathbb N$, the greedy procedure finds an optimal solution to $max_{T\subseteq S:|T|=\ell}\mathring u_i(T, p)$. Furthermore, every optimal solution can be found with a suitable tiebreaking rule. \citep{dress1995well}.
        \item (WL) $v$ is submodular, and for any $p$ and any positive integer $\ell$, the greedy procedure finds an optimal solution to $\max_{T\subseteq S, |T|=\ell}\mathring u_i(T, p)$. This property is known as the well-layered property \citep{dress1995well}.
        \item (ISO) $v$ is submodular, and for all sets $T$ and distinct items $j_1, j_2, j_3\not\in T$,
        $$v_i(j_1,j_2\mid T) + v_i(j_3\mid T)\le\max\bigg[v_i(j_1,j_3\mid T) + v_i(j_2\mid T),v_i(j_2,j_3\mid T) + v_i(j_1\mid T)\bigg]$$
        \citep{reijnierse2002verifying}

        \item ($M^\natural$) For all $T, T'\subseteq S$ and $j\in T\setminus T'$,
        $$v(T) + v(T')\le \max\left[v(T\setminus \{j\}) + v(T'\cup\{j\}), \max_{j'\in T'\setminus T}v(T\cup\{j'\}\setminus \{j\}) + v(T'\cup\{j\}\setminus \{j'\})\right]\label{eq:m_natural}$$
        \citep{fujishige2003note}

        \item (DT) $v_i$ is submodular, and for any price vector $p$ and any triple disjoint of sets $\{j\}, T, T'$ with $|T'|\ge 2$, 
        $$\mathring u_i(\{j\}\mid T, p) + \mathring u_i(T'\mid T, p)\le \max_{j'\in T'}\mathring u_i(j'\mid T, p) + \mathring u_i(T'\cup\{j\}\setminus\{j'\}\mid T, p)$$
        \citep{dress1995well}

    \end{enumerate}
\end{proposition}

We can also show that $u_i$ satisfies each of the properties (SI), (GR), (WL), and (DT).

\begin{lemma}\label{lem:gs_si_gr_wl}
    If $v_i$ satisfies gross substitutes, then $u_i$ satisfies (SI), (GR), (WL), and (DT).
\end{lemma}
\begin{proof}{Proof}
    Define the price vector $r$ by setting $r_j = q_{ij}(p_j)$ for each $j\in S$. Then $u_i(\cdot, p) = \mathring u_i(\cdot, r)$. (SI), (GR), (WL), and (DT) on $\mathring u_i(\cdot, r)$ then immediately imply (SI), (GR), (WL), and (DT) on $u_i(\cdot, p)$.
    
\end{proof}

We can now prove Lemma~\ref{lem:ce_fsbl_stable}.

\begin{proof}{Proof of Lemma~\ref{lem:ce_fsbl_stable}}
    First, let $(\mu, p)$ be a competitive equilibrium. By definition, $(\mu, p)$ is feasible. If there was a profitable addition, swap, or drop for buyer $i$ from the initial bundle $\mu(i)$, then the resulting bundle would have greater utility, contradicting the competitive equilibrium. Thus, $(\mu, p)$ is stable.

    Let $(\mu, p)$ be feasible and stable. We wish to show $\mu(i)\in F_p(i)$, which is equivalent to showing that $u_i(\mu(i), p)\ge u_i(T, p)$ for all $T\subseteq S$. Suppose for contradiction that there existed $T$ such that $u_i(\mu(i), p) < u_i(T, p)$. Then, by (SI), there exists $X,Y\subseteq S$ with $|X|,|Y|\le 1$ such that $u_i(\mu(i)\cup X\setminus Y, p) > u_i(\mu(i), p)$. If $|X|=|Y|=1$, this corresponds to a profitable swap. If $|X| = 1$ and $|Y|=0$, this corresponds to a profitable addition. If $|X| = 0$ and $|Y|=1$, this corresponds to a profitable drop, contradicting stability. We conclude that $T$ cannot exist, so $(\mu, p)$ is a competitive equilibrium.

\end{proof}

It is also well-known that gross substitutes is preserved when an additional ``capacity constraint'' is added. The following lemma is folklore, but we include a proof for completeness.

\begin{lemma}[Capacity constraint preserves gross substitutes]\label{lem:gs_capacity}
    Let $v_i:2^S\to\mathbb R$ satisfy gross substitutes. Then, for any fixed constant $\ell$, the function $\tilde v_i:2^S\to\mathbb R$ given by
    \[\tilde v_i(T) = \begin{cases}
    v_i(T) & |T|\le \ell\\
    v_i(T) - M(|T|-\ell) & \text{o.w.}
    \end{cases}\]
    where $M$ is a constant also satisfies gross substitutes. Similarly, $\overline v_i: 2^S\to\mathbb R$ given by
    \[\overline v_i(T) = \begin{cases}
    v_i(T)+M|T| & |T|\le \ell\\
    v_i(T) + M\ell - M(|T|-\ell) & \text{o.w.}
    \end{cases}\]
    satisfies gross substitutes.
\end{lemma}

\begin{proof}{Proof of Lemma \ref{lem:gs_capacity}}
    We will show that (WL) holds for $\tilde v_i$. For submodularity, consider a set $T\subseteq S$ and items $j_1,j_2\not\in T$. If $|T|\le \ell-1$ we have by submodularity of $v_i$ that
    $$\tilde v_i(T\cup\{j_1\}) + \tilde v_i(T\cup\{j_2\})= v_i(T\cup\{j_1\}) +  v_i(T\cup\{j_2\})\ge v_i(T\cup\{j_1,j_2\}) + v_i(T)\ge \tilde v_i(T\cup\{j_1,j_2\}) + \tilde v_i(T).$$
    If $|T|\ge \ell$ then
    $$\tilde v_i(T\cup\{j_1\}) + \tilde v_i(T\cup\{j_2\})= v_i(T\cup\{j_1\}) +  v_i(T\cup\{j_2\}) - 2M(|T|+1-\ell)$$
    $$\ge v_i(T\cup\{j_1,j_2\}) + v_i(T) - 2M(|T|+1-\ell)= \tilde v_i(T\cup\{j_1,j_2\}) + \tilde v_i(T).$$
    Thus, $\tilde v_i$ is submodular. Next, we will show that $\tilde v_i$ is well-layered. Since $v_i$ is well-layered, we know that for $\ell'\le \ell$ the greedy procedure finds an optimal solution to
    $\max_{T\subseteq S,\ |T|=\ell'}u_i(T,p) = \max_{T\subseteq S,\ |T|=\ell'}\tilde u_i(T,p).$ 
    For all $|T|> \ell$, by letting $\tilde u_i(\cdot, p) = \tilde v_i(\cdot) - \sum_{j\in T}q_j(p_j)$, we know that
    $\tilde u_i(T, p) = u_i(T, p) - M(|T|-\ell).$
    As a result, if $u_i(T, p)\ge u_i(T', p)$ for sets $T,T'$ of the same size, then $\tilde u_i(T, p)\ge \tilde u_i(T', p)$ as well. It follows that when $\ell'>\ell$ the greedy procedure finds an optimal solution to
    $\max_{T\subseteq S,\ |T|=\ell'}\tilde u_i(T,p)$
    as well. We conclude that (WL) holds for $\tilde v_i$, and so $\tilde v_i$ satisfies gross substitutes.

    The proof for $\overline v_i$ is identical, so we omit it for brevity.
 \end{proof}

Next, we have the matroid property of gross substitutes valuations, which shows that optimal bundles that also have a fixed size form the bases of a matroid. This proof is similar to a proof of \citet{gul2000english} that the optimal bundles of minimum size form the bases of a matroid; it also follows from ($M^\natural$).

\begin{lemma}[Matroid property]\label{lem:gs_matroid_property}
    Let $v_i: 2^S\to\mathbb R$ satisfy gross substitutes, $\ell \in \mathbb{N}$ with $\ell \leq |S|$ and $p$ any price vector. Then, the family $\arg\max_{T\subseteq S:|T|=\ell} u_i(T, p)$, i.e. the set of bundles maximizing $u_i(\cdot, p)$ among bundles of size $\ell$, form the bases of a matroid.
\end{lemma}

\begin{proof}{Proof}
    We show that the basis exchange property holds (see Lemma~\ref{lem:basis_exchange}) for all bundles maximizing $u_i(\cdot, p)$ among bundles of size $\ell$. Let $T_1,T_2$ be bundles maximizing $u_i(\cdot, p)$ among bundles of size $\ell$, such that $u_i(T_1, p) = u_i(T_2,p) = u^*$. Fix $j_1\in T_1\setminus T_2$. Define $\overline v$ as in Lemma \ref{lem:gs_capacity}, for a large constant $M$. Bundles maximizing $u_i(\cdot, p)$ among bundles of size $\ell$ also maximize $\overline u_i(\cdot, p)$, since for all $T\subseteq S$ with $|T|\geq \ell+1$, we have
    $$\begin{array}{lll} \overline v_i(T)& = & \overline v_i(T_1) + (\overline v_i(T)-\overline v_i(T_1)) \\ & = & \overline v_i(T_1) + v_i(T) + M\ell - M(|T|-\ell) - v_i(T_1) - M\ell \\ & = & \overline v_i(T_1) + v_i(T) - v_i(T_1) - M \underbrace{(|T|-\ell)}_{>0} \\ 
    & < & \overline  v_i(T_1)
    \end{array}
    $$
    for $M$ large enough, while for $T\subseteq S$ with $|T|\leq \ell -1$ and $M$ large enough, we have $$
    \overline v_i(T)= v_i(T) + M |T| < v_i(T_1) + M \ell = \overline v_i(T_1).
    $$ 
    
    Consider the price vector $p'$ given by $p'_j = p_j$ if $j\in T_1\cup T_2\setminus \{j_1\}$, and $p'_j = p_j + \epsilon$ otherwise. $T_1$ no longer maximizes $\overline u_i(\cdot, p)$, as it gives utility $\overline u_i(T_1, p') = u^*-\epsilon < u^* = \overline u_i(T_2,p)$. By (SI), there is a utility-improving local move. By construction of $\overline u_i$ if $M$ is large enough and any $\epsilon > 0$ small enough, we cannot profitably add or drop an item, so we must be able to profitably swap an item. There is then some $j_2\not\in T_1$ such that
    $\overline u_i(T_1, p') < \overline u_i(T_1\cup\{j_2\}\setminus \{j_1\}, p')$. 
    Suppose $j_2\not\in T_2$. Then $p'_{j_2} = p'_{j_2} + \epsilon$, which means that
    $$\overline u_i(T_1, p) - \epsilon = \overline u_i(T_1, p')< \overline u_i(T_1\cup\{j_2\}\setminus \{j_1\}, p')= \overline u_i(T_1\cup\{j_2\}\setminus \{j_1\}, p)-\epsilon,$$
    contradicting the fact that $T_1$ maximizes $\overline u_i(\cdot, p)$. Thus, $j_2\in T_2\setminus T_1$. Furthermore, we have that
    $$u_i(T_1, p)-\epsilon = u_i(T_1, p')< u_i(T_1\cup\{j_2\}\setminus \{j_1\}, p')= u_i(T_1\cup\{j_2\}\setminus \{j_1\}, p)$$
    for arbitrarily small $\epsilon > 0$. We see that it must be the case that $u_i(T_1\cup\{j_2\}\setminus \{j_1\}, p) = u_i(T_1, p) = u^*$. Thus, the exchange property holds, and we conclude that the bundles which maximize $u_i(\cdot, p)$ among bundles of size $\ell$ form the bases of a matroid.
 \end{proof}

\section{Proofs for Section \ref{sec:framework}}\label{sec:proofs:ce_cert}

In this section, we provide proofs omitted from Section \ref{sec:framework}. First, we prove Lemma~\ref{lem:partial_stability_alt_defn}.

\begin{proof}{Proof of Lemma~\ref{lem:partial_stability_alt_defn}}
    First, let $(\mu, p)$ be partially stable. Fix a buyer $i$. Then, for every $j_1\in\mu(i)$ and $j_2\not\in\mu(i)$, we have $u_i(\mu(i), p) \ge u_i(\mu(i)\cup\{j_2\}\setminus \{j_1\}, p)$ by definition of partial stability. We also have $u_i(\mu(i)\setminus\{j_1\}, p)\le u_i(\mu(i), p)$ by definition of partial stability. Thus, $i$ has not profitable swaps or drops.

    Now suppose that at $(\mu, p)$, every buyer $i$ has no profitable swaps or drops. Suppose for contradiction that $(\mu, p)$ is not partially stable, so that there exists some set of items $T$ with $|T|\le |\mu(i)|$ such that $u_i(T, p)> u_i(\mu(i), p)$. Consider the ``capacity constrained'' valuation function
     \[\tilde v_i(T') = \begin{cases}
    v_i(T') & |T'|\le |\mu(i)|\\
    v_i(T') - M(|T'|-|\mu(i)|) & \text{o.w.}
    \end{cases}\]
    where $M$ is a large constant, along with the corresponding utility function $\tilde u_i$. By Lemma~\ref{lem:gs_capacity}, $\tilde v_i$ also satisfies gross substitutes. $\tilde u_i$ has the same value as $u_i$ on sets of size at most $|\mu(i)|$, including $\mu(i)$ and $T$. Then by Lemma~\ref{lem:gs_si_gr_wl}, there exist sets of items $X,Y$ with $|X|, |Y|\le 1$ such that $\tilde u_i(\mu_i\cup X\setminus Y, p) > \tilde u_i(\mu(i), p)$. We see by construction that we cannot have $|X| = 1$ and $|Y| = 0$. Also, $|X| = |Y| = 0$ is impossible trivially. If $|X| = |Y| = 1$, then setting $X = \{j_2\}$ and $Y = \{j_1\}$, we have $$u_i(\mu(i), p) =\tilde u_i(\mu(i), p) < \tilde u_i(\mu(i)\cup\{j_2\}\setminus \{j_1\}, p)= u_i(\mu(i)\cup\{j_2\}\setminus \{j_1\}, p),$$ a contradiction. If $|X| = 0$ and $|Y| = 1$, then setting $Y = \{j_1\}$, we have $$u_i(\mu(i)\setminus\{j_1\}, p) = \tilde u_i(\mu(i)\setminus\{j_1\}, p)> \tilde u_i(\mu(i), p) = u_i(\mu(i), p),$$ again a contradiction. We conclude that $(\mu, p)$ is partially stable.
 \end{proof}

Now, we prove Lemma~\ref{lem:smo_core_cm}.

\begin{proof}{Proof of Lemma~\ref{lem:smo_core_cm}}
    First, assume $(\tau, p)$ is a competitive equilibrium such that all unit-buyers are matched. Because a MAT requires an unmatched unit-buyer as a root, a MAT cannot exist in $D(\tau, p)$.

    Now, assume that a MAT does not exist in $D(\tau, p)$. We know by Remark \ref{rem:mat_exists} that if an unmatched unit-buyer exists, a MAT must exist. Since there do not exist any MATs, all unit-buyers must be matched by $\tau$. By definition, we know that stability is implied by partial stability plus unit-buyer perfectness. Since we are assuming partial stability, we have stability, and since we are assuming feasibility, we have that $(\tau, p)$ is a competitive equilibrium.
 \end{proof}

\section{Proofs for Section \ref{sec:FPI}}\label{sec:proofs:FPI}

In this section, we provide proofs omitted from Section \ref{sec:FPI}.

\subsection{Proofs for Section 
\ref{sec:FPI:MATppi}}\label{sec:proofs:FPI:MATppi}

We prove Lemma \ref{lem:envyfree_slopes}, which shows how to verify a MAT-preserving price increase using the slopes of the effective price functions.

\begin{proof}{Proof}
    We wish to show that $(\tau', p+\lambda d)$ is partially stable and $\mathcal T'\subseteq D(\tau', p+\lambda d)$. By Lemma~\ref{lem:partial_stability_alt_defn}, to prove partial stability it is sufficient to show that there are no profitable swaps or drops. Throughout the proof, we assume that $\lambda>0$ is arbitrarily small. Fix a unit-buyer $k$, and let $T = \tau'(\text{copy}(k)\setminus\{k\})$. We will first show that there are no profitable swaps for $k$. That is, for all $(k,j_1)\in \tau'$ and all $j_2\in S\setminus T$, we have $u_k(j_1\mid T, p+\lambda d) \ge u_k(j_2\mid T, p+\lambda d)$. First, assume $d_{j_1} = 0$. Then, for any other item $j_2\in S\setminus T$ we have    
    \begin{align*}
        u_k(j_1\mid T, p+\lambda d) &= v_k(j_1\mid T) - q_{kj_1}(p_{j_1}+\lambda d_{j_1}) \\
        &= v_k(j_1\mid T) - q_{kj_1}(p_{j_1}) \\
        &= u_k(j_1\mid T, p) \\
        &\ge u_k(j_2\mid T, p) \\
        &= v_k(j_2\mid T) - q_{kj_2}(p_{j_2}) \\
        &\ge v_k(j_2\mid T) - q_{kj_2}(p_{j_2} + \lambda d_{j_2}) \\
        &= u_k(j_2\mid T, p+\lambda d)
    \end{align*}
    as desired, where the first inequality follows from partial stability of $(\tau', p)$, the second inequality follows from monotonicity of $q_{kj_2}$, and each equality follows from the definitions of $u$ and $q$. Now assume $d_{j_1} > 0$ and let $j_2\in F_{\tau', p}(k)$. Then $(k, j_1)\in \mathcal T'$ by definition of a MAT-preserving price increase, and so
    $u_k(j_1\mid T, p+\lambda d)\ge u_k(j_2\mid T, p+\lambda d).$
    For $j_2\not\in F_{\tau', p}(k)\cup T$, we have 
    \[u_k(j_1\mid T, p) > u_k(j_2\mid T, p).\]
    Thus, for sufficiently small $\lambda > 0$, it is also true that 
    \[u_k(j_1\mid T, p+\lambda d) >  u_k(j_2\mid T, p+\lambda d).\]
    We conclude that each unit-buyer $k$ has no profitable swaps. Now we will show that each unit-buyer $k$ has no profitable drops. Notice that dropping an item $j_1\in\tau'$ is equivalent to swapping it for a dummy item $j_0$: $u_k(\emptyset \mid T, p+\lambda d) = 0 = u_k(j_0\mid T, p+\lambda d)$. Thus, since $u_k(j_1\mid T, p+\lambda d) \ge u_k(j_0\mid T, p+\lambda d)$, we also know that $u_k(j_1\mid T, p+\lambda d) \ge u_k(\emptyset\mid T, p+\lambda d)$. We see that $(\tau', p+\lambda d)$ is partially stable. 
    
    Now, we will show that the edges of $\mathcal T'$ are contained in the marginal demand graph, which we write as $\mathcal T'\subseteq D(\tau', p+\lambda d)$. Fix $(k,j_1)\in\mathcal T'$ and $j_2\in S\setminus T$ not matched to a copy of $k$. For any $j\in S\setminus T$, we have
    $$u_k(j\mid T, p+\lambda d)= v_k(j\mid T) - q_{kj}(p_{j}+\lambda d_{j}) = v_k(j\mid T) - q_{kj}(p_{j}) - \lambda q'_{kj}(p_{j})d_{j}$$
    where the first equality is by definition and the second is by piecewise linearity. Note that since $\lambda$ is arbitrarily small, all prices are within the same linear domain. Then, since $j_1,j_2\in F_{\tau', p}(k)$, we know that
    $$v_k(j_1\mid T) - q_{kj_1}(p_{j_1}) = u_k(j_1\mid T, p)= u_k(j_2\mid T, p)= v_k(j_2\mid T) - q_{kj_2}(p_{j_2}).$$
    Since $q'_{kj_1}(p_{j_1})d_{j_1} \le q'_{kj_2}(p_{j_2})d_{j_2}$, we know that
    $$v_k(j_1\mid T) - q_{kj_1}(p_{j_1}) - \lambda q'_{kj_1}(p_{j_1})d_{j_1}\ge  v_k(j_2\mid T) - q_{kj_2}(p_{j_2})- \lambda q'_{kj_2}(p_{j_2})d_{j_2}$$
    $$\implies u_k(j_1\mid T, p+\lambda d) \ge u_k(j_2\mid T, p+\lambda d).$$
    Then, for $j_2\not\in F_{\tau', p}(k)$ and not matched to a copy of $k$, we have
    $u_k(j_1\mid T, p) >  u_k(j_2\mid T, p).$
    and as before for sufficiently small $\lambda > 0$ it is also true that
    $u_k(j_1\mid T, p+\lambda d) >  u_k(j_2\mid T, p+\lambda d).$
    Thus, $j_1\in F_{\tau', p+\lambda d}(k)$. We conclude that $\mathcal T'\subseteq D(\tau', p+\lambda d)$, and so $\mathcal T'$ is indeed contained in the marginal demand graph $D(\tau', p+\lambda d)$. \end{proof}

\subsection{Proofs for Section \ref{sec:FPI:matroid}}\label{sec:proofs:FPI:matroid}

First, we have a helper result showing that gross substitutes is preserved with any initial endowment and prices.

\begin{lemma}\label{lem:gs_endowment_prices}
    Let $v_i:2^S\to\mathbb R$ satisfy gross substitutes. Fix any $T\subseteq S$, any price vector $p$ over the items, and $\tilde q_i$ defined by $\tilde q_{ij}(\tilde p_j) := q_{ij}(\tilde p_j + p_j) - q_{ij}(p_j)$. Viewing $u_i(\cdot\mid T, p):2^{S\setminus T}\to \mathbb R$ as a valuation function with corresponding utility function
    $$\tilde u_i(\tilde T, \tilde p) := u_i(\tilde T\mid T, p) - \sum_{j\in \tilde T}\tilde q_{ij}(\tilde p_j) = u_i(\tilde T\mid T, \tilde p),$$  
    then the \emph{valuation} function $u_i(\cdot\mid T, p)$ also satisfies gross substitutes.
\end{lemma}
\begin{proof}{Proof}

    We show that $u_i(\cdot\mid T, p)$ satisfies property ($M^\natural$) from Proposition~\ref{prop:gs_equivalent}. Fix $\tilde T, \tilde T'\in S\setminus T$ and any $j\in \tilde T\setminus \tilde T'$. We can write 
    $$u_i(\tilde T, p) + u_i(\tilde T', p) = v_i(\tilde T\cup T) + v_i(\tilde T'\cup T) - \sum_{j\in \tilde T\cup T}q_{ij}(p_j) - \sum_{j\in \tilde T'\cup T}q_{ij}(p_j)$$
    Since $v_i$ satisfies ($M^\natural$), there exists $R\subseteq \tilde T'\setminus \tilde T$ with $|R|\le 1$ such that
    $$\le v_i(\tilde T\cup T\cup R\setminus \{j\}) + v_i(\tilde T'\cup T\cup\{j\}\setminus R) - \sum_{j\in \tilde T\cup T}q_{ij}(p_j) - \sum_{j\in \tilde T'\cup T}q_{ij}(p_j)$$
    $$ = v_i(\tilde T\cup T\cup R\setminus \{j\}) + v_i(\tilde T'\cup T\cup\{j\}\setminus R) - \sum_{j\in \tilde T\cup T\cup R\setminus \{j\}}q_{ij}(p_j) - \sum_{j\in \tilde T'\cup T\cup\{j\}\setminus R}q_{ij}(p_j)$$
    $$ = u_i(\tilde T\cup R\setminus \{j\}, p) + u_i(\tilde T'\cup \{j\}\setminus R, p)$$
    Thus, $u_i(\cdot\mid T, p)$ also satisfies ($M^\natural$). By Proposition~\ref{prop:gs_equivalent}, $u_i(\cdot\mid T, p)$ also satisfies gross substitutes.   
\end{proof}

It has also been shown that the gross substitutes is preserved under convolution.

\begin{lemma}[Convolution property \cite{lehmann2001combinatorial} \cite{murota1996convexity}]\label{lem:gs_convolution}
    Let $v_1,\ldots, v_z:S\to \mathbb R$ satisfy gross substitutes. Given a set $T\subseteq S$, let $\Pi(T)$ be the set of partitions of $T$ into $z$ groups. Thus, if $\pi=(\pi_1,\dots,\pi_z)\in \Pi(T)$, then $\bigcup_{x = 1}^z \pi_x = T$, and $\pi_x\cap \pi_y = \emptyset$ for $x\ne y$. Then, the function $v^*:S\to\mathbb R$ given by
    \[v^*(T) = \max_{\pi\in \Pi(T)}\sum_{x=1}^z v_x(\pi_x(T))\]
    also satisfies gross substitutes.
\end{lemma}

Now we prove Lemma~\ref{lem:MB_matroid}. Recall that we are given a price vector $p$, a partially stable outcome $(\tau,p)$ and a MAT ${\mathcal T}$ in $D(\tau,p)$. 

\begin{proof}{Proof of Lemma \ref{lem:MB_matroid}}
    Let $\mu$ be the one-to-many projection of $\tau$, and let $\hat\mu := \mu\setminus B(\mathcal T)\times S(\mathcal T)$ be the subset of $\mu$ that do not include buyers and items in $\mathcal T$. Also let $\mu' = \mu\setminus\hat\mu = \mu\cap B(\mathcal T)\times S(\mathcal T)$ be the subset of $\mu$ restricted to the buyers and items in $\mathcal T$. Let $\mathcal J_B$ be as in the definition of the lemma. Then clearly $\mathcal J_B$ can be alternatively defined as follows:
    \begin{equation}\mathcal J_B=\arg\max_{E \in \mathcal{E}} \sum_{i \in B'(\mathcal T)} u_i(S(E(i)) \mid \mu'(i), p),\label{eq:J-B_alt}\end{equation}
    where $\mathcal E = \prod_{i \in B'({\mathcal T})} {\mathcal E}_i$ and $E(i) =\{(i,j) \in E\}$.

    For $T\subseteq S(\mathcal T)$ and $i\in B'(\mathcal T)$, define 
    $$\tilde v_i(T) = \begin{cases}
    v_i(T\mid \hat\mu(i))  & \hbox{if } |T|\le |\mu'(i)|,\\
    v_i(T\mid \hat\mu(i)) - M(|T|-|\mu'(i)|) & \text{otherwise,}
    \end{cases}$$
    where $M$ is a large constant to be fixed. We know by Lemma~\ref{lem:gs_endowment_prices} and Lemma~\ref{lem:gs_capacity} that $\tilde v_i$ satisfies gross substitutes. 
    Then, for a set of edges $E\subseteq W(\mathcal T)$, define
    \[\hat v_i(E) := \tilde v_i(E(i)), \]
    \[\hat u_i(E) := \hat v_i(E) - \sum_{(i,j): (i,j)\in E}q_{ij}(p_j).\]
    Similarly, for $E \subseteq W({\mathcal T})$ define
\begin{equation}\label{eq:v*_new}v^{*}(E) := \max_{\pi\in \Pi(E)}\sum_{i\in B'(\mathcal T)}\hat v_i(\pi_i)\end{equation}    \begin{equation}\label{eq:ustar}u^*(E) := v^*(E) - \sum_{(i,j)\in E}q_{ij}(p_j).\end{equation}   
    We claim that
    \begin{equation}\label{eq:J-B} {\mathcal J}_B=\arg\max_{E \subseteq W({\mathcal T}) : |E| = |\mu'|} u^*(E).\end{equation}
    
    First, let $E\subseteq W({\mathcal T})$ satisfy $ |E| = |\mu'|$. An optimal partition $\pi\in \Pi(E)$ maximizing \eqref{eq:ustar} satisfies
    \begin{equation}\label{eq:right-size}|\pi_i| = |\mu'(i)|.\end{equation}
    This equality holds because all $|\mu'| = \sum_{i\in B'(\cal T)}|\mu'(i)|$ items in $E$ must be matched, and by construction of $\hat v_i$, any buyer who is allocated strictly more than $|\mu'(i)|$ elements of $E$ decreases their utility when we choose $M$ large enough. Now, consider the set of edges $\mu'$ and its partition $\pi^* \in \Pi(\mu')$ given by $\pi^*_i = \mu'(i)$ for $i \in B'({\mathcal{T}})$. That is, $\pi^*$  assigns to each buyer $i$ exactly the edges of the matching $\mu$ that are contained in the MAT ${\cal T}$ and are incident to $i$. By partial stability of $\tau$, we know that
    \[\hat u_i(\pi^*_i))\ge \hat u_i(E_i)\]
    for all $E_i\subseteq W(\mathcal T)$ with $|E_i| \le |\mu'(i)|$. Thus, we also know that
    \[u^{*}(\mu') =  \sum_{i\in B'(\mathcal T)} \hat u_i(\pi^*_i) \ge \sum_{i\in B'(\mathcal T)} \hat u_i(\pi_i)\]
    for all $E\subseteq W(\mathcal T)$ with $|E| = |\tau\cap\mathcal T|$ and all partitions $\pi \in \Pi(E)$. It follows that for any $E$ that achieves the maximum in the right-hand side of~\eqref{eq:J-B} with optimal partition $\pi$, we cannot have $(i',j)\in \pi_i$, where $i'\ne i$. Thus, $E$ also maximizes \eqref{eq:J-B_alt}, and is an element of $\mathcal J_B$.
    
    Now assume $E$ is an element of $\mathcal J_B$. Then, $E(i)$ achieves maximum utility for buyer $i$'s utility $u_i(\cdot\mid \hat\mu(i), p)$, among all bundles of size $|\mu'(i)|$. Consider the partition $\pi^*$ over $E$ given by $\pi^*_i = E(i)$. Then,
    $$\hat u_i(\pi^*_i) = u_i(E(i)\mid \mu'(i), p)\ge u_i(E'(i)\mid \mu'(i), p)\ge \hat u_i(\pi'_i)$$
    for any set of edges $E'\subseteq W(\mathcal T)$ and any partition $\pi'$ of $E'$. We see that $E(i)$ maximizes $\hat u_i(\cdot)$ among sets of size $|\mu'(i)|$. It follows that $E$, via the partition $\pi^*$, maximizes $u^*(\cdot)$ among sets of size $|\mu'|$. We conclude that \eqref{eq:J-B} holds.

    It then suffices to show that the right-hand side of~\eqref{eq:J-B} is the set of bases of a matroid. Suppose for the moment that $\hat v_i$ satisfies gross substitutes for $i \in B'(\mathcal T)$. Then $v^*$ satisfies gross substitutes by Lemma~\ref{lem:gs_convolution}. Hence, we can apply Lemma~\ref{lem:gs_matroid_property} to conclude that $\arg\max_{E \subseteq W({\mathcal T}) : |E| = |\mu'|} u^*(E)$ form the bases of a matroid.

    Thus, let us show that, for $i \in B'(\mathcal T)$, the valuation function $\hat v_i$ satisfies gross substitutes. To do this, we  show that the property (ISO) from Proposition \ref{prop:gs_equivalent} holds. First, we observe $\hat v_i$ is submodular. Fix any $E\subseteq W(\mathcal T)$, and distinct edges $(i_1,j_1),(i_2,j_2),(i_3,j_3)\in W(\mathcal T)\setminus E$. Let $\overline E = \{j\mid (i,j)\in E\}$. We consider several cases. 
    
    \textbf{Case 1}: $i_1 = i_2 = i_3 = i$. Then, we can write
    
    \noindent\begin{minipage}{\textwidth}
        $$\hat v_i((i_1,j_1),(i_2,j_2)\mid E) + \hat v_i((i_3,j_3)\mid E)= \tilde v_i(j_1,j_2\mid \overline E) + \tilde v_i(j_3\mid \overline E) $$
        $$\le \max\bigg[\tilde v_i(j_1,j_3\mid \overline E) + \tilde v_i(j_2\mid \overline E),\tilde v_i(j_2,j_3\mid \overline E) + \tilde v_i(j_1\mid \overline E)\bigg], $$
        \begin{equation}\label{eq:from-proof-1}\end{equation}
    \end{minipage}
    where the first equality is by definition of $\hat v$ and the last inequality is by (ISO) on the function $\tilde v_i$. Then as desired we can rewrite the rightmost term in~\eqref{eq:from-proof-1} as
    $$\max\bigg[\hat v_i((i_1,j_1),(i_3,j_3)\mid E) + \tilde v_i((i_2,j_2)\mid E),\hat v_i((i_2,j_2),(i_3,j_3)\mid E) + \tilde v_i((i_1,j_1)\mid E)\bigg].$$

    \textbf{Case 2}: $i_1\ne i$, $i_2 = i_3 = i$. Then, we have
    $$\hat v_i((i_1,j_1),(i_2,j_2)\mid E) + \hat v_i((i_3,j_3)\mid E)= \tilde v_i(j_2\mid \overline E) + \tilde v_i(j_3\mid \overline E)$$
    $$\le \max\bigg[\tilde v_i(j_3\mid \overline E) + \tilde v_i(j_2\mid \overline E),\tilde v_i(j_2,j_3\mid \overline E) + \tilde v_i(j_1\mid \overline E)\bigg]$$
    $$=  \max\bigg[\hat v_i((i_1,j_1),(i_3,j_3)\mid E) + \hat v_i((i_2,j_2)\mid E),\hat v_i((i_2,j_2),(i_3,j_3)\mid E) + \hat v_i((i_1,j_1)\mid E)\bigg]$$
    The inequality follows since the first element in the $\max$ is equal to the LHS of the inequality.

    \textbf{Case 3}: $i_2\ne i$, $i_1 = i_3 = i$. This case follows from a symmetric argument to the previous case.

    \textbf{Case 4}: $i_3\ne i$, $i_1 = i_2 = i$. Then, we have
    $$\hat v_i((i_1,j_1),(i_2,j_2)\mid E) + \hat v_i((i_3,j_3)\mid E) = \tilde v_i(j_1,j_2\mid \overline E)$$
    $$\le \tilde v_i(j_1\mid \overline E) + \tilde v_i(j_2\mid \overline E)$$
    $$\le  \max\bigg[\hat v_i((i_1,j_1),(i_3,j_3)\mid E) + \hat v_i((i_2,j_2)\mid E),\hat v_i((i_2,j_2),(i_3,j_3)\mid E) + \hat v_i((i_1,j_1)\mid E)\bigg]$$
    where the middle inequality follows from submodularity of $v_i$. 

    \textbf{Case 5}: $i_1, i_2\ne i$, $i_3 = i$. Then, we have 
    $$\hat v_i((i_1,j_1),(i_2,j_2)\mid E) + \hat v_i((i_3,j_3)\mid E) = \tilde v_i(j_3\mid \overline E)$$
    $$ = \hat v_i((i_1,j_1),(i_3,j_3)\mid E) + \hat v_i((i_2,j_2)\mid E)$$
    $$\le  \max\bigg[\hat v_i((i_1,j_1),(i_3,j_3)\mid E) + \hat v_i((i_2,j_2)\mid E),\hat v_i((i_2,j_2),(i_3,j_3)\mid E) + \hat v_i((i_1,j_1)\mid E)\bigg]$$
    where the middle equality follows because $(i_1, j_1)$ and $(i_2, j_2)$ always have marginal value zero for $\hat v_i$ when $i_1, i_2\ne i$.

    \textbf{Case 6}: $i_1, i_3\ne i$, $i_2 = i$ or $i_2, i_3\ne i$, $i_1 = i$. These cases follow from substantially similar arguments to the previous case. 

    \textbf{Case 7}: $i_1, i_2, i_3\ne i$. Then, we can write
    $$\hat v_i((i_1,j_1),(i_2,j_2)\mid E) + \hat v_i((i_3,j_3)\mid E) = 0$$
    $$ = \max\bigg[\hat v_i((i_1,j_1),(i_3,j_3)\mid E) + \hat v_i((i_2,j_2)\mid E),\hat v_i((i_2,j_2),(i_3,j_3)\mid E) + \hat v_i((i_1,j_1)\mid E)\bigg]$$
    since the marginal value of $(i_1, j_1)$, $(i_2, j_2)$, and $(i_3, j_3)$ is always zero under $\hat v_i$ when $i_1, i_2, i_3 \ne i$. We see that in all cases, (ISO) holds. We conclude that $\hat v_i$ satisfies gross substitutes. Thus, the elements of \eqref{eq:J-B} form the bases of a matroid.
 \end{proof}

We next prove Lemma \ref{lem:M1_independence_oracle}, and show that an independence oracle for $M_B$ can be implemented efficiently. The oracle, described formally in Algorithm \ref{alg:M1_ind}, proceeds as follows. Let $E\subseteq W({\mathcal T})$ be the input set and $\mu'$ be obtained by restricting the one-to-many projection $\mu$ of $\tau$ to edges in $\mathcal T$, where $(\tau,p)$ is the current partially stable outcome. Using Lemma \ref{lem:MB_matroid}, we show that $\mu'$ is a basis of $M_B$. The main idea is to use the exchange property (Lemma~\ref{lem:basis_exchange}) to iteratively replace in $\mu'$ a buyer-item pair $(i,j)\in E\setminus \mu'$ with a buyer-item pair $(i',j') \in \mu'\setminus E$. If this exchange is possible, the algorithm updates $\mu'$ and repeats until $E$ is contained in the basis $\mu'$ (line \ref{alglin:M1_ind:update_mu'}), thus showing that $E$ is independent. If there is no such exchange, the algorithm concludes that $E$ is not independent, and terminates (line \ref{alglin:M1_ind:return_false}). In the proof of the correctness of Algorithm~\ref{alglin:M1_ind:endwhile}, we are going to use extensively the alternative definition of ${\cal J_B}$ given in~\eqref{eq:J-B} and the related notions introduced in the proof of Lemma~\ref{lem:MB_matroid}.

\smallskip
\begin{algorithm}[H]
    \caption{\textsc{$M_B$-Independence}}
    \label{alg:M1_ind}
    \begin{algorithmic}[1]
        \Require Partially stable outcome $(\tau, p)$, MAT $\mathcal T$ where all items are matched, set of edges $E\subseteq B'(\mathcal T)\times S(\mathcal T)$.
        \Ensure $True$ if $E\in \mathcal I_B$ or $False$ otherwise.
            \State $\mu'\gets\{(\text{buy}(k), j)\mid (k,j)\in \tau\cap\mathcal T\}$ \label{alglin:M1_ind:mu}
            \While{$E\not\subseteq \mu'$\label{alglin:M1_ind:while}}
                \State $(i,j)\gets$ any edge in $E\setminus\mu'$\label{alglin:M1_ind:select_edge}
                \If{$\exists (i',j')\in \mu'\setminus E$ s.t. $u^{*}(\mu') = u^{*}(\mu'\cup\{(i',j')\}\setminus\{(i,j)\})$\label{alglin:M1_ind:if}}\Comment{check exchange}
                    \State $\mu'\gets \mu'\cup \{(i',j')\}\setminus \{(i,j)\}$ \label{alglin:M1_ind:update_mu'}\Comment{make exchange}
                \Else
                    \State\Return $False$ \label{alglin:M1_ind:return_false}
                \EndIf
            \EndWhile \label{alglin:M1_ind:endwhile}
        \State \Return $True$
    \end{algorithmic}
\end{algorithm}
\smallskip

\begin{proof}{Proof of Lemma \ref{lem:M1_independence_oracle}}
    We first claim that $\mu' = \{(\text{buy}(k), j)\mid (k,j)\in\tau\cap \mathcal T\}$ (as defined in line 1 of Algorithm~\ref{alg:M1_ind}) is a basis of $M_B$. To see this, we can write using the definitions of $u^*$, $v^*$, $\hat \mu$, and $\tilde v$ that
    
    \begin{align*}
        u^{*}(\mu') &= v^{*}(\mu') - \sum_{(i,j)\in \mu'}q_{ij}(p_j) \\
        &= \sum_{i\in B'(\mathcal T)}\left(\tilde v_i(\mu'(i)) - \sum_{j\in\mu'(i)} q_{ij}(p_j)\right) \\
        &= \sum_{i\in B'(\mathcal T)}\left(v_i(\mu'(i) \mid \hat\mu(i)) - \sum_{j\in \mu'(i)} q_{ij}(p_j)\right) \\
        &= \sum_{i\in B'(\mathcal T)}u_i(\mu'(i) \mid \hat\mu(i), p). \\
    \end{align*}
    
    Then, for $T_i'\subseteq S(\mathcal T)$ with $|T_i'| = |\mu'(i)|$, we have $|T_i'\cup \hat \mu(i)|=|\mu(i)|$ and by partial stability of $(\tau,p)$, we deduce
    
    $$u_i(\mu'(i) \mid \hat\mu(i), p) = u_i (\mu(i),p) - u_i(\hat\mu(i), p) \geq u_i(T_i'\cup \hat \mu(i),p) - u_i(\hat\mu(i), p) = u_i(T'_i \mid \hat\mu(i), p).$$ 
        Let $E'$ be a basis of $M_B$. By~\eqref{eq:right-size}, $E'$ is the disjoint union of a collection of sets $\{T_i'\mid i\in B'(\mathcal T)\}$ with $|T_i'| = |\mu'(i)|$ for all $i$. Thus, we have 
    \begin{align*}
        u^{*}(\mu') &= \sum_{i\in B'(\mathcal T)}u_i(\mu'(i) \mid \hat\mu(i), p) \\
        &\ge \sum_{i\in B'(\mathcal T)}u_i(T'_i \mid \hat\mu(i), p) \\
        &= \sum_{i\in B'(\mathcal T)}\left(v_i(T'_i \mid \hat\mu(i)) - \sum_{j\in T'_i} q_{ij}(p_j)\right) \\
        &= \sum_{i\in B'(\mathcal T)}\left(\hat v_i(E'_i) - \sum_{(i,j)\in E_i'}q_{ij}(p_j)\right) \\
        &= v^{*}(E') - \sum_{(i,j)\in E'}q_{ij}(p_j) \\
        &= u^{*}(E'),
    \end{align*}
    
    where $E'_i = \{(i,j)\mid j\in T'_i\}$ and $E' = \bigcup_{i\in B'(\mathcal T)}E'_i$. It follows that that $\mu'$ is a basis of $M_B$.
    
    Suppose first $E$ is independent. Then $E\subseteq E'$ for some basis $E'$. By the basis exchange property (Lemma~\ref{lem:basis_exchange}) for each $(i',j')\in E\setminus \mu'$ there is some $(i,j)\in \mu'\setminus E'$ such that $\mu'\cup\{(i',j')\}\setminus\{(i,j)\}$ is also a basis. To find such a pair $(i,j),(i',j')$, following~\eqref{eq:J-B}, it suffices to verify if $\mu'\cup\{(i',j')\}\setminus\{(i,j)\}$ is an optimal set of edges for $u^{*}$. If this is true, then the if statement in line \ref{alglin:M1_ind:if} activates, and setting $\mu'\gets \mu'\cup\{(i',j')\}\setminus\{(i,j)\}$ in line \ref{alglin:M1_ind:update_mu'} preserves the fact that $\mu'$ is a basis of $M_B$ and strictly decreases $|E\setminus \mu'|$. Thus, the while loop repeats $O(|E|)=O(|B||S|)$ times until $E\subseteq \mu'$, when the algorithm completes the while loop and correctly returns $True$.
    
    Now suppose $E$ is not independent. Since, as argued above, every repetition of line \ref{alglin:M1_ind:if} decreases $|E\setminus \mu'|$ while preserving the fact that $\mu'$ is a basis of ${\mathcal J}_B$, at some iteration there is no selection of $(i,j)$ that makes line \ref{alglin:M1_ind:if} true. Thus, the algorithm moves to line 7 and correctly returns $False$.

    Now we analyze the runtime of the algorithm. Line \ref{alglin:M1_ind:mu} requires $O(|B||S|)$. The while loop from lines \ref{alglin:M1_ind:while}-\ref{alglin:M1_ind:endwhile} requires at most $O(|B||S|)$ iterations, since in each iteration it either terminates in line \ref{alglin:M1_ind:return_false} or increases $|E\cap\mu'|$ by 1. The if statement in line \ref{alglin:M1_ind:if} requires $O(|B||S|)$ time to verify, by enumeration. We conclude that the total runtime is $O(|B|^2|S|^2)$. 
 \end{proof}

Now, we can prove Lemma \ref{lem:min_weight_common_basis}, showing that the minimum weight common basis problem correctly solves the minimum weight partially stable perfect matching problem.

\begin{proof}{Proof of Lemma \ref{lem:min_weight_common_basis}}
    We again refer to the alternative definition of $\mathcal J_B$ from equation~\eqref{eq:J-B} and the related notions introduced in the proof of Lemma~\ref{lem:MB_matroid}. Notice that the original partially stable matching $\tau$, when restricted to $\tau\cap\mathcal T$, is such that its one-to-one projection $\mu'$ belongs to $\mathcal J_B$. Indeed, by definition of partial stability, $\mu'(i)$ gives bidder $i$ the maximum value for $u_i(\cdot\mid \hat\mu(i), p)$ among sets of size $|\mu'(i)|$, where $\hat \mu = \mu\setminus B(\mathcal T)\times S(\mathcal T)$. It follows that any element $\mu''\in\mathcal J_B$ must achieve the same value of $u_i(\cdot\mid \hat\mu(i), p)$ for each $i$, and thus $\mu''\cup (\mu\cap\mathcal T)$ is a partially stable matching, aside from the fact that items may be allocated more than once. It follows from Lemma \ref{lem:M2_basis} and the definition of partial stability that $\mu^*$ is exactly a minimum weight partially stable matching on $W(\mathcal T)$. Then, by definition, $\tau^*$ is also partially stable on $V(\mathcal T)\setminus \{k_0\}$. Furthermore, for any one-to-one matching $\tau\subseteq U'(\mathcal T)\times S(\mathcal T)$, we know that $\tau$ has the same weight as its one-to-many projection $\{(\text{buy}(k), j)\mid (k,j)\in\tau\}$. It follows that $\tau^*$ is also of minimum weight.
 \end{proof}

To solve the minimum weight common basis problem, we can use any existing algorithm for the maximum weight matroid intersection problem, such as those presented by \citet{lawler1975matroid} and \citet{frank1981weighted}.

\begin{lemma}[Max weight matroid intersection algorithm \cite{lawler1975matroid}, \cite{frank1981weighted}]\label{lem:max_weight_matroid_int_alg}
    There exists an algorithm \textsc{MatroidIntersection} which, given matroid independence oracles for $M_B$ and $M_S$, can find a maximum weight matroid intersection in $O(|B|^3|S|^3)$ oracle calls.
\end{lemma}

Then, we can transform the weights to change a minimum weight common basis problem into a maximum weight matroid intersection problem, to prove Lemma \ref{lem:min_weight_common_basis_alg}.

\begin{proof}{Proof of Lemma~\ref{lem:min_weight_common_basis_alg}}
    Define weights $w$ over the edges in $W(\mathcal T)$ given by
    $$w'((i,j)) = \log q'_{ij}(p_j) - \min_{(i',j')\in W(\mathcal T)}\log q'_{i'j'}(p_{j'})$$
    $$w((i,j)) = |S(\mathcal T)|\left(\max_{(i',j')\in W(\mathcal T)}w'((i',j'))\right)+ 1 - w'((i,j)).$$
    These weights are the typical transformation used to turn a minimum weight maximum cardinality matching problem into a maximum weight matching problem, by first translating and reflecting the weight of each edge to be negative then adding a large constant to the weight of each edge. We see that for $E, E'\subseteq W(\mathcal T)$ if $|E|>|E'|$, then $w(E) > w(E')$. Thus, any maximum weight intersection must be of maximum cardinality. Since we know there exists a common basis of $M_B$ and $M_S$, i.e. the original matching $\tau$, any maximum weight intersection must also have size at least $|\tau|$, and thus is a basis of both $M_B$ and $M_S$. We see that by applying \textsc{MatroidIntersection} to the weights given by $w$, we obtain a minimum weight common basis.
 \end{proof}

\subsection{Proofs for Section \ref{sec:FPI:duality}}\label{sec:proofs:FPI:duality}

First, we have a helper result by \citet{fujishige2003note}.

\begin{lemma}[\cite{fujishige2003note}]\label{lem:fujishige_yang}
    Let $v_i: 2^S \rightarrow \mathbb{R}$ satisfy gross substitutes. Let $T_1,T_2,T_3$ be disjoint sets with $|T_1| = |T_2|$. Then, for any $j\in T_1$, we have

    $$v_i(T_1\mid T_3) + v(T_2\mid T_3)\le \max_{j'\in T_2\setminus T_1}v_i(T_1\cup\{j'\}\setminus\{j\}\mid T_3)+ v_i(T_2\cup\{j\}\setminus\{j'\}\mid T_3).$$
    In particular, when $|T_1| = |T_2| = 2$, we can write $T_1 = \{j_1, j_2\}$ and $T_2 = \{j_3, j_4\}$ and obtain

    $$v_i(j_1, j_2\mid T_3) + v_i(j_3,j_4\mid T_3)\le \max\bigg[v_i(j_1,j_3\mid T_3)+v_i(j_2, j_4\mid T_3),v_i(j_1,j_4\mid S)+v_i(j_2, j_3\mid T_3)\bigg].$$
\end{lemma}

\begin{corollary}\label{corr:fujishige_yang}
    Let $v_i$ satisfy gross substitutes, and let $p$ be any price vector. Then $u_i(\cdot, p): 2^S\to \mathbb R$ also satisfies Lemma~\ref{lem:fujishige_yang}.
\end{corollary}
\begin{proof}{Proof}
    We can write
    $$u_i(T_1\mid T_3, p) + u_i(T_2\mid T_3, p) = v_i(T_1\mid T_3) + v(T_2\mid T_3) + \sum_{j\in T_1\cup T_2}q_{ij}(p_j)$$
    $$\le \max_{j'\in T_2\setminus T_1} v_i(T_1\cup\{j'\}\setminus\{j\}\mid T_3)+ v_i(T_2\cup\{j\}\setminus\{j'\}\mid T_3) + \sum_{j\in T_1\cup T_2}q_{ij}(p_j)$$
    $$ = \max_{j'\in T_2\setminus T_1} u_i(T_1\cup\{j'\}\setminus\{j\}\mid T_3, p)+ u_i(T_2\cup\{j\}\setminus\{j'\}\mid T_3, p)$$
    
\end{proof}

Next, we use properties of gross substitutes to prove Lemma \ref{lem:alt_cycles}, describing the structure of the marginal demand graph.

\begin{proof}{Proof of Lemma \ref{lem:alt_cycles}}
    By hypothesis and definition of demand graph, the set $\{j_1,\dots,j_z,j'_1,\dots,j'_z\}$ is composed of $2z$ distinct elements. For each $X\subseteq [z]$, define $T(X) := \tau(\text{ubuy}(i))\cup \{j_x'\mid x\in X\}\setminus \{j_x\mid x\in X\}$.
    
    Assume (a) is not true. Then, let $y$ be the largest number such that for all sets $X$ with $|X|\le y$, the bundle $T(X)$ optimizes $u_i(\cdot, p)$ among bundles of size $|\tau(\text{ubuy}(i))|$. Note that by definition of the marginal demand graph, this is always true for $y = 1$. Relabel the pairs so that $T([y+1])$ does not optimize $u_i(\cdot, p)$ among bundles of size $|\tau(\text{ubuy}(i))|$. Define $T = \tau(\text{ubuy}(i))\cup\{j_3',\ldots, j_{y+1}'\}\setminus\{j_1,\ldots, j_{y+1}\}$. By Lemma~\ref{lem:fujishige_yang} and Corollary~\ref{corr:fujishige_yang}, we know that 
    $$u_i(j_1, j_2'\mid T, p) + u(j_1', j_2\mid T, p)\le \max\bigg[u_i(j_1, j_2\mid T, p) + u_i(j_1', j_2'\mid T, p),u_i(j_1, j_1'\mid T, p) + u_i(j_2, j_2'\mid T, p)\bigg].$$
    By selection of $T$, we know that $T\cup\{j_1', j_2\}$, $T\cup\{j_1, j_2'\}$, and $T\cup\{j_1,j_2\}$ optimize $u_i(\cdot, p)$ among bundles of size $|\tau(\text{ubuy}(i))|$. Then, we have
    $$u_i(j_1, j_2'\mid T, p) = u_i(j_1', j_2\mid T, p) = u_i(j_1, j_2\mid T, p) = \max_{j,j'\in S}u_i(j,j'\mid T, p):= u^*.$$
    As a result, we have
    
    \noindent\begin{minipage}{\textwidth}
        $$u^* + u^*\le \max\bigg[u^* + u_i(j_1', j_2'\mid T, p),u_i(j_1, j_1'\mid T, p) + u_i(j_2, j_2'\mid T, p)\bigg].$$
        \begin{equation}\label{eq:u-star-twice}\end{equation}
    \end{minipage}
    Furthermore, each of the four terms on the RHS of~\eqref{eq:u-star-twice} are bounded from above by $u^*$. Since $T\cup\{j_1',j_2'\} = T([y+1])$, we know this bundle does not maximize $u_i(\cdot, p)$ among bundles of size $|\tau(\text{ubuy}(i))|$, so we have $u_i(j_1', j_2'\mid T, p) < u^*$. This strict inequality means that the first argument of the max is less than $2u^*$, so the second argument of the $\max$ must be equal to $2u^*$. 
    
    It follows that 
    $u_i(j_1, j_1'\mid T, p) = u_i(j_2, j_2'\mid T, p) = u^*.$ 
    That is, $T\cup\{j_1, j_1'\}$ and $T\cup\{j_2, j_2'\}$ optimize $u_i(\cdot, p)$ among bundles of size $|\tau(\text{ubuy}(i))|$. By Lemma \ref{lem:gs_matroid_property}, the set of bundles which optimize $u_i(\cdot, p)$ among bundles of size $|\tau(\text{ubuy}(i))|$ form the bases of a matroid. Consider two such bundles, $\tau(\text{ubuy}(i))$ and $T\cup \{j_1,j_1'\}$. By the strong basis exchange property, there is a bijection $f: \{j_2,\ldots, j_{y+1}\}\to \{j_1', j_3', \ldots, j_{y+1}'\}$ such that for every $j\in \{j_2,\ldots, j_{y+1}\}= \tau(\text{ubuy}(i))\setminus (T\cup\{j_1,j_1'\})$, we have $\tau(\text{ubuy}(i))\cup\{f(j)\}\setminus\{j\}$ optimizes $u_i(\cdot, p)$ among bundles of size $|\tau(\text{ubuy}(i))|$ as well. Let $x_1 = 1$, and let $x_2$ so that $j_{x_1}' = f(j_{x_2})$. We see that $\tau(\text{ubuy}(i))\cup\{j_1'\}\setminus\{j_{x_2}\}$ is an optimal size $|\tau(\text{ubuy}(i))|$ bundle, which means by definition of the marginal demand graph that $(k_{x_2}, j_1')\in D(\tau, p)$. Setting $x_3$ so that $j_{x_2}' = f(j_{x_3})$, we then see by the same reasoning that $(k_{x_3}, j_{x_2}')\in D(\tau, p)$. This pattern continues until $j_{x_{t-1}}' = f(j_{x_t})$, where $x_t = 2$. Thus, we have a sequence $x_1,\ldots, x_t$ such that $(k_{x_\ell+1}, j'_{x_{\ell}})\in D(\tau, p)$ for $1\le \ell\le t-1$, with $x_1 = 1$ and $x_t = 2$. Overall, this gives the path
    \[(k_1 = k_{x_1}, j'_{x_1}, k_{x_2}, j'_{x_2},\dots, j'_{x_{t-1}}, k_{x_t} = k_2).\]
    
    Now consider bundles $\tau(\text{ubuy}(i))$ and $T\cup \{j_2,j_2'\}$, which optimize $u_i(\cdot, p)$ among bundles of size $|\tau(\text{ubuy}(i))|$. By the bijective basis exchange property, there is a bijection $g: \{j_1,j_3\ldots, j_{y+1}\}\to \{j_2', \ldots, j_{y+1}'\}$ such that for every $j\in \{j_1,j_3\ldots, j_{y+1}\}$, we have $\tau(\text{ubuy}(i))\cup\{g(j)\}\setminus\{j\}$ is an optimal size $|\tau(\text{ubuy}(i))|$ bundle as well. Keeping $x_{t} = 2$, let $j_{x_t}' = g(j_{x_{t+1}})$. By the same reasoning as before, we see that $(k_{x_{t+1}}, j_{x_t})\in D(\tau, p)\setminus \tau$. Again, this creates a sequence until $j_{x_{s}}' = g(j_{x_{s+1}})$, where $x_{s+1} = 1$. We thus have another sequence $x_t, \ldots, x_{s+1}$ such that $(k_{x_\ell+1}, j'_{x_{\ell}})$ for $t\le \ell\le s+1$, with $x_t = 2$ and $x_{s+1} = 1$. This gives us the path
    \[(k_2 = k_{x_t}, j'_{x_t}, k_{x_{t+1}}, j'_{x_{t+1}},\dots, k_{x_{s}}, j'_{x_{s}}, k_{x_{s+1}} = k_1).\]
    Combining this with the previous path gives us a closed walk 
    $$(k_1 = k_{x_1}, j'_{x_1}, k_{x_2}, j'_{x_2},\dots, j'_{x_{t-1}}, k_{x_t} = k_2, j'_{x_t}, k_{x_{t+1}}, j'_{x_{t+1}},\dots, k_{x_{s}}, j'_{x_{s}}, k_{x_{s+1}} = k_1)$$
    in $D(\tau,p)\setminus \tau$. However, some nodes and edges may be repeated in this walk. If there are indices $1\le r<r'\le s$ such that $x_{r} = x_{r'}$, then we can shortcut the cycle to become
    $$(k_1 = k_{x_1}, j'_{x_1}, k_{x_2}, \dots, j'_{x_{r-1}}, k_{x_r} = k_{x_{r'}}, j'_{x_{r'}}, k_{x_{r'+1}}, \dots,  k_{x_{s}}, j'_{x_{s}}, k_{x_{s+1}} = k_1).$$
    This gives us a new, shorter sequence of indices and corresponding cycle. Furthermore, we can see that the cycle still alternates between the original edges $(k_{x_\ell}, j'_{x_\ell})$ and our newly found edges $(k_{x_\ell}, j'_{x_{\ell-1}})$. If there continue to be duplicate indices in this sequence, then we can repeat this shortcutting process, until the final sequence does not have any duplicate indices, satisfying (b).
 \end{proof}

Next, it has also been shown by \citet{edmonds1972theoretical} that $LP$ and $DP$ can be solved efficiently.

\begin{lemma}[Complexity of $LP$ \cite{edmonds1972theoretical}]\label{lem:LP_DP_runtime}
    $LP$ can be solved in $O(|B|^3|S|^3)$ time.
\end{lemma}

Finally, we can prove Lemma \ref{lem:connectmat} on the correctness and runtime of Algorithm \ref{alg:connectMAT}.

\begin{proof}{Proof of Lemma \ref{lem:connectmat}}
    First, note that $\omega^*_{k_0}$ is chosen in line \ref{alglin:connectmat:omega_k0} so that the constraint $\omega_{k_0} + \rho_j\le \log q'_{k_0j}(p_j)$ is satisfied for all $j\in S(\mathcal T)$, and furthermore at least one of these constraints is tight.

    Consider the main loop of Algorithm \ref{alg:connectMAT}. We can verify that in each iteration, if $k\in U(\mathcal T')$ and $j\in S(\mathcal T')$, then $\omega_k + \rho_j <(=) \log q'_{kj}(p_j)$ before the iteration implies $\omega_k + \rho_j <(=) \log q'_{kj}(p_j)$ after the iteration as well. Similarly, if $k\in U(\mathcal T)\setminus U(\mathcal T')$ and $j\in S(\mathcal T)\setminus S(\mathcal T')$, then the relationship $\omega_k + \rho_j <(=) \log q'_{kj}(p_j)$ is preserved. As a result, once an agent enters $\mathcal T'$, they remain there throughout the loop.
    
    Next, immediately after ending each iteration of the loop, there is at least one $k\in U(\mathcal T')$ and $j\in S(\mathcal T)\setminus S(\mathcal T')$ such that $\omega_k+\rho_j = \log q'_{kj}(p_j)$. It follows that after the next execution of lines \ref{alglin:connectmat:H_rho}-\ref{alglin:connectmat:T_prime}, $(k, j)$ enters $\mathcal T'$ and $j$ enters $S(\mathcal T')$. We also know that $j$ was matched to at least one other unit-buyer $k_1\in U(\mathcal T)\setminus U(\mathcal T')$, which enters $U(\mathcal T')$ after line \ref{alglin:connectmat:T_prime} as well. Thus, the size of $\mathcal T'$ increases each iteration, for at most $|S(\mathcal T)|$ iterations, at which point all items in $S(\mathcal T)$ and all of their matched unit-buyers, which is all of $U(\mathcal T)$, are in $\mathcal T'$.

    For the runtime, note that the loop takes $O(|S|)$ iterations, since each iteration adds at least one unit-buyer and item to $\mathcal T'$. Each iteration of lines \ref{alglin:connectmat:H_rho}-\ref{alglin:connectmat:decr_omega} take $O(|B||S|^2)$, by Lemma \ref{lem:existmat}. We conclude that the entire algorithm takes $O(|B||S|^3)$ time.
 \end{proof}

\subsection{Proofs for Section \ref{sec:FPI:alg}}\label{sec:proofs:FPI:alg}

In this section, we prove Lemma \ref{lem:gs_price_inc} on the correctness and runtime of Algorithm \ref{alg:findpriceincrease}.

\begin{proof}{Proof of Lemma \ref{lem:gs_price_inc}}
    First, by Lemma \ref{lem:MWlegalmatching_correct}, $\tau^*$ is an optimal solution for $DP(\tau^*, p, \mathcal T)$. By duality, we see that for $\omega^*, \rho^*$ as computed in line \ref{alglin:findpriceincrease:omegarho}, we have $\omega^*_k + \rho^*_j = \log q'_{\text{buy}(k)j}(p_j)$ if $x_{kj} = 1$, and $\omega^*_k + \rho^*_j\le \log q'_{\text{buy}(k)j}(p_j)$ for all $(k, j)\in D(\tau, p)\times U(\mathcal T)\setminus \{k_0\}\times S(\mathcal T)$. 

    After line \ref{alglin:findpriceincrease:connectmat} in Algorithm \ref{alg:findpriceincrease}, we know by Lemma \ref{lem:connectmat} that there exists a MAT $\mathcal T'$ that contains all of $U(\mathcal T)\cup S(\mathcal T)$, such that for each edge $(k,j)\in\mathcal T'$, we have $\omega^*_k+\rho^*_j = \log q'_{\text{buy}(k)j}(p_j)$. For all $(k,j)\in D(\tau^*, p)$, we also still have $\omega^*_k+\rho^*_j \le \log q'_{\text{buy}(k)j}(p_j)$. As a result, by taking $e$ to the power of these (in)equalities, it follows that for $(k,j_1)\in \mathcal T'$ and $(k,j_2)\in D(\tau^*, p)$,
    $$q'_{kj_1}(p_{j_1})d_{j_1} = q'_{kj_1}(p_{j_1})e^{-\rho^*_{j_1}} = e^{\omega^*_k} \le q'_{kj_2}(p_{j_2})e^{-\rho^*_{j_2}} = q'_{kj_2}(p_{j_2})d_{j_2}.$$
    By Lemma \ref{lem:envyfree_slopes}, $d$ is a MAT-preserving price increase, and so for any $0 \le \lambda\le \lambda^*$ we have $\tau^*$ is partially stable at $p+\lambda d$, and $\mathcal T'\subseteq D(\tau^*, p+\lambda d)$.

    Now we verify that the algorithm also finds the correct step size $\lambda^*$. In line \ref{alglin:findpriceincrease:lambda1}, we ensure that $\lambda_1^*$ is the minimum step size such that a new edge in $U(\mathcal T)\times S\setminus S(\mathcal T)$ enters the demand graph. In  line \ref{alglin:findpriceincrease:lambda2}, we ensure that $\lambda_2^*$ is the minimum step size required to reach the edge of the current linear domain. It follows that $p+\lambda^* d$ either has a new edge enter the demand graph, which is also adds an edge and an item to the MAT rooted at $k_0$, or $p+\lambda^* d$ reaches the edge of the current linear domain.

    Now we  consider the runtime of Algorithm \ref{alg:findpriceincrease}. Line \ref{alglin:findpriceincrease:maxweight} takes $O(|B|^5|S|^5)$, by Lemmas \ref{lem:M1_independence_oracle}, \ref{lem:M2_independence_oracle}, and \ref{lem:min_weight_common_basis_alg}. Line \ref{alglin:findpriceincrease:omegarho} can be solved in $O(|B|^3|S|^3)$ time by Lemma \ref{lem:LP_DP_runtime}. Line \ref{alglin:findpriceincrease:connectmat} takes $O(|B||S|^3)$ times by Lemma \ref{lem:connectmat}. Lines \ref{alglin:findpriceincrease:duality_trick}-\ref{alglin:findpriceincrease:lambda} require $O(|B||S|^2)$. In total, the entire algorithm takes $O(|B|^5|S|^5)$ time.
 \end{proof}

\section{Proofs for Section \ref{sec:main_result}}\label{sec:proofs:main_result}

In this section, we provide the proof for Lemma~\ref{lem:existmat}.

\begin{proof}{Proof of Lemma~\ref{lem:existmat}}

    First, we will show that the output $\mathcal T$ of Algorithm \ref{alg:mat} is an alternating tree. Edges are only added to $\mathcal T$ in line \ref{alglin:mat:addedges}. We also know that $j\not\in V(\mathcal T)$ chosen in line \ref{alglin:mat:j} and $k'\in U''$ chosen in lines \ref{alglin:mat:u''}-\ref{alglin:mat:addedges} must not already be in $V(\mathcal T)$. We see that $\mathcal T$ is acyclic, and so it is a tree. A unit-buyer $k$ can only be added to $\mathcal T$ if it is $k_0$ or on an edge $(k', j)$ in line \ref{alglin:mat:addedges}. We also know by line \ref{alglin:mat:u''} that $(k', j)$ is a matched edge. Then, consider the edge $(k,j)$ added in line \ref{alglin:mat:addedges}. When $k$ was added to $\mathcal T$, it was also added via a matched edge, and since unit-buyers have at most one match, we see that $(k,j)$ is not matched. We see that $\mathcal T$ is alternating as well.

    Now we will show that $\mathcal T$ is maximal. Consider any unit-buyer $k\in U(\mathcal T)$. We know that $k = k_0$, or $k = k'$ in some iteration of line \ref{alglin:mat:addedges}. In the first case, $k$ is added to $U'$ in line \ref{alglin:mat:init}. In the second case, $k$ is added to $U'$ in line \ref{alglin:mat:updateu}. In both cases, $k$ is not removed from $U'$ until after it is selected in line \ref{alglin:mat:while}. When this happens, all items $j\not\in V(\mathcal T)$ with $(k,j)\in H$ are added to $\mathcal T$ via the edge $(k,j)$, unless they were already in $V(\mathcal T)$. Thus, for every unit-buyer $k\in V(\mathcal T)$, all items $j$ with an edge to $k$ are also in $V(\mathcal T)$. Next, consider any item $j\in V(\mathcal T)$. We know that $j$ is added to $V(\mathcal T)$ in line \ref{alglin:mat:addedges}. At the same time, any matched edge $(k',j)\in\tau$ is also added to $\mathcal T$ unless $k'$ is already in the tree. Thus for every item $j\in V(\mathcal T)$, its matched unit-buyer is in $V(\mathcal T)$ as well, and conclude that $\mathcal T$ is maximal.

    Now, we consider the runtime of Algorithm \ref{alg:mat}. In each iteration of the for loop, an item $j\not\in V(\mathcal T)$ is selected and added to $V(\mathcal T)$ permanently. As a result, the for loop can iterate at most $O(|S|)$ times throughout the entire algorithm. In each iteration, lines \ref{alglin:mat:u''}-\ref{alglin:mat:updateu} require at most $O(|B||S|)$ time. We conclude that the total runtime is at most $O(|B||S|^2)$.
 \end{proof}

\section{Proofs for Section \ref{sec:lattice_IC}}\label{sec:proofs:lattice_IC}

In this section, we provide the proof of Theorem~\ref{thm:min_ce}, omitted from Section~\ref{sec:lattice_IC}. 

\subsection{Proof of Theorem~\ref{thm:min_ce}}

The main lemma shows that throughout the execution of the algorithm, for every item, we have that its current price is upper bounded by its price in any competitive equilibrium.

\begin{lemma}\label{lem:matpp_le}
Let $(\tau^*, p^*)$ be a competitive equilibrium.
At any iteration of Algorithm~\ref{alg:gs_core_mech} with current prices $p$, we have that $p \leq p^*$.
\end{lemma}
Note that Theorem~\ref{thm:min_ce} follows immediately from Lemma~\ref{lem:matpp_le}. The remainder of this section is thus devoted to proving Lemma~\ref{lem:matpp_le}.

\subsubsection{Helper Lemmas for proving Lemma~\ref{lem:matpp_le}}

The first helper lemma provides properties that follow from gross substitutability.

\begin{lemma}\label{lem:perturb_demand}
    Let $p$ be a price vector and $R_1, R_2\subseteq S$ disjoint sets. Fix a buyer $i$. For small $\epsilon > 0$, define $p'$ by $p'_j = p_j + \epsilon/q'_{ij}(p_j)$ for $j\in R_1$, $p'_j = p_j + \epsilon^2/q'_{ij}(p_j)$ for $j\in R_2$, and $p'_j = p_j$ otherwise. Let $N\subseteq 2^S$. Then, $T\in \arg\max_{Q\subseteq S:Q\in N}u_i(Q, p')$ if and only if:
    \begin{enumerate}[(a)]
        \item $T\in \arg\max_{Q\subseteq S:Q \in N}u_i(Q, p)$, \label{lem:perturb_demand:a}
        \item $T$ minimizes $|T\cap R_1|$ among sets satisfying (\ref{lem:perturb_demand:a})\label{lem:perturb_demand:b}, and
        \item $T$ minimizes $|T\cap R_2|$ among sets satisfying (\ref{lem:perturb_demand:a}) and (\ref{lem:perturb_demand:b}).\label{lem:perturb_demand:c}
    \end{enumerate}
    Similarly, if $p'$ is defined by $p'_j = p_j - \epsilon/(q'_{ij}(p_j-\epsilon))$ for $j\in R_1$, $p'_j = p_j - \epsilon^2/(q'_{ij}(p_j-\epsilon))$ for $j\in R_2$, and $p'_j = p_j$ otherwise, then $T\in \arg\max_{Q\subseteq S:Q \in N}u_i(T', p')$ if and only if:
    \begin{enumerate}[(a)]
    \setcounter{enumi}{3}
        \item $T\in \arg\max_{Q\subseteq S:Q \in N}u_i(Q, p)$,\label{lem:perturb_demand:d}
        \item $T$ maximizes $|T\cap R_1|$ among sets satisfying (\ref{lem:perturb_demand:d}), and \label{lem:perturb_demand:e}
        \item $T$ maximizes $|T\cap R_2|$ among sets satisfying (\ref{lem:perturb_demand:d}) and (\ref{lem:perturb_demand:e}).\label{lem:perturb_demand:f}
    \end{enumerate}
\end{lemma}
\begin{proof}{Proof}
    Consider the first statement. Assume $T$ satisfies all three conditions, and $T'$ is any other set with $|T'|\in N$. We can write using the definition of $u_i$ that
    $$u_i(T, p') = u_i(T, p) - \sum_{j\in T\cap R_1}q'_{ij}(p_j)\epsilon/q'_{ij}(p_j) - \sum_{j\in T\cap R_2}q'_{ij}(p_j)\epsilon^2/q'_{ij}(p_j) = u_i(T, p) - \sum_{j\in T\cap R_1}\epsilon - \sum_{j\in T\cap R_2}\epsilon^2$$
    Similarly, we have 
    $$u_i(T', p') = u_i(T', p) - \sum_{j\in T'\cap R_1}q'_{ij}(p_j)\epsilon/q'_{ij}(p_j) - \sum_{j\in T'\cap R_2}q'_{ij}(p_j)\epsilon^2/q'_{ij}(p_j) = u_i(T', p) - \sum_{j\in T'\cap R_1}\epsilon - \sum_{j\in T'\cap R_2}\epsilon^2$$
    To compare these two expressions, we examine each term in several different cases.

    \textbf{Case 1}: $T'$ does not satisfy (\ref{lem:perturb_demand:a}). Then $u_i(T, p) > u_i(T', p)$, and this inequality dominates because $\epsilon$ is arbitrarily small. Thus, $u_i(T, p') > u_i(T', p')$.

    \textbf{Case 2}: $T'$ satisfies (\ref{lem:perturb_demand:a}) but not (\ref{lem:perturb_demand:b}). Then $u_i(T, p) = u_i(T', p)$, but $-\sum_{j\in T\cap R_1}\epsilon > -\sum_{j\in T'\cap R_1}\epsilon$. This term dominates since $\epsilon^2 \ll \epsilon$, and so we have $u_i(T, p') > u_i(T', p')$.

    \textbf{Case 3}: $T'$ satisfies (\ref{lem:perturb_demand:a}) and (\ref{lem:perturb_demand:b}) but not (\ref{lem:perturb_demand:c}). Then $u_i(T, p) = u_i(T', p)$ and $-\sum_{j\in T\cap R_1}\epsilon = -\sum_{j\in T'\cap R_1}\epsilon$, but $-\sum_{j\in T\cap R_2}\epsilon^2 > -\sum_{j\in T'\cap R_2}\epsilon^2$. Thus, $u_i(T, p') > u_i(T', p')$.

    \textbf{Case 4}: $T'$ satisfies (\ref{lem:perturb_demand:a}), (\ref{lem:perturb_demand:b}), and (\ref{lem:perturb_demand:c}). Then $u_i(T, p) = u_i(T', p)$, and $-\sum_{j\in T\cap R_1}\epsilon = \sum_{j\in T'\cap R_1}\epsilon$, and $-\sum_{j\in T\cap R_2}\epsilon^2 = \sum_{j\in T'\cap R_2}\epsilon^2$, giving us $u_i(T, p') = u_i(T', p')$.

    Altogether, we see that $T\in \arg\max_{Q\subseteq S:Q\in N}u_i(Q, p')$ if and only if the three conditions are satisfied. The proof of the second statement follows from a substantially similar argument. 
\end{proof}

We next need a helper result on a variation of the definition of gross substitutes.

\begin{lemma}\label{lem:gs_monotone_demand}
    For any  prices $p$ and $p'$ such that $p\le p'$, items $R\subseteq S$ such that  $p_j = p_j'$ for all $j\in R$, and buyer $i \in B$, we have that
$$\min_{T\in F_p(i)}|T\cap R| \leq \min_{T\in F_{p'}(i)}|T\cap R|.$$
\end{lemma}

\begin{proof}{Proof}
    Let $\ell = \min_{T\in F_p(i)}|T\cap R|$, and suppose for contradiction that there exists $T'\in F_{p'}(i)$ with $|T'\cap R|< \ell$. For small $\epsilon>0$, let $\hat p$ be defined by $\hat p_j = p_j+\epsilon/q'_{ij}(p_j)$ for $j\in R$, and $\hat p_j = p_j$ otherwise. Similarly define $\hat p'$ from $p'$. By Lemma~\ref{lem:perturb_demand}, $\hat T\in F_{\hat p}(i)$ if and only if $\hat T\in F_p(i)$ and $\hat T$ minimizes $|\hat T\cap R|$ among such sets. That is, any $\hat T\in F_{\hat p}(i)$ has $|\hat T\cap R| = \ell$. Lemma~\ref{lem:perturb_demand} also gives us that $\hat T'\in F_{\hat p'}(i)$ if and only if $\hat T'\in F_{p'}(i)$ and $\hat T'$ minimizes $|\hat T'\cap R|$ among such sets. So, any $\hat T'\in F_{\hat p'}(i)$ has $|\hat T'\cap R| < \ell$. Now, since $\epsilon > 0$ is small and $p\le p'$, we know that $\hat p \le \hat p'$ with $\hat p_j = \hat p'_j$ whenever $p_j = p_j'$, including all $j\in R$. Fix any $\hat T\in F_{\hat p}(i)$. By gross substitutes, there exists $\hat T'\in F_{\hat p'}(i)$ such that $\hat T\cap R\subseteq \hat T'$. Since $|\hat T\cap R| = \ell$, we see that $|\hat T'\cap R| \ge \ell$, a contradiction. We conclude that for all $T'\in F_{p'}(i)$, we have $|T'\cap R|\ge \ell$.  
\end{proof}

We also have a helper result showing that optimal sets contain smaller optimal sets.

\begin{lemma}\label{lem:well_layered_again}
    Fix a gross substitutes valuation function $v_i$ and price vector $p$. Let $\ell_1\in \mathbb N$ and $T_1\in \arg\max_{T\subseteq S:|T|=\ell_1}u_i(T, p)$ be a utility maximizing bundle among bundles of size $\ell_1$. Then, for any $\ell_2\le\ell_1$, there exists $T_2\in \arg\max_{T\subseteq S:|T|=\ell_2}u_i(T, p)$ such that $T_2\subseteq T_1$.
\end{lemma}
\begin{proof}{Proof}

    By property (GR) from Proposition~\ref{prop:gs_equivalent}, there exists a greedy procedure with a suitable tiebreaking rule which produces $T_1$ as the optimal solution to $\max_{T\subseteq S:|T|=\ell_1}u_i(T, p)$. Then, the greedy procedure with the same tiebreaking rule produces some $T_2\subseteq T_1$ as the optimal solution to $\max_{T\subseteq S:|T|=\ell_2}u_i(T, p)$. 
\end{proof}

\subsubsection{Proof of Lemma~\ref{lem:matpp_le}}

Now, we can prove Lemma~\ref{lem:matpp_le}.

\begin{proof}{Proof of Lemma~\ref{lem:matpp_le}}
    We show the result by induction. For the base case, the initial price vector $p = 0$ clearly satisfies $p\le p^*$. For the inductive step, assume the current price vector $p$ in Algorithm~\ref{alg:gs_core_mech} satisfies $p\le p^*$. Let $p+\lambda^* d$ be the price vector at the next iteration, for a MAT-preserving price increase $d$ and step size $\lambda^*$ as computed by Algorithm~\ref{alg:findpriceincrease}. Let $\tau'$ and $\mathcal T'$ be the matching and MAT associated with $d$, respectively. Suppose for contradiction that $p^*_j < p_j + \lambda^* d_j$ for some item $j\in S(\mathcal T')$. Then, there is some $\xi\in [0, \lambda^*)$ such that $p^*_j \ge p_j + \xi d_j$ for all $j$, with equality for at least one item $j\in S(\mathcal T')$. Define $p' = p + \xi d$, and let $R = \{j\in S(\mathcal T')\mid p'_j = p^*_j\}$. By Lemma~\ref{lem:gs_price_inc}, $(\tau', p')$ is also partially stable. Fix any buyer $i$, and for small $\epsilon > 0$ define $\hat p'$ by $\hat p'_j = p'_j + \epsilon / q'_{ij}(p'_j)$ if $j\in S(\mathcal T')\setminus R$, and $\hat p'_j = p'_j$ otherwise. Then $p'\le \hat p'\le p^*$, with equality for $j\in R$. Our goal is to show that for any $\hat T'_i\in F_{\hat p'}(i)$ we have $|\hat T'_i\cap R|\ge |\tau'(\text{ubuy}(i))\cap R|$, with strict inequality for some $i$. Then, we will use Lemma~\ref{lem:gs_monotone_demand} to show that the same is true for all $i$ and $T^*_i\in F_{p^*}(i)$, from which it follows that $\sum_{i\in B}|T^*_i\cap R| > \sum_{i\in B}|\tau'(\text{ubuy}(i))\cap R| = |R|$, which means that $R$ is overdemanded at $p^*$.

    For any buyer $i$, by Lemma~\ref{lem:perturb_demand}, for any $\ell\in \mathbb N$ we have $\hat T'_i\in \arg\max_{T\subseteq S:|T|=\ell}u_i(T, \hat p')$ if and only if
    \begin{enumerate}[(a)]
        \item $\hat T'_i\in \arg\max_{T\subseteq S:|T|=\ell}u_i(T, p')$, and \label{lem:matpp_le:a}
        \item $|\hat T'_i\cap (S(\mathcal T')\setminus R)|$ is minimized among sets satisfying (\ref{lem:matpp_le:a})\label{lem:matpp_le:b}.
    \end{enumerate}

    Next, we claim that any $X'_i\in \arg\max_{T\subseteq S:|T|=|\tau'(\text{ubuy}(i))|}u_i(T, p')$ has $|X'_i\cap S(\mathcal T')| \ge |\tau'(\text{ubuy}(i))\cap S(\mathcal T')|$. Suppose for contradiction there exists $X'_i\in \arg\max_{T\subseteq S:|T|=|\tau'(\text{ubuy}(i))|}u_i(T, p')$ with $|X'_i\cap S(\mathcal T')| < |\tau'(\text{ubuy}(i))\cap S(\mathcal T')|$. First note that we also have $\tau'(\text{ubuy}(i)) \in \arg\max_{T\subseteq S:|T|=|\tau'(\text{ubuy}(i))|}u_i(T, p')$ since $(\tau', p')$ is partially stable. By Lemma~\ref{lem:gs_matroid_property}, the collection of bundles $\arg\max_{T\subseteq S:|T|=|\tau'(\text{ubuy}(i))|}u_i(T, p')$ form the bases of a matroid. By the bijective bases exchange property, there exists a bijection $f:\tau'(\text{ubuy}(i))\setminus X'_i\to X'_i\setminus \tau'(\text{ubuy}(i))$ such that $\tau'(\text{ubuy}(i))\cup\{f(j)\}\setminus \{j\}$ is a basis as well for each $j\in \tau'(\text{ubuy}(i))\setminus X'_i$. Because $|X'_i\cap S(\mathcal T')| < |\tau'(\text{ubuy}(i))\cap S(\mathcal T')|$, by the pigeonhole principle there is some $j\in (\tau'(\text{ubuy}(i))\setminus X'_i)\cap S(\mathcal T')$ with $f(j)\in (X'_i\setminus \tau'(\text{ubuy}(i)))\setminus S(\mathcal T')$, which means $\tau'(\text{ubuy}(i))\cup\{f(j)\}\setminus \{j\} \in \arg\max_{T\subseteq S:|T|=|\tau'(\text{ubuy}(i))|}u_i(T, p')$. By definition of the marginal demand graph $D(\tau', p')$, since $\tau'(j)$ would be indifferent between swapping $j$ for $f(j)$, there must be an unmatched edge  $(\tau'(j), f(j))\in D(\tau', p')$. Since $j\in S(\mathcal T')$, it follows that $\tau'(j)\in U(\mathcal T')$, and then $f(j)\in S(\mathcal T')$, contradicting our choice of $f(j)\not\in S(\mathcal T')$. We see that $|X'_i\cap S(\mathcal T')| \ge |\tau'(\text{ubuy}(i))\cap S(\mathcal T')|$. By (\ref{lem:matpp_le:a}), we also know any $\hat X'_i\in \arg\max_{T\subseteq S:|T|=|\tau'(\text{ubuy}(i))|}u_i(T, \hat p')$ also has $\hat X'_i\in \arg\max_{T\subseteq S:|T|=|\tau'(\text{ubuy}(i))|}u_i(T, p')$, and so $|\hat X'_i\cap S(\mathcal T')| \ge |\tau'(\text{ubuy}(i))\cap S(\mathcal T')|$.

    Next, we claim that any $\hat X'_i\in \arg\max_{T\subseteq S:|T|=|\tau'(\text{ubuy}(i))|}u_i(T, \hat p')$ has $|\hat X'_i\cap R| \ge |\tau'(\text{ubuy}(i))\cap R|$. Additionally, we claim that there exists some $i^*$ such that $\hat X'_{i^*}\in \arg\max_{T\subseteq S:|T|=|\tau'(\text{ubuy}(i))|}u_i(T, \hat p')$ with $|\hat X'_i\cap R| > |\tau'(\text{ubuy}(i))\cap R|$.
    
    Consider an arbitrary buyer $i$ such that $\hat X'_i\in \arg\max_{T\subseteq S:|T|=|\tau'(\text{ubuy}(i))|}u_i(T, \hat p')$. We consider two cases:

    \textbf{Case 1}: there exists $k\in \text{ubuy}(i)$ such that $\tau'(k)\not\in R$, but $(k, j)\in D(\tau', p')$ for some $j\in R$. To see that there exists a buyer $i = i^*$ satisfying this case, note that since $\mathcal T'$ is a MAT, there is an alternating path from the root $k_0$ to each item $j'\in S(\mathcal T)$. Furthermore, each item $j'\in S(\mathcal T')$ has $\tau'(j')\in\mathcal T'$. Fix any $j'\in R$, and consider the shortest alternating path from $k_0$ to $j'$. Then, $k_0\not\in \tau'(R)$ and $\tau'(j')\in \tau'(R)$ are the first and last unit-buyers on the path, respectively. There thus exists some unit-buyer $k$ which is the last unit-buyer on the path that has $\tau'(k)\not\in\tau'(R)$. Then, then next two nodes on the path are of the form $j, \tau'(j)$ for some $j\in R$. $k, j$ and $i^* = \text{buy}(k)$ thus satisfy the requirements of Case 1. 
    
    Now, for any $i$ satisfying this case, set $Y'_i = \tau'(\text{ubuy}(i))\cup\{j\}\setminus\tau'(k)$. We know that $\tau'(k)\in S(\mathcal T')\setminus R$ and $j\in R$, and so it follows that $|Y'_i\cap (S(\mathcal T')\setminus R)| < |\tau'(\text{ubuy}(i))\cap (S(\mathcal T')\setminus R)|$, and by definition of $D(\tau', p')$, we know that $Y'_i\in\arg\max_{T\subseteq S:|T|=|\tau'(\text{ubuy}(i))|}u_i(T, p')$.

    \textbf{Case 2}: Case 1 does not apply. Then, set $Y'_i = \tau'(\text{ubuy}(i))$, in which case we have $|Y'_i\cap (S(\mathcal T')\setminus R)|\le |\tau'(\text{ubuy}(i))\cap (S(\mathcal T')\setminus R)|$. By definition of $\tau'(\text{ubuy}(i))$ we know that $Y'_i\in\arg\max_{T\subseteq S:|T|=|\tau'(\text{ubuy}(i))|}u_i(T, p')$.
    
    In both cases, by (\ref{lem:matpp_le:b}) any $\hat X'_i\in\arg\max_{T\subseteq S:|T|=|\tau'(\text{ubuy}(i))|}u_i(T, \hat p')$ must have $|\hat X'_i\cap (S(\mathcal T')\setminus R)|\le |Y'_i\cap (S(\mathcal T')\setminus R)|\le |\tau'(\text{ubuy}(i))\cap (S(\mathcal T')\setminus R)|$. Since $|\hat X'_i\cap S(\mathcal T')| \ge |\tau'(\text{ubuy}(i))\cap S(\mathcal T')|$, it follows that $|\hat X'_i\cap R| \ge |\tau'(\text{ubuy}(i))\cap R|$, with strict inequality for $i = i^*$ in Case 1.

    Now we show that any $\hat T_i'\in F_{\hat p'}(i)$ and $\hat X'_i\in\arg\max_{T\subseteq S:|T|=|\tau'(\text{ubuy}(i))|}u_i(T, \hat p')$ has $|\hat T'_i\cap R| \ge |\hat X_i'\cap R|$. First, we know by partial stability of $\tau'$ at $p'$ that any $X_i'$ with $|X_i'| \le |\tau'(i)|$ has $u_i(X_i', p')\le u_i(\tau'(i), p') = u_i(\hat X'_i, p')$. Thus, if there exists $T_i'$ with $u_i(T_i', p') > u_i(\hat X_i', p')$, then $|T_i'| > |\hat X'_i|$. It would follow by Lemma~\ref{lem:well_layered_again} that there exists $T_i'\in F_{p'}(i)$ such that $T_i'\supseteq \hat X_i'$. If there does not exist $T_i'$ with $u_i(T_i', p') > u_i(\hat X_i', p')$, then $\hat X_i'\in F_{p'}(i)$. In either case, there exists $T_i'\in F_{p'}(i)$ with $T_i'\supseteq \hat X_i'$, and thus $|T_i'\cap R|\ge |\hat X_i'\cap R|$. By (a) and (b), $\hat T_i'\in F_{\hat p'}(i)$ must then also have $|\hat T_i'\cap R| \ge |\hat X_i'\cap R|$.
    
    Finally, by Lemma~\ref{lem:gs_monotone_demand}, since $p^*\ge \hat p'$ with equality for $j\in R$, we know that for every $i$ and every $T^*_i\in F_{p^*}(i)$, we have $|T^*_i\cap R| \ge |\hat T'_i\cap R|\ge |\hat X_i'\cap R|\ge |\tau'(\text{ubuy}(i))\cap R|$, with strict inequality when $i = i^*$. Thus, we have $\sum_{i\in B}|T^*_i\cap R| > \sum_{i\in B}|\tau'(\text{ubuy}(i))\cap R| = |R|$, which means there can be no stable matching at $p^*$, a contradiction. We conclude that $p^*\ge p+\lambda^* d$, and so by induction the current price $p$ at every iteration of Algorithm~\ref{alg:gs_core_mech} has $p\le p^*$.

\end{proof}

\section{Proofs for Section \ref{sec:negative}}\label{sec:proofs:negative}

In this section, we provide the proofs for Theorem~\ref{thm:hardness2} and Theorem~\ref{thm:hardness}.

\begin{proof}{Proof of Theorem~\ref{thm:hardness2}}
    We will reduce the integer knapsack problem to an instance of our problem. Consider any knapsack instance with items $S$, where each item $j\in S$ has value $w_j$ and cost $c_j$, and we have a maximum budget $C$. We assume that $w_j$, $c_j$, and $C$ are all integers. The goal is to maximize    
    \begin{align*}
        \max & \sum_{j\in S}w_j x_j\\
         \text{s.t. } & \sum_{j\in S} c_j x_j\le C \\
          & x_j\in \{0,1\}\quad \forall j\in S
    \end{align*} 
    
    Now consider the following instance of the assignment game. Let the items be given by $S' = \{j_1, j_2\mid j\in S\}$. Then, for each knapsack item $j\in S$, define three buyers $i_j^1, i_j^2, i_j^3$ that have unit-demand, quasilinear utilities given by $v_{i_j^\ell}(j_1) = v_{i_j^\ell}(j_2) =  c_j$, $v_{i_j^\ell}(j') = 0$ for $j'\ne j_1, j_2$, $v_{i_j^\ell}(T) = \max_{j'\in T}v_{i_j^\ell}(j')$, and $r_{i_j^\ell}(p)  = p$.
    Thus, $i_j^1$, $i_j^2$, and $i_j^3$ only value the copies of item $j$, at $c_j$. Also, define a buyer $i_0$ with utility function as in Definition~\ref{def:nqbu} and valuation function
    $v_{i_0}(T) = 2C\sum_{j\in T}(\mathbf{1}_{\{j_1,j_2\}\cap T\ne\emptyset}) w_j$. 
    That is, $v_{i_0}$ is an additive valuation function over the items in $S'$, but only values at most one copy of each item. Now consider any competitive equilibrium $(\mu, p)$. If two copies $j_1,j_2$ of the same item $j$ have different prices, then no buyer including $i_j^1, i_j^2, i_j^3, i_0$ will demand the item with higher price, which means the outcome is infeasible. Thus, we know that $p_{j_1} = p_{j_2}$. If $p_{j_1} = p_{j_2}<c_j$, then all three buyers $i_j^1, i_j^2, i_j^3$ demand one of $j_1, j_2$ but only two can receive one, which would not be a competitive equilibrium. If $p_{j_1} = p_{j_2} > c_j$, then none of the buyers $i_j^1, i_j^2, i_j^3$ demand either $j_1, j_2$, so at most one buyer ($i_0$) can match to one of $j_1, j_2$. Thus, one of $j_1, j_2$ is unmatched but has a nonzero price, which is impossible. It follows that $p_{j_1} = p_{j_2} = c_j$ for all items $j\in S$. Then, consider the set of items $\mu(i_0)$ matched to $i_0$. If two copies $j_1,j_2$ of the same knapsack item $j$ are in $\mu(i_0)$, removing one from $\mu(i_0)$ would maintain the same value while reducing the price paid by $c_j$, improving utility. This is impossible at a competitive equilibrium, so we know that $\mu(i_0)$ contains at most one of $j_1,j_2$ for each $j\in S$. For any bundle $T$ where the sum of the costs $c_j$ is greater than $C$, we have
    $$
        u_{i_0}(T, p) = v_{i_0}(T) - r_{i_0}\left(\sum_{j_\ell\in T}p_{j_\ell}\right)= 2C\sum_{{j_\ell}\in T}w_{j_\ell} - r_{i_0}\left(\sum_{{j_\ell}\in T}c_j\right)\le -\infty
    $$
    where the first inequality follows from the fact that $\sum_{{j_\ell}\in T}c_j> C$ implies $\sum_{{j_\ell}\in T}c_j\ge C+1$ with integer data. We see that $i_0$ cannot be matched to any bundle $T$ whose costs sum to greater than $C$. Now consider $\mu(i_0)$, and any bundle $T$ where the sum of the costs is at most $C$. Thus,
    $u_{i_0}(\mu(i_0), p) \ge u_{i_0}(T, p)$ and $2C\sum_{{j_\ell}\in \mu(i_0)}w_{j_\ell} - r_{i_0}\left(\sum_{{j_\ell}\in \mu(i_0)}p_{j_\ell}\right)\ge 2C\sum_{{j_\ell}\in T}w_{j_\ell} - r_{i_0}\left(\sum_{{j_\ell}\in T}p_{j_\ell}\right).$
    This can be rewritten as $2C\sum_{{j_\ell}\in \mu(i_0)}w_{j_\ell} - r_{i_0}\left(\sum_{{j_\ell}\in \mu(i_0)}c_j\right) \ge 2C\sum_{{j_\ell}\in T}w_{j_\ell} - r_{i_0}\left(\sum_{{j_\ell}\in T}c_j\right).$
    Since the costs sum to at most $C$, we know by the definition of $r_{i_0}$ that 
    $r_{i_0}\left(\sum_{{j_\ell}\in \mu(i_0)}c_j\right),\ r_{i_0}\left(\sum_{{j_\ell}\in T}c_j\right)\le C.$ 
    As a result, since the $w_j$ are integer, we know that
    $\sum_{{j_\ell}\in \mu(i_0)}w_{j_\ell} \ge \sum_{{j_\ell}\in T}w_{j_\ell}.$
    Thus, $i_0$ must be matched to a bundle $\mu(i_0)$ with total cost at most $C$ that has the maximum total weight, which is exactly the optimal solution to the original integer knapsack problem. Thus, computing a competitive equilibrium is at least as hard as the NP-hard integer knapsack problem
\end{proof}

\begin{proof}{Proof of Theorem~\ref{thm:hardness}}
    We follow a very similar plan to the proof of Theorem~\ref{thm:hardness2}. Using the same integer knapsack problem, we again define an instance of the assignment game with items $S' = \{j_1, j_2\mid j\in S\}$, and buyers $i_j^1, i_j^2, i_j^3$ for each $j\in S$ with unit-demand, quasilinear utilities identical to those defined in the proof of Theorem~\ref{thm:hardness2}. Then, define a buyer $i_0$ with utility function as in Definition~\ref{def:nqbr} and valuation function
    $v_{i_0}(T) = 2C\sum_{j\in T}(\mathbf{1}_{\{j_1, j_2\}\cap T\ne \emptyset}) w_j$.
    As before, $v_{i_0}$ is an additive valuation function over the items in $S'$, but only values at most one copy of each item. For $r_{i_0}$, we have $C$ equal to the budget and $\alpha = 2C(1+\sum_{j\in S}w_j)$. Consider any competitive equilibrium $(\mu, p)$. As in the proof of Theorem~\ref{thm:hardness}, we know that $p_{j_1} = p_{j_2} = c_j$ for all items $j\in S$. We also know that that $\mu(i_0)$ contains at most one of $j_1,j_2$ for each $j\in S$. For any bundle $T$ where the sum of the costs $c_j$ is greater than $C$, we know that
    $$u_{i_0}(T, p) = v_{i_0}(T) - r_{i_0}\left(\sum_{j_\ell\in T}p_{j_\ell}\right)= 2C\sum_{{j_\ell}\in T}w_{j_\ell} - r_{i_0}\left(\sum_{{j_\ell}\in T}c_j\right)\le 2C\sum_{{j_\ell}\in T}w_{j_\ell} - 2C\left(\frac{1}{2} + 1 + \sum_{j\in S}w_j\right)\le - 3C.$$
    We see that $i_0$ cannot be matched to any bundle $T$ whose costs sum to greater than $C$. By identical reasoning to the proof of Theorem~\ref{thm:hardness2}, we see that $i_0$ must be matched to a bundle $\mu(i_0)$ with total cost at most $C$ that has the maximum total weight, which is exactly the optimal solution to the original integer knapsack problem. Thus, computing a competitive equilibrium is at least as hard as the integer knapsack problem, which is NP-hard.
 \end{proof}

\section{Payment functions and $r$-gross substitutes}\label{sec:r-gross}

Finally, we show that under non-separable payments, gross substitutes fails to generalize. The following definition extends the definition of gross substitutes to non-separable payments analogously to the way that it was generalized to the ITU model in Definition~\ref{def:defn_gs}.

\begin{definition}\label{def:payment_function}
    A \textit{payment function} is a function $r: 2^S\times \mathbb R^S \to \mathbb R$ that is continuous and strictly increasing in its second argument satisfying the following boundary conditions for any $p, p'\in \mathbb R^S$, $T\subseteq S$, and $j\in S$:
    \begin{enumerate}
        \item\label{it:payment:prop-1} $r(\emptyset, p) = 0$,
        \item\label{it:payment:prop-2} $r(T\cup\{j\}, p) = r(T, p)$ if $p_j = 0$,
        \item\label{it:payment:prop-3} $r(T, p) = r(T, p')$ if $j\not\in T$ and $p_{j'} = p'_{j'}$ for all $j'\ne j$,
        \item\label{it:payment:prop-4} $r(T\cup\{j\}, p)\to \infty$ as $p_j\to\infty$ and $r(T\cup\{j\}, p)\to -\infty$ as $p_j\to-\infty$.
    \end{enumerate}
    Given a payment function $r$, and a valuation function $v: 2^S\to \mathbb R$, define the utility function $u^r(T, p) = v(T) - r(T, p)$. We say $v$ satisfies \emph{$r$-gross substitutes} if, for any price vectors $p, p'\in \mathbb R^S$ such that $p\le p'$ and $T\in \arg\max_{R\subseteq S}u(R, p)$, there exists $T'\in \arg\max_{R\subseteq S}u(R, p')$ such that $T\cap \{j\mid p_j = p'_j\}\subseteq T'$.
\end{definition}

We can show that under any non-separability in $r$, $r$-gross substitutes is no longer equivalent to gross substitutes. These results also further justify our choice of separable effective price functions.

\begin{theorem}\label{thm:nonsep}
    Let $r: 2^S\times \mathbb R^S \to \mathbb R$ be a payment function, and let ${\cal V}_r$ be the set of valuations satisfying $r$-gross substitutes. Then ${\cal V}_r$ coincides with the set of functions satisfying gross substitutes if and only if there exist continuous, strictly increasing, surjective functions $q_{j}:\mathbb R\to\mathbb R$ such that $r(T, p) = \sum_{j\in T}q_{j}(p_j)$. 
\end{theorem}

\begin{proof}{Proof of Theorem~\ref{thm:nonsep}}
    Suppose there exit continuous, strictly increasing, surjective functions $q_{ij}:\mathbb{R}\rightarrow \mathbb{R}$ such that $r(T,p)=\sum_{j \in T} q_{ij} (p_j)$. We show that ${\cal V}_r$ coincides with the set of valuation functions satisfying gross substitutes. 

    Assume first that $v_i=v$ satisfies $r$-gross substitute. Fix any price vectors $p, p'$ with $p\le p'$, and let $T\subseteq S$ maximize $\mathring u(T',p)=v(T')-\sum_{j \in T'}p_j$. Since the effective price functions $q_{ij}$ are continuous, strictly increasing, and surjective, they admit inverse functions $q^{-1}_{ij}$. Define price vectors $\phi, \phi'$ by $\phi_j = q^{-1}_{ij}(p_j)$ and $\phi'_j = q^{-1}_{ij}(p'_j)$ for each $j\in S$. By monotonicity, $\phi\le \phi'$ as well. By definition, $\arg\max_{T'\subseteq S} u^r(T', \phi) = \arg\max_{T'\subseteq S}\mathring u(T',p)$, and so $T\in \arg\max_{T'\subseteq S} u^r(T', \phi)$. By $r$-gross substitutes, there exists $T^*\in \arg\max_{T'\subseteq S} u^r(T', \phi')$ such that $T\cap\{j\mid \phi_j = \phi'_j\}\subseteq T^*$. Since $\phi_j = \phi'_j$ if and only if $p_j = p'_j$, we see that there exists $T^*\in \arg\max_{T'\subseteq S}\mathring u(T',p)$ such that $T\cap\{j\mid p_j = p'_j\}\subseteq T^*$, giving gross substitutes.

    Now assume $v$ satisfies gross substitutes. Fix any price vectors $\phi, \phi'$ with $\phi\le \phi'$, and let $T\subseteq S$ maximize $u^r(T', \phi)$. Define $p, p'$ by $p_j = q_{ij}(\phi_j)$ and $p'_j = q_{ij}(\phi'_j)$ for each $j\in S$. By monotonicity, $p\le p'$ as well. By definition, $\arg\max_{T'\subseteq S} u^r(T', \phi) = \arg\max_{T'\subseteq S}\mathring u(T',p)$, and so $T\in \arg\max_{T'\subseteq S} \mathring u(T', p)$. By gross substitutes, there exists $T^*\in \arg\max_{T'\subseteq S}\mathring u(T', p')$ such that $T\cap\{j\mid p_j = p'_j\}\subseteq T^*$. Since $p_j = p'_j$ if and only if $\phi_j = \phi'_j$, we see that there exists $T^*\in \arg\max_{T'\subseteq S} u^r(T', \phi')$ such that $T\cap\{j\mid \phi_j = \phi'_j\}\subseteq T^*$, giving $r$-gross substitutes.

    To prove the opposite direction, we show the contrapositive. Assume that such functions $q_{j}$ do not exist. For each price vector $p$ and item $j$, define $q_j(x):=r(\{j\},p)$ for any $p$ with $p_j=x$. This is well-defined because $r(\{j\},p)$ is unaffected by prices of items outside $\{j\}$, as in Property~\ref{it:payment:prop-3} from the definition of payment function. If no bundle $T \subseteq S$ and price vector $p$ violated
    $$
    r(T,p)=\sum_{j\in T}q_j(p_j),
    $$
    then these $q_j$'s would give the separable representation, a contradiction. Hence there exist $T\subseteq S$ and $p$ such that
    $$
    r(T,p)\ne \sum_{j\in T}q_j(p_j).
    $$
    Choose such a $T$ of minimum cardinality. Then $|T|\ge 2$. Fix any $j\in T$ and let $T'=T\setminus\{j\}$. By minimality of $T$,
    $$
    r(T',p)=\sum_{\ell\in T'}q_\ell(p_\ell).
    $$
    If $r(T,p)=r(T',p)+r(\{j\},p)$, then since $r(\{j\},p)=q_j(p_j)$, we would get
    $$
    r(T,p)=\sum_{\ell\in T}q_\ell(p_\ell),
    $$
    contradicting the choice of $T$. Therefore,
    $$
    r(T,p)\ne r(T',p)+r(\{j\},p).
    $$
    
    We consider two cases.

    \textbf{Case 1}: $r(T, p) < r(T', p) + r(\{j\}, p)$. Let $\epsilon \in (0, (r(T', p) + r(\{j\}, p) - r(T, p))/2)$. Consider an additive valuation function $v$ where $v(\{\ell\}) = q_\ell(p_\ell)- \frac{\epsilon}{|T'|}$, for $\ell \in T'$, $v(\{j\}) = r(\{j\}, p) - \epsilon$, and $v(\{j'\}) = -M$ for all $j'\not\in T$, where $M = 2\max_{R\subseteq S}|r(R, p)|$. Clearly, any bundle with non-empty intersection with $S\setminus T$ has negative utility at price $p$.  Moreover, $$u(T)=v(T)-r(T,p)=r(T',p) + r(\{j\},p) - 2\epsilon -r(T,p) >0 \; \hbox{ and } \; u(\{j\})=r(\{j\},p)-\epsilon - r(\{j,p\})<0. $$ 
    So the optimal bundle at price $p$ contains at least one element of $T'$. 
    
    Now set $p'$ such that $p'_{j'} = p_{j'}$ for $j'\ne j$, and $p'_j$ is large enough such that $v(R) - r(R, p') < 0$ for any $R \subseteq S$ containing $j$ (which is possible because of Property~\ref{it:payment:prop-4} from the definition of payment function). We claim that the optimal bundle at $p'$ is $\emptyset$. Thus, $v$ satisfies gross substitutes (being additive) but does not satisfy $r$-gross substitutes. First observe that by Property~\ref{it:payment:prop-1},
    $
    u(\emptyset,p')=0.
    $ 
    Let now $R\subseteq S$, $R \neq \emptyset$. If $j \in R$, $u(R,p')<0$ by our choice of $p'_j$. Else If $R\setminus T'\neq \emptyset$, $u(R,p')=u(R,p)<0$, by what argued for prices $p$ and Property~\ref{it:payment:prop-3}. Last, for $R \subseteq T'$, by our choice of $T$, we have
    $$u(R,p')=v(R)-r(R,p')=\sum_{\ell \in R}q_\ell(p_\ell) - \frac{\epsilon}{|R|} - \sum_{\ell \in R} q_\ell(p'_\ell) = \sum_{\ell \in R}q_\ell(p_\ell) - \frac{\epsilon}{|R|} - \sum_{\ell \in R} q_\ell(p_\ell) < 0.$$
    The claim follows.

    \textbf{Case 2:} $r(T,p)>r(T',p)+r(\{j\},p)$. Let $\underline p$ be defined by $\underline p_j=0$ and
    $\underline p_\ell=p_\ell$ for all $\ell\neq j$, and set
    $$
        \delta:=r(T,p)-r(T',p)-r(\{j\},p)>0 .
    $$
    By the minimality of $T$, for every $R\subseteq T'$, $r(R,p)=\sum_{\ell\in R} q_\ell(p_\ell),$ and, for every $R\subsetneq T'$, $r(R\cup\{j\},p)=r(R,p)+r(\{j\},p).$ Also, by Property~\ref{it:payment:prop-2} and Property~\ref{it:payment:prop-3}, $r(R\cup\{j\},\underline p)=r(R,\underline p)=r(R,p)$ for all $R\subseteq T'$. Choose $\eta_\ell>0$ for $\ell\in T'$ such that $\eta:=\sum_{\ell\in T'}\eta_\ell<\delta$, and choose $C>\eta$. Define an additive valuation $v$ by
    $$
        v(\{\ell\})=q_\ell(p_\ell)+\eta_\ell \quad (\ell\in T'), 
        \qquad
        v(\{j\})=r(\{j\},p)+C,
    $$
    and set $v(\{\ell\})=-2\max_{R\subseteq S}|r(R, \underline p)|$ for $\ell\notin T$. Since $v$ is additive, it satisfies gross
    substitutes. At prices $\underline p$, for every $R\subseteq T'$,
    $$u(R\cup\{j\},\underline p)=v(R)+v(\{j\})-r(R\cup\{j\},\underline p)=r(\{j\},p)+C+\sum_{\ell\in R}\eta_\ell>v(T)-r(T,\underline p)=u(R,\underline p).$$
    As this utility is uniquely maximized at $R=T'$ and bundles or containing
    items outside $T$ are clearly not optimal, $T=T'\cup\{j\}$ is the unique demanded bundle at $\underline p$.
    
    Now consider again the original prices $p$ and recall that $p\ge \underline p$. For every $R\subsetneq T'$,
    $$u(R\cup\{j\},p)=v(R)+v(\{j\})-r(R\cup\{j\},p)=C+\sum_{\ell\in R}\eta_\ell ,$$
    whereas
    $$u(T,p)=v(T')+v(\{j\})-r(T,p)=C+\eta-\delta<C.$$
    Thus $T$ is not demanded at $p$. Moreover,
    $$u(T',p)=v(T')-r(T',p)=\eta<C=u(\{j\},p),$$
    so $T'$ is not demanded either. Since the only bundles containing $T'$ are $T'$ and $T$, no
    demanded bundle at $p$ contains $T'$, contradicting $r$-gross substitutes. This proves the contrapositive in
    Case 2.     
\end{proof}

\end{document}